\documentclass[letterpaper,11pt,prx,aps,superscriptaddress,showpacs,nofootinbib]{revtex4-2}

\usepackage[utf8]{inputenc}
\usepackage[colorlinks=true,linkcolor=blue,citecolor=blue,urlcolor=blue]{hyperref}
\usepackage{graphicx}
\usepackage{amsmath,amssymb}
\usepackage{dsfont}
\usepackage{xcolor}
\usepackage{amsthm}

\usepackage{algorithm}
\usepackage{algpseudocode}

\usepackage{subfig}
\usepackage{ragged2e}

\usepackage[mode=buildnew,subpreambles=true]{standalone}

\newtheorem{lemma}{Lemma}
\newtheorem{definition}{Definition}
\newtheorem{theorem}{Theorem}
\newtheorem{corollary}{Corollary}
\theoremstyle{remark}
\newtheorem{remark}{Remark}

\newcommand{\TC}{${\mathcal{Z}(\text{Vec}_{\mathbb{Z}_2})}$ }

\definecolor{FGreen}{RGB}{1,100,10}

\makeatletter
\def\l@subsection#1#2{}
\def\l@subsubsection#1#2{}
\makeatother

\begin{document}

\title{Layer codes}
\author{Dominic J. Williamson}
\author{Nou\'edyn Baspin}
\affiliation{Centre for Engineered Quantum Systems, School of Physics, University of Sydney, Sydney, NSW 2006, Australia}
\date{May 2023}

\begin{abstract}
\noindent
The surface code is a two-dimensional topological code with code parameters that scale optimally with the number of physical qubits, under the constraint of two-dimensional locality. 
In three spatial dimensions an analogous simple yet optimal code was not previously known. 
Here, we introduce a construction that takes as input a stabilizer code and produces as output a three-dimensional topological code with related code parameters.
The output codes have the special structure of being topological defect networks formed by layers of surface code joined along one-dimensional junctions, with a maximum stabilizer check weight of six. 
When the input is a family of good low-density parity-check codes, the output is a three-dimensional topological code with optimal scaling code parameters and a polynomial energy barrier.
\end{abstract}

\maketitle

\tableofcontents

\section{Introduction} 
\label{sec:Introduction}

Quantum computing promises to open a new window into highly-entangled quantum many-body physics~\cite{Feynman1982,manin1980computable,Benioff1980,Deutsch1985,Arute2019}.
Quantum error correction is a crucial ingredient in the design of scalable quantum computers that aim to access new regimes of physics beyond the reach of classical simulations~\cite{Shor1995,Steane1996,shor1996fault,gottesman1997stabilizer,aharonov1997fault,knill1998resilientQC,kitaev1997quantum,preskill1998reliable,Preskill1997fault}. 
Topological quantum codes form the basis of leading approaches to implement quantum error correction~\cite{kitaev2003fault,bravyi1998quantum,dennis2002topological,raussendorf2007fault,raussendorf2007topological,Postler2022,Krinner2022,Acharya2023}. 
This can be attributed to their favourable properties, including relatively simple stabilizer checks and high thresholds.
Topological codes form an important subclass of low-density parity-check (LDPC) codes~\cite{gottesman2014fault,tillich2014quantum,leverrier2015quantum,breuckmann2021ldpc}, set apart by the geometric locality of their checks.

The surface code is the simplest example of a topological code and as such is the most widely studied quantum error-correcting code both theoretically and experimentally~\cite{kitaev2003fault,bravyi1998quantum,dennis2002topological,raussendorf2007fault,raussendorf2007topological,Postler2022,Krinner2022,Acharya2023}. 
It is known by several names including the toric code, the planar code, $\mathbb{Z}_2$ lattice gauge theory, the $\mathbb{Z}_2$ quantum double, and $\mathcal{Z}(\text{Vec}_{\mathbb{Z}_2})$~\cite{kitaev2003fault,bravyi1998quantum,Wegner1971,Kogut1979,DijkgraafWitten,kitaev2006anyons}.
Despite its popularity and positive features, the surface code is far from attaining optimal scaling of its code parameters $[[n,k,d]]$. 
Here $n$ is the number of physical qubits, $k$ is the number of encoded qubits, and $d$ is the code distance, i.e.~the minimum weight of a nontrivial logical operator. 
The surface code is a $[[2L(L-1),1,L]]$ code and hence has a vanishing rate $\frac{k}{n}$ and suboptimal distance scaling when compared to general LDPC stabilizer codes which can achieve a constant $\frac{d}{n}$ scaling. 
In a recent flurry of work, \textit{good} LDPC codes were discovered~\cite{hastings2020fiber,panteleev2020quantum,breuckmann2020balanced,Panteleev2022,leverrier2022quantum,Dinur2023} that attain code parameters with optimal scaling i.e.~$[[n,\Theta(n),\Theta(n)]]$. 

Topological codes defined on a regular lattice in $D$-dimensional Euclidean (flat) space with a finite density of qudits can never attain optimal scaling of the code parameters. 
This is because they are constrained by several bounds, including the Bravyi-Poulin-Terhal (BPT) bound~\cite{bravyi2010tradeoffs}
\begin{align}
    kd^{\frac{2}{D-1}} \leq O(n), 
\end{align}
the Bravyi-Terhal (BT) bound~\cite{bravyi2009no}
\begin{align}
    d \leq O(n^{\frac{D-1}{D}}),
\end{align}
and Haah's bound~\cite{Haah2021}
\begin{align}
    k \leq O(n^{\frac{D-2}{D}}). 
\end{align}

In light of these bounds, the best scaling of code parameters for a topological code in $D$ dimensions one can hope to achieve is $[[n,\Theta(n^{\frac{D-2}{D}}),\Theta(n^{\frac{D-1}{D}})]]$. Writing these code parameters in terms of the linear extent of a hypercuboid lattice $L$ we find $[[\Theta(L^D),\Theta(L^{D-2}),\Theta(L^{D-1})]]$. 
The most physically relevant cases are $D=2,3$.
It is apparent from the parameters given above that the surface code achieves optimal scaling for $D=2$. 
Here, we introduce families of \textit{layer codes} that achieve optimal scaling for $D=3$.

\subsection{Overview of Main Results} 

In this work we introduce a construction that takes as input an arbitrary Calderbank-Shor-Steane~\cite{calderbank1996good,Steane1996Simple} (CSS) stabilizer code and outputs a topological CSS code that is local in dimension $D=3$ which we call a \textit{layer code}. 
The code parameters of the input and output codes of our construction are related by
\begin{align}
\label{eq:CodeParameterMap}
    [[n,k,d]] \mapsto [[\Theta(nn_Xn_Z),k,\Omega(\frac{d}{w} n_{*})]],
\end{align}
where $n_X$ is the number of $X$ checks, $n_Z$ is the number of $Z$ checks, $n_{*}= \min(n_X,n_Z)$, and $w$ is the maximum weight of the checks in the input code. Furthermore, the maximum check weight of the output code is $6$. 

The layer codes output by our construction have additional structure, taking the form of topological defect networks~\cite{Slagle2018,Aasen2020,Song2021}. 
More specifically, each layer code is constructed from layers of surface code with one layer for each physical qubit, $X$ check, and $Z$ check of the input code.
The surface code layers associated to physical qubits are stacked on $xz$-planes of a cubic lattice, the layers associated to $X$ checks are stacked on $xy$-planes, and the layers associated to $Z$ checks are stacked on $yz$-planes. 
The grid of surface code layers are then coupled together at nontrivial junctions that are determined by the Tanner graph of the input code. 

When applied to a family of good CSS LDPC codes our construction outputs a family of topological CSS codes that achieve optimal scaling of the code parameters for $D=3$, that is
\begin{align}
    [[n,\Theta(n),\Theta(n)]] \mapsto [[\Theta(n^3),\Theta(n),\Theta(n^2)]].
\end{align}
These code parameters achieve the point `a' depicted in Figure \ref{fig:BoundsOnCodes}. 
Furthermore, the range of code parameters on the line between `a' and $(0,1)$ in Figure \ref{fig:BoundsOnCodes} can be achieved by tiling layer code blocks. 
To demonstrate this we parametrize the `a' to $(0,1)$ line by $\alpha \in [0,1]$ and consider $N^{1-\alpha}$ layer code blocks of size $N^\alpha$ to obtain codes with parameters $[[N,\Omega(N^{1-\alpha}\cdot N^{\alpha/3}), \Omega(N^{2\alpha/3})]]$. 
This construction saturates the BPT bound for all attainable values of $d$. 
We remark that a union of disjoint layer code blocks ceases to be topological in the sense required for Haah's bound to apply. 
Nonetheless, each single block of layer code obeys Haah's bound and can saturate it.
The saturation of the above bounds demonstrates that topological codes are, in a sense, optimal local codes in 3D. 

\begin{figure}[t]
    \includegraphics[page=12]{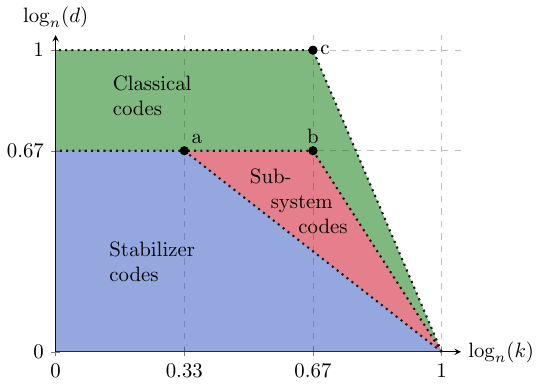}
    \caption{
    \justifying
    An illustration of existing bounds on local codes in 3D, and constructions known to saturate them. The blue shaded area corresponds to the BPT bound~\cite{bravyi2010tradeoffs} ${kd \leq O(n)}$, and the BT bound~\cite{bravyi2009no} $d \leq O(n^{2/3})$. The point labeled `a' indicates the layer code construction, see Theorem~\ref{thm:1}. 
    The red shaded area corresponds to Bravyi's bounds for subsystem codes~\cite{bravyi2011subsystem}. The point `b' indicates the construction in Ref.~\cite{bacon2015sparse}, see also Ref.~\cite{WireCodes}. The green shaded area corresponds to the bound on classical codes from Ref.~\cite{bravyi2010tradeoffs}. The point `c' indicates the construction in Ref.~\cite{baspin2023combinatorial}.
    All points along the dashed line from `a' to $(1,0)$ can be achieved by tiling smaller layer code blocks, see Corollary~\ref{corr:1}, and similarly for points `b' and `c'. 
    }
    \label{fig:BoundsOnCodes}
\end{figure}

Finally, we demonstrate that the layer code construction preserves the scaling of the energy barrier when the input code is LDPC. 
We show that a family of good CSS LDPC codes from Ref.~\cite{leverrier2022quantum} has energy barrier $\Theta(n)$. This leads to our main result
\begin{theorem}
\label{thm:1}
    There exist families of topological CSS stabilizer codes in $D=3$ spatial dimensions that achieve the optimal code parameters $[[\Theta(L^3),\Theta(L),\Theta(L^2)]]$ and a $\Theta(L)$ energy barrier. In particular, layer codes based on the good LDPC codes from Ref.~\cite{leverrier2022quantum} achieve these properties and have checks of weight 6 or less.
\end{theorem}
\begin{proof}
    The proof is divided into a high level sketch of the code properties in Section~\ref{sec:CodeProperties}, supported by technical lemmas in Section~\ref{sec:Proofs}, and the explicit construction of lattice checks in Section~\ref{sec:LatticeModel}. 
    
    In Section~\ref{sec:CodePropertiesParameters} we explain how the code parameters of the layer codes are related to those of the input codes. 
    This is summarized in Eq.~\eqref{eq:CodeParameterMap}. 
    In Section~\ref{sec:CodePropertiesEnergy} we explain how the scaling of the energy barrier for a layer code is related to the energy barrier of the input code. 
    In Section~\ref{sec:ProofsLogicals} we present technical lemmas characterizing the form of the logical operators in the layer codes.
    We rely upon this characterization to establish the scaling of the distance, number of encoded qubits and energy barrier. 
    In Section~\ref{sec:ProofsLZBarrier} we show that the family of good CSS LDPC codes from Ref.~\cite{leverrier2022quantum} has energy barrier $\Theta(n)$. 
    In Section~\ref{sec:LatticeModel} we introduce a lattice model for the layer codes involving local stabilizer generators of weight at most 6. 
    Together, these results demonstrate that applying the layer code construction to the codes from Ref.~\cite{leverrier2022quantum} produces families of layer codes with the claimed properties. 
\end{proof}
\begin{corollary}
\label{corr:1}
    There exist families of local CSS stabilizer codes in D=3 spatial dimensions that achieve code parameters $[[N,\Omega(N^{1-\alpha}\cdot N^{\alpha/3}), \Omega(N^{2\alpha/3})]]$ for $\alpha \in [0,1]$.
\end{corollary}
\begin{proof}
    Consider $N^{1-\alpha}$ blocks of size $N^\alpha$ each containing a layer code based on the good LDPC codes from Ref.~\cite{leverrier2022quantum}. 
    The layer code in each block has code parameters $[[N^\alpha,\Omega(N^{\alpha/3}), \Omega(N^{2\alpha/3})]]$ and checks of weight 6 or less.
    Since there are $N^{1-\alpha}$ such code blocks, the stated result follows.
\end{proof}
The layer code construction opens the door to a vast, unexplored landscape of promising non translation-invariant topological codes with striking features. Table~\ref{tab:summary} provides a summary of the input and output codes involved in the construction. Algorithm~\ref{alg:layer-codes-alg} provides a summary of the construction procedure.

\setlength\tabcolsep{1.5cm}
\begin{table}[t]
    \centering
    \begin{tabular}{ c c}
    Input code                  &   Output layer code \\
    \hline \hline
    General CSS stabilizer code      &   3D topological defect network code \\
    Physical qubit              &   $xz$ surface code layer \\
    $X$ check                   &   $xy$ surface code layer \\
    $Z$ check                   &   $yz$ surface code layer \\
    Max check weight $w$       &   Max check weight 6 \\
    Energy barrier $\Delta(n)$  &   Energy barrier $\frac{1}{w w'} \Delta(n)$ \\
    $[[n,k,d]]$                 &   $[[\Theta(nn_Xn_Z),k,\Omega(\frac{1}{w}d \min(n_X,n_Z))]]$
\end{tabular}
    \caption{Summary of the input and output code properties in the layer code construction.}
    \label{tab:summary}
\end{table}

\subsection{Prior Work} 

Quantum codes in three dimensions have received significant attention over the past two decades~\cite{dennis2002topological,Bombin2007,Bombin2015,Brown2016,Castelnovo2007,Castelnovo2008,yoshida2011classification,Yoshida2011Feasibility,chamon2005quantum,haah2011fractal,walker2012,kim20123d,yoshida2013exotic,brell2014proposal,PhysRevB.92.235136,Vijay2016,Williamson2016,bravyi2011topological,michnicki2012quantum,PhysRevLett.107.150504,Bravyi2013,williamson2016hamiltonian,Brown2019Parallel,Weinstein2019,Devakul2020b,Zhu2022,aitchison2023no,Dua2023Fractal}.  
A stack of $L$ two-dimensional surface codes of size $L\times L$ demonstrates the possibility of three-dimensional topological codes with a number of encoded qubits that scales as $\Theta(L)$. 
The three-dimensional toric code~\cite{dennis2002topological} established that planelike logical operators with minimum weight $\Omega(L^2)$ are possible in three-dimensional topological codes. 
However the distance of the three-dimensional toric code scales as $\Theta(L)$ as it also supports a stringlike logical operator. 
This further implies that the three-dimensional toric code has a constant logical energy barrier and hence is not a self-correcting quantum memory~\cite{Castelnovo2007,Castelnovo2008}. 
There have been numerous efforts searching for codes beyond the toric code that have no stringlike logical operators and consequently have superlinear distances and logical energy barriers that increase with system size. 
Research in this area has resulted in the discovery of a number of fracton codes~\cite{chamon2005quantum,haah2011fractal,kim20123d,yoshida2013exotic,PhysRevB.92.235136,Vijay2016}.  
One particularly remarkable fracton code is Haah's cubic code which has no stringlike logical operators, a superlinear distance, a logarithmic energy barrier, a number of encoded qubits that scales as $\Theta(L)$ for a certain family of system sizes, and which exhibits partial self correction~\cite{haah2011fractal,PhysRevLett.107.150504,Bravyi2013}. 
Despite these promising characteristics, neither the cubic code nor any other of the so-called type-II fracton models that satisfy the no stringlike logical operator rule have been shown to achieve optimal $\Theta(L^2)$ scaling of the code distance. 
Codes with polynomially growing energy barriers have been found by abandoning translation invariance~\cite{michnicki2012quantum}.
However, these codes also failed to achieve optimal scaling of the code distance. 
The layer codes output by our construction can be viewed as a new variety of non translation-invariant fracton codes which can achieve the desired optimal code parameters $[[\Theta(L^3),\Theta(L^{1}),\Theta(L^{2})]]$ and a logical energy barrier that scales as $\Theta(L)$.

A large amount of research has been done on new constructions of fracton models and attempts to characterize and classify them~\cite{PhysRevB.92.235136,haah2014bifurcation,Vijay2016,Williamson2016,vijay2017generalization,PhysRevB.95.245126,Dua2019,Dua2019c,Pai2019,Tantivasadakarn2019,Tantivasadakarn2021,Tantivasadakarn2021Nonabelian,Song2023,Bulmash2018,song2018twisted,Prem2019,Prem2019Cage,Bulmash2019,Dua2019b,Williamson2020,Sullivan2021}. 
This led to the discovery that networks of two and three-dimensional topological codes joined together by topological defects can recover all known fracton models (with finite order excitations)~\cite{Aasen2020,Song2021}. 
The layer codes introduced in this work are local topological codes in three spatial dimensions that take the form of topological defect networks. 
However, unlike previous topological defect network constructions, they are not translation invariant even at a coarse grained scale. 
The construction of layer code topological defect networks in this work is different to the topological defect network code construction introduced in Ref.~\cite{Song2021}. 
For contrast, the earlier construction takes as input a topological CSS code that is already local in a fixed spatial dimension, and produces as output a new topological CSS code that is local in the same spatial dimension but has the special form of a topological defect network. 
The layer codes in this work are also related to the gauged layer constructions in Refs.~\cite{Williamson2020,Sullivan2021}.
The models in both the present and previous works can be understood as resulting from gauging planar symmetries formed by the product of stringlike logical operators over layers of surface code stacked along one direction. 
However in the previous works the products were taken over all layers, whereas in this work the layers involved in each product are determined by the checks of an input CSS code. 

\textit{Note added} --- While this project was in progress~\cite{Williamson2023talk}, a manuscript proving a related result appeared~\cite{portnoy2023local} which also builds upon the recent breakthroughs establishing good LDPC codes~\cite{Panteleev2022,breuckmann2021ldpc,leverrier2022quantum,Dinur2023} to construct new topological codes in $D\geq 3$.
Our construction was formulated and worked out independently of Ref.~\cite{portnoy2023local} and differs on several key points. 
First, our work makes no use of the main techniques employed in Ref.~\cite{portnoy2023local}, including the construction due to Freedman and Hastings~\cite{Freedman2021} and a quantitative embedding theorem due to Gromov and Guth~\cite{Gromov2012}. 
The distinction is further accentuated by the special structure inherent to the layer codes output by our construction, they are topological defect networks made up of surface code layers and defects with maximum check weight 6, which does not appear to hold for the codes in Ref.~\cite{portnoy2023local}.
Finally, we remark that while our construction is only defined for $D=3$ it saturates the BPT bound without any polylog factor, whereas the result in Ref.~\cite{portnoy2023local} applies to arbitrary finite $D$ but only achieves optimal code parameter scaling up to a polylog factor. 
The precise relationship between the construction we present here and that of Ref.~\cite{portnoy2023local} is an interesting question we leave to future work.
Finally, we remark that another related work Ref.~\cite{Lin2023} appeared concurrently with this work. 

\subsection{Open questions}
\label{sec:openq}

The layer code construction raises a number of questions which we review below. 

\textbf{Self-correction:} There are families of layer codes that have optimal energy barriers, and hence it is natural to ask whether any layer code family is a self-correcting quantum memory. 
Despite the promising energy barrier, it appears to be challenging to answer this question. 
The standard method to show self-correction following Ref.~\cite{alicki2010thermal}, Ref.~\cite{Bombin2013Self} (Section VI), and Ref.~\cite{roberts2020symmetry} (Appendix D), relies on bounding the number of syndrome configurations with low energy. 
However, in the layer codes there are at least $n^{\Omega (n)}$ configurations with energy $\gamma n$, for some constant $\gamma >0$.\footnote{This is because one can always find a quasi-concatenated logical supported on $d$ $xz$-layers, where $d$ is the distance of the input code.
Without loss of generality we consider an $X$-type logical. 
By truncating the support of this logical on the $xz$-layers, we can obtain an operator with energy at most $2d$. 
Since the pointlike excitations created by the truncated logical are all mobile along the $x$ direction, there are at least $L^d$ configurations with the same energy where $L$ is the length of $xz$ surface code layer.
For good LDPC codes $L\sim d \sim n$ and hence the above results in an $n^{\Omega (n)}$ degeneracy for configurations of energy $\sim n$.} 
For these parameters the upper bound from Equation (51) in Ref.~\cite{alicki2010thermal} gives $e^{- n (\beta \gamma - \log(n)) } = \omega(1)$, which fails to prove self-correction. This points to the need for alternative methods for bounding the memory time. 
We leave the resolution of the memory time question to future work.

\textbf{Decoders:} We have not addressed the important problem of decoding layer codes in this work. 
We remark that the general renormalization group decoder~\cite{Bravyi2013} is applicable to layer codes, as they satisfy the local topological order condition in three dimensions. 
However, the structure of the layer codes suggests that more specialized decoders should exist. 
We have also not addressed the problem of performing fault-tolerant logical gates on the layer codes. 
The surface code layer structure of the layer codes presents a great opportunity to leverage well-developed surface code techniques for performing such gates.
We leave these directions to future work.

\textbf{Product constructions:} The layer code construction we have introduced is evocative of product constructions of quantum codes~\cite{tillich2014quantum, breuckmann2020balanced}. 
More explicitly, the logicals of a layer code closely resemble those obtained from concatenating the input LDPC code with a surface code. 
For the special case of a classical input code the resulting layer code \textit{is} a product of the input code with a repetition code, see Figure~\ref{fig:zzDefectHGP}. 
Even beyond that special case, topological defects can equivalently be understood in terms of the stabilizers and logicals they induce as shown in Figure \ref{fig:zzDefectNonlocalExamples}. 
This raises the question, can one find a combinatorial interpretation of layer codes? A combinatorial point of view is likely to produce complementary insights on the construction that could help with generalizing it to different target spaces.

\textbf{Higher dimensions:} An interesting challenge is the extension of our construction to higher dimensions in a way that produces output codes that saturate the BPT bound in all dimensions. 
Similarly, it would be interesting to generalize our construction to produce optimal codes in spaces with constant negative curvature, rather than flat space. 
We anticipate that the techniques of Ref.~\cite{baspin2023combinatorial} could prove useful for this line of research. 
It would also be interesting to generalize the construction of layer codes to allow arbitrary nonlocal, but low weight, connections between the surface codes and to explore the tradeoff in code parameters and properties that this allows. For example instead of using data layers of size $n_X \times n_Z$, one could attempt to use data layers of size $\delta_X \times \delta_Z$, where $\delta_X$ is the maximum $X$ degree of the code, to obtain codes of constant degree and weight with parameters $[[n \delta_X \delta_Z, k, \frac{d}{\delta_X \delta_Z}]]$. If this was possible it would yield an alternative sparsification theorem for CSS LDPC codes~\cite{hastings2021quantum,Sabo2024}.

\textbf{Different topological codes:} 
It is interesting to ask how replacing the surface code with an alternative topological code might affect the construction.  
First, it would be interesting to find a floquet implementation of the layer codes that reduces the check weight to 2.
It would also be interesting to extend our construction to color code layers via folding~\cite{kubica2015unfolding} in a way that preserves the transversal gates of the input code, such as the Steane code example. One could imagine trying to preserve a transversal $T$ gate by using the 3D color code~\cite{bombin2007topological} in lieu of 2D topological codes, as no 2D code possesses such a gate~\cite{bravyi2013classification}.
Another question is about the generalization of our construction to non-CSS codes, and codes beyond the Pauli stabilizer formalism. 
Can our construction be leveraged to build translation invariant non-stabilizer\footnote{It is known that translationally invariant \emph{stabilizer} codes in $3$ dimensions have energy barrier at most $O(\log(n))$~\cite{haah2013commuting} and therefore cannot be self correcting~\cite{temme2016thermalization}.} topological codes in three dimensions with optimal code parameters?

\textbf{Equivalence relation:} We remark that our construction depends on the choice of stabilizer generators, or checks, of the input code.
This raises the question of how the output codes produced from the same input code with different choices of generators are related. 
In particular, are output codes with different choices of input generators in the same topological phase of matter? 
That is, are they equivalent up to a generalized local unitary in three dimensions?
Similarly, we have not addressed the question of how the layer code output by our construction is related to the three-dimensional code output by the construction of Ref.~\cite{portnoy2023local} for the same input LDPC code.

\subsection{Examples} 

We now illustrate our construction by depicting the layer codes that are output given several simple examples of input CSS codes. 

\subsubsection{Repetition Code}

The first example is based on the 3 qubit repitition code, which has stabilizer checks $ZZI,IZZ$. 
The layer code derived from the 3 qubit repetition code involves 5 surface code layers in total. 
There are 3 $xz$ surface code layers corresponding to the input physical qubits, and 2 $yz$ surface code layers corresponding to the input code $Z$ checks. 
Trijunction topological defects join the check layers to the qubit layers as shown in Figure~\ref{fig:RepetitionCode} below, the details of these defects are explained in Section~\ref{sec:LayerCodeConstruction}.
\begin{figure}[H]
    \centering
    \includegraphics[page=5]{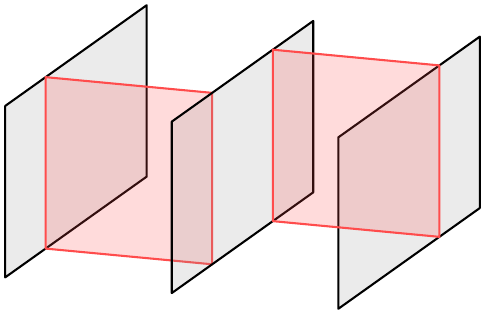}
    \caption{The layer code based on the 3 qubit repetition code. 
    3 grey $xz$-layers depict surface codes corresponding to input physical qubits. 
    2 red $yz$-layers depict surface codes corresponding to input $Z$ checks. 
    Trijunctions between the red and grey layers correspond to nontrivial topological defects. }
    \label{fig:RepetitionCode}
\end{figure}

\subsubsection{[[4,2,2]] error detecting code}

The [[4,2,2]] error detecting code has the following stabilizer checks $XXXX,ZZZZ$. 
The layer code based on the  [[4,2,2]] code involves 6 surface code layers in total, 4 $xz$-layers corresponding to encoded qubits, 1 $xy$-layer corresponding to the $X$ check, and 1 $yz$-layer corresponding to the $Z$ check. 
The surface code layers are joined together at their junctions by topological defects as shown in Figure~\ref{fig:422Code}, see Section~\ref{sec:LayerCodeConstruction} for a detailed description of the defects. 
\begin{figure}[H]
    \centering
    \includegraphics[page=61]{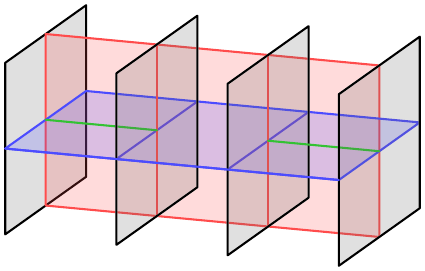}
    \caption{The layer code based on the [[4,2,2]] code. 
    4 grey $xz$-layers depict surface codes corresponding to input physical qubits. 
    The blue $xy$-layer depicts a surface code corresponding to an input $X$ check. 
    The red $yz$-layer depicts a surface code corresponding to an input $Z$ check. 
    Blue, red and green junctions where the layers meet correspond to nontrivial topological defects. }
    \label{fig:422Code}
\end{figure}

\subsubsection{Shor's Code}

Shor's code has the following stabilizer checks~\cite{Shor1995}
\begin{align}
    \begin{array}{c c c c c c c c c }
        X & X & X & X & X & X & I & I & I  \\
        I & I & I & X & X & X & X & X & X  \\
        Z & Z & I & I & I & I & I & I & I  \\
        I & Z & Z & I & I & I & I & I & I  \\
        I & I & I & Z & Z & I & I & I & I  \\
        I & I & I & I & Z & Z & I & I & I  \\
        I & I & I & I & I & I & Z & Z & I  \\
        I & I & I & I & I & I & I & Z & Z .  
    \end{array}
\end{align}
The layer code based on Shor's code involves 17 surface code layers in total, 9 $xz$-layers corresponding to encoded qubits, 2 $xy$-layers corresponding to $X$ checks, and 6 $yz$-layers corresponding to $Z$ checks. 
These surface code layers are joined together at their junctions by defects that follow the incidence relations of the qubits and checks in Shor's code, see Figure~\ref{fig:ShorsCode}. The defects are described in Section~\ref{sec:LayerCodeConstruction}. 

\begin{figure}[H]
    \centering
    \includegraphics[scale = 1, page=1]{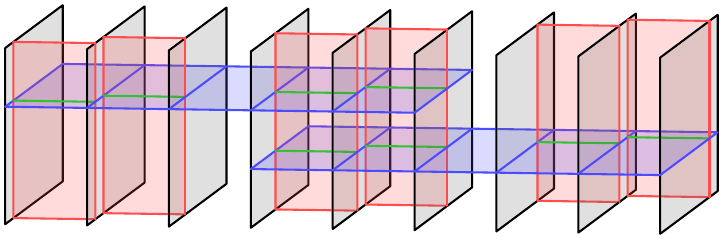}
    \caption{The layer code based on Shor's code. 
    9 grey $xz$-layers depict surface codes corresponding to input physical qubits. 
    2 blue $xy$-layers depict surface codes corresponding to input $X$ checks. 
    6 red $yz$-layers depict surface codes corresponding to input $Z$ checks. 
    Blue, red and green junction lines where the layers meet correspond to nontrivial topological defects. }
    \label{fig:ShorsCode}
\end{figure}

\subsubsection{Steane's Code}

Steane's code has the following stabilizer checks~\cite{steane1996multiple}
\begin{align}
    \begin{array}{c c c c c c c }
        X & I & X & I & X & I & X   \\
        I & X & X & I & I & X & X   \\
        I & I & I & X & X & X & X   \\
        Z & I & Z & I & Z & I & Z   \\
        I & Z & Z & I & I & Z & Z   \\
        I & I & I & Z & Z & Z & Z  .  
    \end{array}  
\end{align}
The layer code based on Steane's code involves 13 surface code layers in total, 7 $xy$-layers corresponding to encoded qubits, 3 $xy$-layers corresponding to $X$ checks, and 3 $yz$-layers corresponding to $Z$ checks. 
The surface code layers are joined together at their junctions by topological defects as described in Section~\ref{sec:LayerCodeConstruction}. 
The resulting layer code is depicted in Figure~\ref{fig:SteanesCode}
\begin{figure}[H]
    \centering
    \includegraphics[scale = 1, page=2]{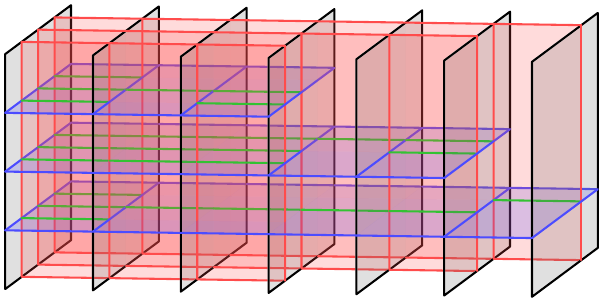}
    \caption{The layer code based on Steane's code. 
    7 grey $xz$-layers depict surface codes corresponding to input physical qubits. 
    3 blue $xy$-layers depict surface codes corresponding to input $X$ checks. 
    3 red $yz$-layers depict surface codes corresponding to input $Z$ checks. 
    Blue, red and green junction lines where the layers meet correspond to nontrivial topological defects.}
    \label{fig:SteanesCode}
\end{figure}

\subsection{Section Outline} 

In Section~\ref{sec:Background} we review relevant background material. 
In Section~\ref{sec:LayerCodeConstruction} we explain the layer code construction. 
In Section~\ref{sec:CodeProperties} we describe the properties of the layer codes. 
In Section~\ref{sec:Proofs} we provide proofs of our main results. 
In Section~\ref{sec:LatticeModel} we describe explicit stabilizer checks to implement layer codes on the lattice. 
In Section~\ref{sec:Discussion} we discuss our results.

\section{Background} 
\label{sec:Background}

In this section we summarize relevant background material that we build upon in the layer code construction. 

\subsection{Quantum Codes} 

A $n$-qubit quantum code $\mathcal{C}$ is a subspace of an $n$-qubit Hilbert space $(\mathbb{C}^2)^{\otimes n}$. We say that $\mathcal{C}$ encodes $k$ logical qubits if $\mathcal{C} \cong (\mathbb{C}^2)^{\otimes k}$. In this work we focus on stabilizer codes, where $\mathcal{C}$ is the common $+1$ eigenspace of $\mathcal{S}$, an abelian subgroup of the $n$-qubit Pauli group. Further, a code is ``CSS" if $\mathcal{S} = \mathcal{S}^X \cdot \mathcal{S}^Z$, where $\mathcal{S}^X$ and $\mathcal{S}^Z$ are generated, respectively, by $X$-type and $Z$-type Pauli operators. The normalizer $\mathcal{N}$ is the group of $n$-qubit Pauli operators that commute with $\mathcal{S}$. Finally, the distance $d$ of a stabilizer code is the cardinality of the support of the smallest operator $L \in \mathcal{N}$ such that $L$ commutes with all elements in $\mathcal{S}$ yet is not in $\mathcal{S}$. We sometimes refer to the $X$ (or $Z$) distance $d_X$ (resp. $d_Z$) of a CSS code, which corresponds to the cardinality of the support of the smallest $X$-type (resp. $Z$-type) Pauli operator that commutes with $\mathcal{S}^Z$ (resp. $\mathcal{S}^X$) and is not in $\mathcal{S}^X$ (resp. $\mathcal{S}^Z$). In that case we have $d = \min (d_X, d_Z)$.

\subsection{Surface Code} 

The surface code refers to a family of $n= \Theta( L_x \times L_z )  $ qubit CSS stabilizer codes~\cite{kitaev1997quantum}. 
The qubits are associated to edges of an $L_x \times L_z$ patch of the square lattice with two \textit{smooth} boundaries at $x=0,L_x$ and two \textit{rough} boundaries at $z=0,L_z$, see Figure~\ref{fig:SurfaceCodeIntro}. 
The stabilizer group is generated by star and plaquette checks
\begin{align}
    \mathcal{S} = \big\langle A_v, B_p \ |\  A_v = \prod_{v \ni e} X_e , B_p = \prod_{e \in p} Z_p \big\rangle. 
    \label{eq:SurfaceCodeGenerators}
\end{align}
The stabilizer checks in the bulk have weight 4, while the boundary stabilizer checks have weight 3. 
Note vertices along the smooth boundary, and plaquettes along the rough boundary are included in the patch of square lattice. 
Counting the number of qubits and independent stabilizer checks reveals there is a single encoded qubit. 
Logical representatives are given by a horizontal $X$ string operator and a vertical $Z$ string operator
\begin{align}
    \overline{X} = \prod_{e\in \hat{\gamma}_x} X_e, &&
    \overline{Z} = \prod_{e\in \gamma_z} Z_e ,
\end{align}
where $\hat{\gamma}_x$ is a path through the dual lattice running form $x=0$ to $x=L_x$, and $\gamma_z$ is a path through the lattice running from $z=0$ to $z=L_z$.

\begin{figure}[H]
    \centering
    \includegraphics[page=10]{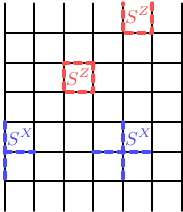}
    \qquad \qquad 
    \includegraphics[page=11]{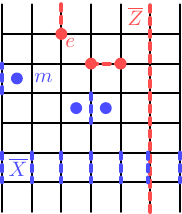}
    \caption{A patch of surface code with examples of $X$ checks (blue) and $Z$ checks is depicted on the left. 
    On the right another patch of surface code is depicted with $X$ logicals, errors, and syndromes (blue) and similar for $Z$ (red). }
    \label{fig:SurfaceCodeIntro}
\end{figure}

\subsubsection{Anyons and Superselection Sectors}

The surface code can be formulated as the ground space of a local Hamiltonian in two dimensions
\begin{align}
    H=-\sum_v A_v -\sum_p B_p. 
\end{align}
The surface code Hamiltonian has a property known as \text{topological order}, which requires that the projection of any local operator onto the degenerate ground space is proportional to the identity~\cite{Bravyi2010Stability}. 
This is similar to the Knill-Laflamme condition for the detection and correction of local errors~\cite{knill1997theory}. 
While all local operators act trivially within the ground space, local excitations can be nontrivial. 
These local excitations correspond to syndromes of the surface code, i.e. collections of stabilizer generators with eigenvalue $-1$, within a local region of the lattice. 
Following Ref.~\cite{kitaev2003fault} we call an $A_v=-1$ syndrome an $e$ anyon, which can be created by a $Z$ string operator, while we call a $B_p=-1$ syndrome an $m$ anyon, which can be created by an $X$ string operator. 

Local excitations in the surface code are organized into equivalence classes known as \textit{superselection sectors}~\cite{kitaev2006anyons}. 
A pair of local excitations $E_1,E_2,$ is considered equivalent if $E_1$ can be converted to $E_2$ via the application of a local operator, and vice versa. 
In the surface code there are four superselections sectors known as the vacuum, electric charge, magnetic flux, and dyon, which we label $\{1,e,m,\psi\}$ respectively, where $\psi=e \times m$. 
Since $e$ and $m$ anyons are alwas created in pairs by local operators contained within the bulk of the surface code, the trivial sector $1$ corresponds to an even number of both $e$ and $m$ anyons. 
The superselection sector $e$ collects all excitations with an odd number of $e$ anyons, while $m$ is the sector corresponding to an odd number of $m$ anyons. 
A pair of superselection sectors can be \textit{fused} by considering representative local excitations for each on adjacent regions $A$ and $B$ and then taking equivalence classes under local operators on a larger region $C$ that contains both $A$ and $B$. 
The \textit{fusion} rules of the surface code superselection sectors are given by $\mathbb{Z}_2^2$. 
A pair of superselection sectors can also be \textit{braided} by moving one around the other using string operators. 
The nontrivial braiding properties of the surface code anyons are generated by a $-1$ braiding process between $e$ and $m$ due to the anticommutation of the relevant $X$ and $Z$ string operators. 
The braiding statistics of abelian anyons, such as those of the surface code, can be extracted form the \textit{exchange} statistics of the particles. 
It is possible to extract these from the lattice following the procedure in Refs.~\cite{Levin2003,Ellison2022}. 
In the surface code, the only nontrivial exchange statistic is a $-1$ associated to the emergent fermion $\psi$. 
The fusion and braiding structure for the surface code anyons is represented abstractly by a modular tensor category~\cite{kitaev2006anyons,Moore1989} that is denoted $\mathcal{Z}(\text{Vec}_{\mathbb{Z}_2})$. 

\subsubsection{Boundaries, Defects and Condensation}

A \textit{defect}, is a modified region of a topological code where different stabilizers or Hamiltonian terms are introduced -- for example the rough and smooth boundaries in the surface code have terms that differ from the bulk.
Because of these alterations, the creation and destruction of $e$ and $m$ anyons follow new rules that determine the structure of the logical operators of the new code. In this work we rely heavily on this fact to characterize the codes we build and their properties. 

\begin{figure}[H]
    \centering
    \includegraphics[page=8]{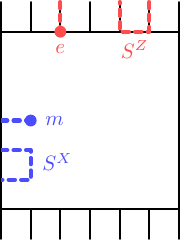}
    \caption{Operators and anyons condensing at the rough and smooth boundary of surface code. In the bulk of the code it is impossible to create a \emph{single} $e$ or $m$ anyon, while this becomes possible at the appropriate boundaries. We say that the rough boundary condenses the $e$ anyon, while the smooth boundary condenses the $m$ anyon. These condensation rules dictate the structure of the logical operators of the surface code. For example they allow for the creation of a single $e$-anyon at the top rough boundary which can be ``dragged down" and condensed at the bottom boundary, creating the logical $\overline{Z}$.}
    \label{fig:SurfaceCodeCondensation}
\end{figure}

A defect is \textit{topological} if its properties, such as which anyons condense there, do not depend on the microscopic realization  of the defect such as its precise shape. 
In this sense it is often said that topological defects are deformable. 
A defect is said to be \textit{gapped} if the topological order condition is still satisfied after the defect is introduced, i.e. the defect does not introduce constant sized logical operators in its surroundings. 
An anyon is said to \textit{condense} on a defect if it can be created (and annihilated) by a local operator on a region containing the defect. 
We remark that in the absence of a defect, anyons in nontrivial superselection sectors cannot be created by local operators, for example in the bulk of the surface code, anyons are \textit{always} created in pairs. 
However, in the vicinity of a defect the condensation of anyons in nontrivial superselection sectors becomes possible. 
That is, near a defect some anyons that would otherwise be nontrivial can be created or annihilated by the application of local operators, see Figure~\ref{fig:SurfaceCodeCondensation}. 
Anyon condensation must satisfy certain consistency conditions~\cite{Bais2009Condensate,Kong2014Anyon}. 
In an abelian anyon theory such as the surface code the condensing anyons form a subgroup of all anyons under fusion~\cite{Levin2013Protected,Barkeshli2013Classification}.

\subsubsection{From Defects to Stabilizers}

As previously mentioned, in this work we make use of defects that condense specific anyons to find new codes with properties that can differ from those of the standard surface code. However, a defect is only well defined if one can find Hamiltonian terms realizing its condensation rules. This is not always possible, and requires certain consistency conditions to be obeyed. In this section, we present a heuristic picture that can be used to find lattice terms that implement a desired defect. Regarding the specific case of the layer codes, the proof that these terms exist and commute -- therefore are stabilizers --  is given by their enumeration in Section \ref{sec:LatticeModel}.

A fundamental type of line defect is known as a gapped boundary~\cite{beigi2011quantum,Kitaev2012Models}. A gapped boundary, is simply a gapped line defect connected to only one topological phase, see Figure \ref{fig:FoldingTrick}: compare a general line defect (in green, left), to a boundary (green, right).
At a gapped boundary there are additional consistency constraints on the anyons that can condense there~\cite{Bais2009Condensate,Levin2013Protected,Kong2014Anyon}. 
The anyons that condense at a gapped boundary must braid trivially with one another, and each anyon must have trivial exchange statistics\footnote{
In the presence of physical fermions this condition is loosened to allow condensing anyons to have fermionic exchange statistics~\cite{aasen2017fermion}.}, 
see Figure \ref{fig:NeedTrivialBraiding}. 

\begin{figure}[H]
    \centering
    \includegraphics[page=9]{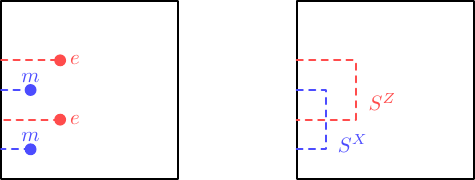}
    \caption{In this example we assume that the left boundary of the surface code patch can condense anyons with nontrivial braiding statistics, namely $e$ and $m$ anyons. We can then find operators $S_X, S_Z$ that both preserve the ground space: they result in no excitations, i.e. they are in the normalizer of the stabilizer group. However, because $e$ and $m$ have nontrivial braiding statistics, $S_X$ and $S_Z$ do not commute with each other and hence are nontrivial logical operators. In this case the boundary is said to not be gapped due to the existence of local logical operators.}
    \label{fig:NeedTrivialBraiding}
\end{figure}

Furthermore, the algebra of anyons that condense at a gapped boundary must be maximal meaning all anyons that do not condense braid nontrivially with some anyon that does condense.
Intuitively, this condition ensures that the condensation of anyons at the boundary leaves nothing behind, and hence is consistent with a gapped boundary with nothing on the other side. 
For abelian anyon theories these conditions are summarized by an object known as a Lagrangian subgroup~\cite{Levin2013Protected}. The existence of a Lagrangian subgroup guarantees the existence of Hamiltonian terms or checks that yield the desired condensation rules:

\begin{lemma}[\cite{Levin2013Protected,Barkeshli2013Classification}]
\label{lem:GappedBoundary}
There is a one-to-one mapping between Lagrangian subgroups of an abelian anyon theory and gapped boundaries.
\end{lemma}

However, this existence Lemma does not provide a practical way to work out the exact lattice terms/stabilizer checks.
Explicit constructions of local Hamiltonian terms that realize boundaries corresponding to any Lagrangian algebra are known for quantum doubles and string-net models~\cite{beigi2011quantum,Kitaev2012Models}.
Here we instead use abstract anyon string operators that appear near a defect as a heuristic guide to obtain a complete set of lattice terms, see Section~\ref{sec:LatticeModel}.

A general line defect in two dimensions can connect a pair of topological phases. 
This includes line defects that are formed by simply placing gapped boundaries for each topological phase that meets the defect line next to one another.
A range of other defects are also possible. 
If the topological phases meeting at a line defect are equivalent, i.e. they are described by the same anyon theory, the defect may be \textit{invertible}. 
An invertible line defect is one that has an inverse, for example the electromagnetic duality line defect in surface code~\cite{bombin2010topological}.  
That is, composing the defect with its inverse defect results in a line defect that is trivial, i.e. does nothing to the anyons passing through it. 
The anyons that condense on an invertible line defect must come in pairs, consisting of an anyon on each side of the defect that are related by an isomorphism of the anyon theory.
Defects that factorize into gapped boundaries from either side are, in a sense, as far as possible from being invertible. 
That is, there is a maximal set of anyons on either side that can condense on the line defect. 

On the other hand, all gapped line defects can be reinterpreted as gapped boundaries of some topological phase via the \textit{folding trick}~\cite{Kitaev2012Models}. 
The folding trick maps an arbitrary line defect to a gapped boundary by folding along the defect line, thereby stacking the topological phases on either side of the line defect onto one another, see Figure~\ref{fig:FoldingTrick}. 
Combined with Lemma \ref{lem:GappedBoundary}, this trick allows us to obtain Hamiltonian terms for arbitrary gapped line defects.

\begin{figure}[H]
    \centering
    \includegraphics[page=62]{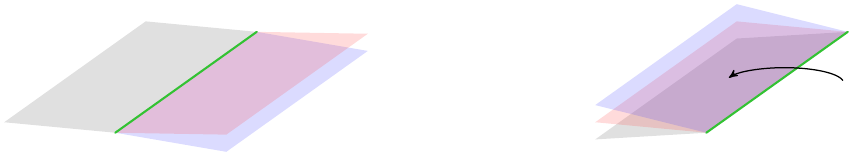}
    \caption{An illustration of the folding trick applied to a line defect (left) to produce an equivalent gapped boundary (right). Each layer connected to the green line is a patch of surface code. On the right, the three superposed patches can be interpreted as a single topological phase, which allows us to use Lemma \ref{lem:GappedBoundary} and figure out the checks corresponding to the defect on the left.}
    \label{fig:FoldingTrick}
\end{figure}

In this work we focus on line defects that reduce to gapped boundaries of multiple copies of the surface code under the folding trick. 
Gapped boundaries of several copies of the surface code are described by Lagrangian subgroups of surface code anyons. 
A Lagrangian subgroup for $c$ copies of surface code is defined by a set of $c$ independent nontrivial anyons with trivial self-statistics, and trivial mutual braiding statistics. 
For a single copy of surface code the possible choices of generator for a Lagrangian subgroup are $e$ or $m$, which correspond to rough and smooth boundaries, respectively. 
The Lagrangian subgroup description of a line defect specifies the anyon types that can be locally created in the vicinity of the defect. 
This information can then be used to derive products of string operators near the defect that create pairs of anyons and push them into the defect, leaving behind no excitations~\cite{Levin2013Protected}, see Figure \ref{fig:NeedTrivialBraiding}. 
Such operators are commonly referred to as condensed on the line defect. 
This picture can be used to derive stabilizer generators for the defect on the lattice. 
For example, a single copy of the surface code with a gapped boundary defined by the condensation of $e$ has short $e$ string operators that end on the gapped boundary. 
Converting these short strings into lattice operators produces the 3-body rough boundary $B_p$ terms introduced in Eq.~\eqref{eq:SurfaceCodeGenerators}. 
A similar procedure applied to $m$ condensing boundaries produces the smooth boundaries of the surface code, see Figure~\ref{fig:SurfaceCodeCondensation}. 
The layer code construction also involves point defects, where line defects meet in multiple copies of surface code. 

Point defects in a two-dimensional topological phase are described by projection operators in the \textit{annulus algebra} that surrounds them. 
An annulus algebra is generated by equivalence classes of all operators supported on an annulus shaped region that create no excitations, modulo the subalgebra of operators that are local within a disclike subregion of the annulus and create no excitations~\cite{haahInvariant}. 
For a stabilizer code this corresponds to the group of logical operators supported on an annulus, modulo stabilizers supported on disclike subregions of the annulus. 
Intuitively this separates the stringlike operators that wrap around the annulus from the short loops that do not. 
Each nontrivial equivalence class corresponds to an independent stabilizer that can be included in the vicinity of the point defect. 
The annulus algebra is a matrix algebra, and can be decomposed into irreducible representation blocks. 
Each irreducible block corresponds to a definite type of point defect.
The projection operator in the annulus algebra that picks out a certain block can be converted into stabilizer generators that realize the corresponding point defect on the lattice~\cite{haahInvariant}. 

\begin{figure}[H]
    \centering
    \includegraphics[page=15]{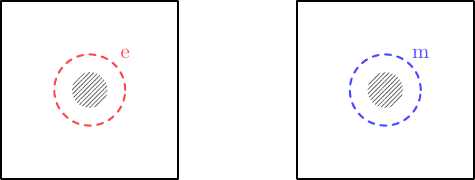}
    \caption{Closed $e$ (red) and $m$ (blue) string operator generators of the annulus algebra around a puncture (shaded circle) in the surface code with no line defects.}
    \label{fig:AnnulusAlg}
\end{figure}

A simple example of the annulus algebra is given by a single copy of surface code with no line defects, see Figure~\ref{fig:AnnulusAlg}. In this case the algebra reduces to a $\mathbb{Z}_2^2$ algebra of closed loop operators generated by $e$ and $m$ loops around the chosen point. 
The irreducible projection operators in this algebra pick out states that have $\pm 1$ braiding statistics with $e,m,$ respectively. 
These four sectors precisely specify the four anyon types due to the modularity, or nondegeneracy, of the braiding~\cite{kitaev2006anyons}. 
Converting the annulus projection operators into lattice operators generates the familiar star and plaquette stabilizer checks of the surface code. 
Here modularity means that the eigenstates of a single star and plaquette can specify a surface code anyon of any type.

\begin{figure}[H]
    \centering
    \includegraphics[page=16]{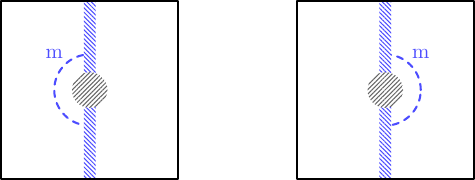}
    \caption{String arc operator generators (blue) of the annulus algebra around a puncture (shaded circle) that lies on a vertical domain wall given by $m$-condensing gapped boundaries (shaded bar).}
    \label{fig:AnnulusAlg2}
\end{figure}

Another simple example is given by a single copy of the surface code with a pair of m-condensing gapped boundaries to vacuum, see Figure~\ref{fig:AnnulusAlg2}. 
In this case the algebra again reduces to $\mathbb{Z}_2^2$, generated by independent arcs of $e$ string operator that enclose the point defect from the left and right side of the gapped boundary region. 
The irreducible projection operators in this algebra pick out states that have $\pm 1$ braiding with the pair of $e$ string arcs.
These four sectors are generated by an $e$ excitation pinned to the defect from the left and the right, respectively.
Converting these projections into lattice operators produces the familiar truncated $X$ stabilizers that appear on the smooth boundary of the surface code.

The above examples demonstrate a heuristic counting rule for a point defect to be fully gapped. This hueristic corresponds to saturating an upper bound of at most one independent generator for the annulus algebra per half-plane of surface code that meets the defect.
In Section~\ref{sec:LatticeModel} we prove that all point defects in the layer code lattice model are fully gapped, i.e. there are no local logical operators, by showing that the local neighborhood around each defect is correctable.

\subsection{Energy Barrier}
\label{sec:intro-energy-barrier}
In this subsection we introduce the energy barrier of a quantum stabilizer code. 

The \textit{energy penalty} of a Pauli operator is defined to be the weight of the syndrome generated by that operator.
The energy barrier of a sequence of singe qubit Pauli operators $(P_i)_{i=1}^{N}$ is the maximum energy penalty of their cumulative product sequence $(\prod_{i=1}^j P_i)_{j=1}^{N}$. 
The energy barrier of a multi-qubit Pauli operator $L$ is defined to be the minimum energy barrier over all sequences $(P_i)_{i=1}^{N}$ of single qubit Pauli operators such that $\prod_{i=1}^N P_i = L$. 
The energy barrier of a logical class is defined to be the minimum energy barrier for any logical representative in the logical class. The \textit{energy barrier} of a code is the minimum energy barrier of a nontrivial logical class. 
The above definitions have variants that are restricted to products of Pauli $X$ (or $Z$) operators. 
A more universal definition allows for a product of a constant number of Pauli operators to be applied at each step, rather than single Pauli operators. 
For LDPC codes, the scaling of the energy barrier with the total number of qubits under both the more and less restrictive definitions must match, up to a universal constant, since applying a constant number of Pauli operators at once can only change the energy penalty by a constant amount. It might be tempting to conclude that a macroscopic energy barrier is sufficient to guarantee the existence of a quantum memory, however there are indications that this assumptions is misguided~\cite{siva2017topological}.

\section{Layer Code Construction} 
\label{sec:LayerCodeConstruction}

In this section we describe the central construction of the current work. 
This construction takes as input an arbitrary CSS stabilizer code and produces as output a layer code. 
Layer codes are three-dimensional topological CSS codes formed by stacks of surface codes, one for each physical qubit, $X$ check, and $Z$ check of the input code, joined together by topological defect lines in a pattern that is based on the Tanner graph of the input code. 

The core idea behind the construction is fairly simple. Our goal is to find topological CSS codes in three dimensions that saturate the BPT bound. 
We start from the observation that concatenating an $[[n,\Theta(n),\Theta(n)]]$ good CSS quantum code $\mathcal{C}$ with an $[[L^2,1, \Theta(L)]]$ surface code 
produces an $[[nL^2, \Theta(n), \Theta(nL)]]$ code with very high stabilizer weight and long range interactions but whose code parameters saturate the BPT bound for $L=\Theta(n)$. 
Our goal is then to reduce the stabilizer weight and interaction range back to constant without sacrificing the desired code parameters. 

\begin{figure}[t]
    \centering
    \includegraphics[page=13]{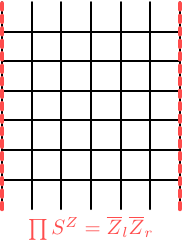}
    \caption{
    \justifying
    A lattice surgery patch allows a product of a pair of global logical operators to be inferred via local measurement of the surface code checks. }
    \label{fig:LatticeSurgeryExample}
\end{figure}

Fortunately, lattice surgery~\cite{horsman2012surface} can address both of these issues. For example, a $\overline{Z}\, \overline{Z}$ check between two distant patches of surface code can be inferred via local stabilizer measurements using a surgery patch as depicted in Figure~\ref{fig:LatticeSurgeryExample}. 
One can then imagine replacing all the high-weight long-range checks with lattice surgery patches to obtain a 3D local code with parameters $[[\Theta(L^3),\Theta(L),\Theta(L^2)]]$. 

This approach, however, quickly runs into issues, as surgery patches for $X$-type checks do not commute with those for $Z$-type checks when they act on the same qubit of $\mathcal{C}$, or equivalently, on the same surface code patch of the concatenated code. The final ingredients we introduce are topological defects that allows us to ``sew together'' lattice surgery patches that act on the same qubits, in a way that restores commutation and preserves the desired code parameters and locality properties.

\begin{figure}[H]
    \centering
    \includegraphics[page=1]{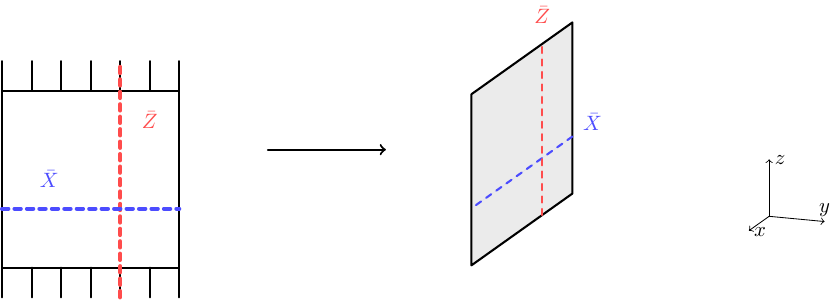}
    \caption{On the left is the usual representation of the surface code with rough and smooth boudaries such that the $\overline{Z}$ logical operator runs from top to bottom, while the $\overline{X}$ logical operator runs from left to right.
    This is an abstraction of the lattice shown in Figure~\ref{fig:SurfaceCodeIntro}. Throughout the rest of this work we represent the surface code in 3D space with the boundary conditions oriented such that the $\overline{Z}$ operator runs along the $\hat{z}$ axis, and the $\overline{X}$ operator runs along the $\hat{x}$ axis. In the three-dimensional picture the rough boundaries are left implict.  }
    \label{fig:surface-code-bdry}
\end{figure}

\newpage
\subsection{Overview}
Before explaining the details of the layer code construction, we present a high level overview in Algorithm \ref{alg:layer-codes-alg}. 
Figure~\ref{fig:shor-steps} illustrates the steps of the algorithm when applied to Shor's code as an input.

\algnewcommand\algorithmicforeach{\textbf{for each}}
\algdef{S}[FOR]{ForEach}[1]{\algorithmicforeach\ #1\ \algorithmicdo}

\begin{algorithm}[H]
\caption{Layer code construction}
\label{alg:layer-codes-alg}
\begin{algorithmic}
    \Require A CSS LDPC code $\mathcal{C}$, with parameters $[[n,k,d]]$, $n_X$ $X$-type checks, $n_Z$ $Z$-type checks, and maximum check weight $w$.
    \Ensure A CSS code that is local in 3D with parameters $[[\Theta(nn_Xn_Z),k,\Omega(\frac{1}{w}d \min(n_X,n_Z))]]$.
    \ForEach {qubit $i \in [0, \dots, n-1]$}
    \State Create an $xz$-layer $\mathcal{D}_i$, with the boundary conditions specified in Figure \ref{fig:surface-code-bdry}. (Depicted in Figure~\ref{fig:shor-steps}~a).
    \EndFor
    \ForEach {$s_j$ in the $Z$-type checks of $\mathcal{C}$}
    \State Let $i_1,\dots,i_T$ be the qubits in the support of $s_j$. Create a $yz$-layer $\mathcal{Z}_j$ starting from $\mathcal{D}_{i_1}$, going through each of the $\mathcal{D}_{i_t}$ until $\mathcal{D}_{i_T}$, with the specifications given in Section \ref{par:k-body-terms} and illustrated in Figure~\ref{fig:k-body-Z-checks}. (Depicted in Figure~\ref{fig:shor-steps}~b).
    \EndFor
    \ForEach {$\tilde{s}_k$ in the $X$-type checks of $\mathcal{C}$}
    \State Let $i_1,\dots,{i_T}$ be the qubits in the support of $\tilde{s}_k$. Create an $xy$-layer $\mathcal{X}_k$ starting from $\mathcal{D}_{i_1}$, going through each of the $\mathcal{D}_{i_t}$ until $\mathcal{D}_{i_T}$, with the specifications given in Section \ref{par:k-body-terms}, and illustrated in Figure~\ref{fig:k-body-X-checks}. (Depicted in Figure~\ref{fig:shor-steps}~c).
    \EndFor
    \ForEach {$s_j$ in the $Z$-type checks of $\mathcal{C}$}
        \ForEach {$\tilde{s}_k$ in the $X$-type checks of $\mathcal{C}$}
            \State If the support of $s_j$ and $\tilde{s}_k$ overlap, introduce line and point defects between the layers $\mathcal{Z}_j$ and $\mathcal{X}_k$ following the specifications of Section~\ref{par:line-point-defects}. (Depicted in Figure~\ref{fig:shor-steps}~d).
        \EndFor
    \EndFor
\end{algorithmic}
\end{algorithm}
\newpage
\begin{figure}[H]
    \centering
    \includegraphics[page=60]{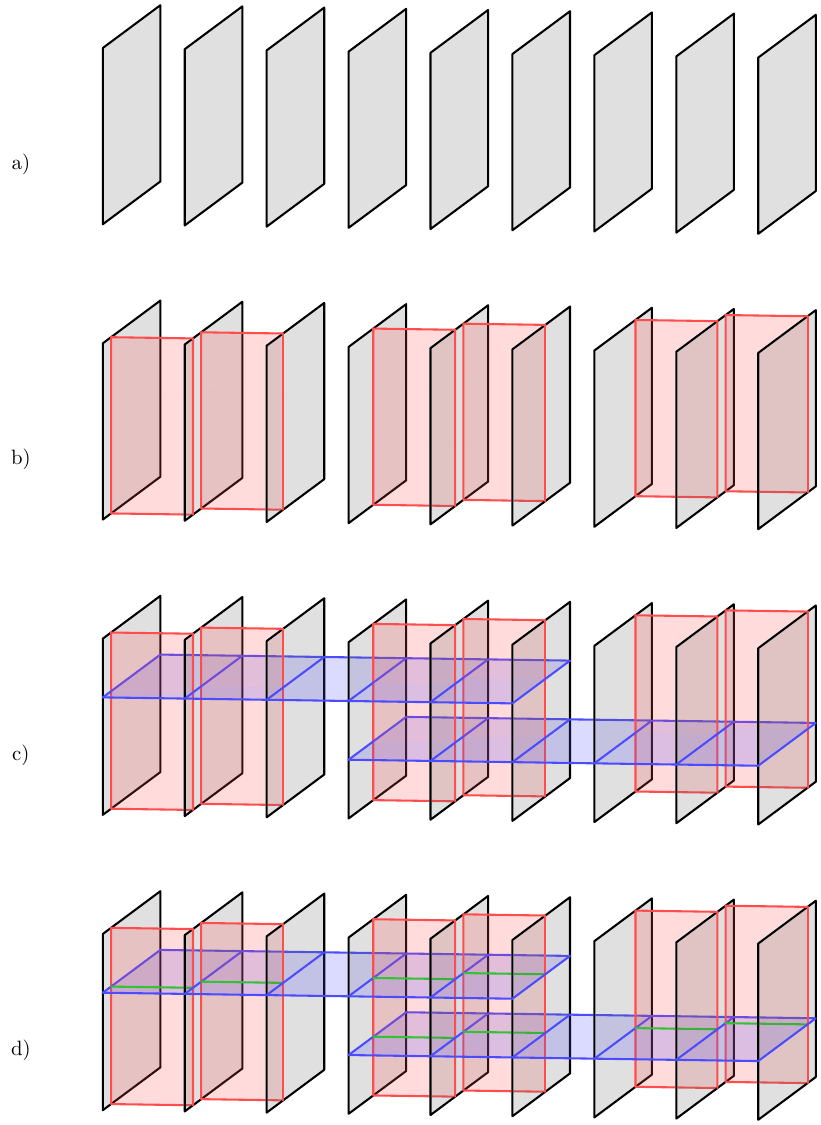}
    \caption{Steps of the layer code construction outlined in Algorithm~\ref{alg:layer-codes-alg} when applied to Shor's code.}
    \label{fig:shor-steps}
\end{figure}

\subsection{Lattice surgery as a topological defect} 

We start by reviewing lattice surgery on the surface code from a topological quantum field theory~\cite{Witten1988TQFT,atiyah1988topological} lens. This allows us to abstract away unnecessary details in the construction. The lattice stabilizer checks that implement the defects we introduce can be found in Section~\ref{sec:LatticeModel}.

The abstract picture for the surface code is a patch of \TC topological order which supports anyon types $\{1,e,m,\psi\}$ together with a pair of $e$-condensing horizontal, and $m$-condensing vertical, gapped boundaries arranged as shown in Figure~\ref{fig:surface-code-bdry}. See Section~\ref{sec:Background} for basic background information on surface code and its topological quantum field theory features including anyons, defects and boundaries. 

Let us start by considering how a logical $\overline{Z}\,\overline{Z}$ stabilizer can be enforced on a pair of surface codes by introducing a topological defect line\footnote{The connection between anyon condensation and measurement has been explored recently in Refs.~\cite{Ellison2022,ellison2022pauli,kesselring2022anyon}.}. The defect line is pictured in red in Figure \ref{fig:zzDefectNonlocal}. 

\begin{figure}[H]
    \centering
    \includegraphics[page=63,scale=.8]{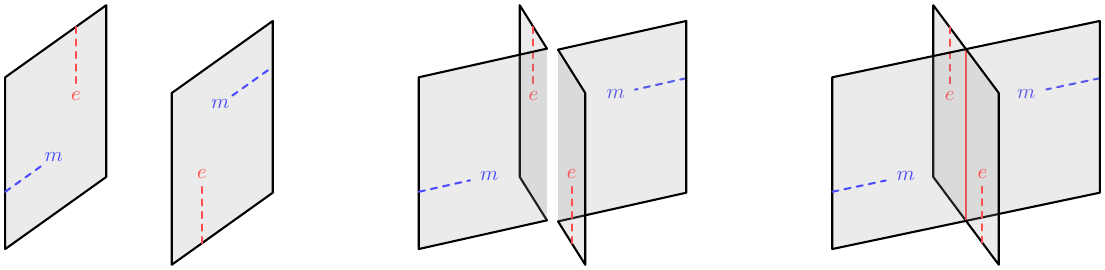}
    \caption{ On the far left we depict a pair of non-interacting surface codes, with the usual rough and smooth boundaries. The second and third pictures show the two patches being bent and brought together at their middle, then ``stitched" with a defect (red line). Away from the defect the condensation rules are unchanged, however near the defect they differ. The additional rules are elucidated below.}
    \label{fig:zzDefectNonlocal}
\end{figure}

Our chief concern is to ensure that the only logical $X$-type operator remaining is $\overline{X}\,\overline{X}$. 
This is achieved by having $m_1^\ell m_1^r m_2^\ell m_2^r$ be the only $m$-type anyon condensed by the defect, where $1,2,$ labels the layer and $\ell,r,$ labels the region to the left, right, of the defect. 

We now move on to the $e$-anyons; these have to condense somehow on the defect, otherwise the $\overline{Z}$ operators on a surface code patch on either side of the defect would become independent.
The analysis is simplified by the fact that condensing anyons map to elements of the normalizer, see Figure \ref{fig:NeedTrivialBraiding}. These elements of the normalizer should further be stabilizers, otherwise the defect introduces constant-size logicals. It is easy to see that in order to satisfy this condition, $e$-anyons have to condense in pairs: $e_1^\ell e_2^\ell, e_1^r e_2^r, e_1^\ell e_1^r$. In terms of elements of the normalizer, these pairs are mapped to pairs of semi circles at the defect, each intersecting with $\overline{X}\,\overline{X}$ twice and therefore satisfying the commutation requirement. In total, the set of anyons that the defect has to condense is $\langle e_1^\ell e_2^\ell, e_1^r e_2^r, e_1^\ell e_1^r, m_1^\ell m_1^r m_2^\ell m_2^r  \rangle$. Another equivalent way to derive this group, is to note that according to Lemma \ref{lem:GappedBoundary} we are looking for a Lagrangian subgroup of anyons that includes $m_1^\ell m_1^r m_2^\ell m_2^r$ and no further purely $m$-type excitations. It can be verified that this corresponds to $\langle e_1^\ell e_2^\ell, e_1^r e_2^r, e_1^\ell e_1^r, m_1^\ell m_1^r m_2^\ell m_2^r  \rangle$. These condensation rules, and the logical operators they give rise to are exemplified in Figure~\ref{fig:zzDefectNonlocalExamples}.

\begin{figure}[H]
    \centering
    \includegraphics[page=64,scale=.8]{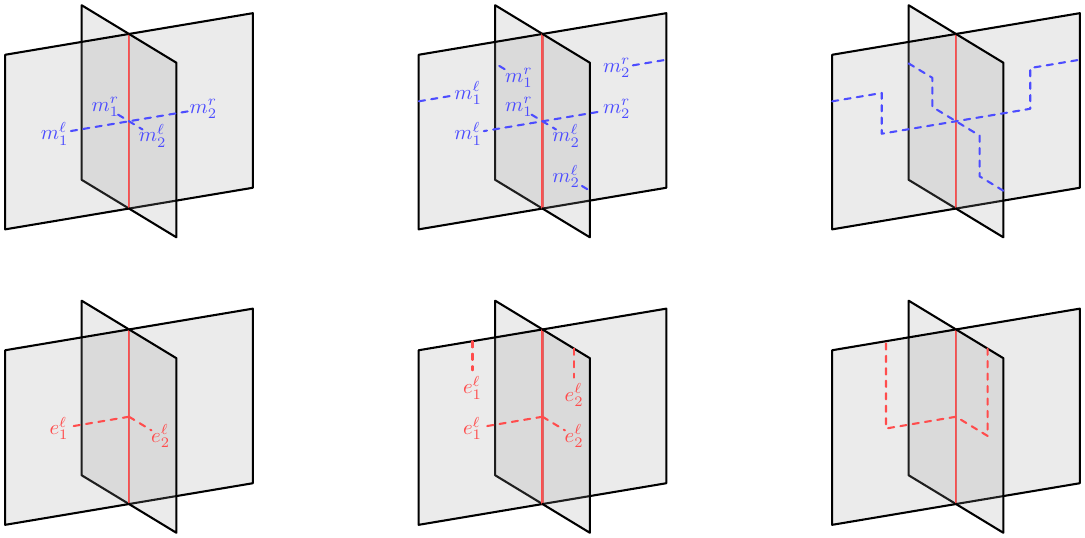}
    \caption{An illustration of how the condensation rules near a defect junction and gapped boundaries produce logical operators (top) and stabilizers (bottom).}
    \label{fig:zzDefectNonlocalExamples}
\end{figure}

Figure~\ref{fig:zzDefectNonlocalExamples} illustrates the impact of the condensation rules on the normalizer. The top row shows how the condensation of $m_1^\ell m_1^r m_2^\ell m_2^r$ allows the operator $\overline{X}\, \overline{X}$ to cross the defect line without leaving excitations behind. This is in contrast with the behaviour of a single $\overline{X}$ operator. Indeed, in the case of two uncoupled surface codes, an $\overline{X}$ operator could be created by creating $m_1^\ell$ and $m_1^r$ via boundary condensation, and joining them together in the middle. 
However, due to the line defect this is no longer possible: if we again create $m_1^\ell$ and $m_1^r$ and drag them close to the line defect, since $m_1^\ell m_1^r m_2^\ell m_2^r$ is the \emph{only} $m$-type anyon that can condense at the defect, we are left with an $m_2^\ell m_2^r$ excitation.
This shows that $\overline{X}$ is no longer in the normalizer of the code.

A dual perspective on the same phenomenon is illustrated in the bottom row of Figure~\ref{fig:zzDefectNonlocalExamples}. Using condensation rules, it is possible to create the dashed red string in the bottom right picture without leaving excitations behind. Because this string is created by a constant size operator, it belongs to the stabilizer group. This string can then be dragged down and deformed into a $\overline{Z}\,\overline{Z}$ operator supported on the two patches. This accomplishes the goal of enforcing a $\overline{Z}\,\overline{Z}$ stabilizer.

We remark that the defect depicted in Figure~\ref{fig:zzDefectNonlocalExamples} is not invertible, and also does not factor into a pair of gapped boundaries -- i.e. it cannot be obtained by \emph{locally} modifying the checks of the surface code patches. 
Above we have used the folding trick to view the line defect as a gapped boundary of four half-planes of \TC topological order. 
Recall that for stacks of \TC topological order, a Lagrangian subgroup of condensing anyons can be described compactly by a set of independent generators,  which have trivial self and mutual braiding statistics, with size equal to the number of toric codes that terminate at the gapped boundary. 

\begin{figure}[H]
    \includegraphics[page=3,scale=.8]{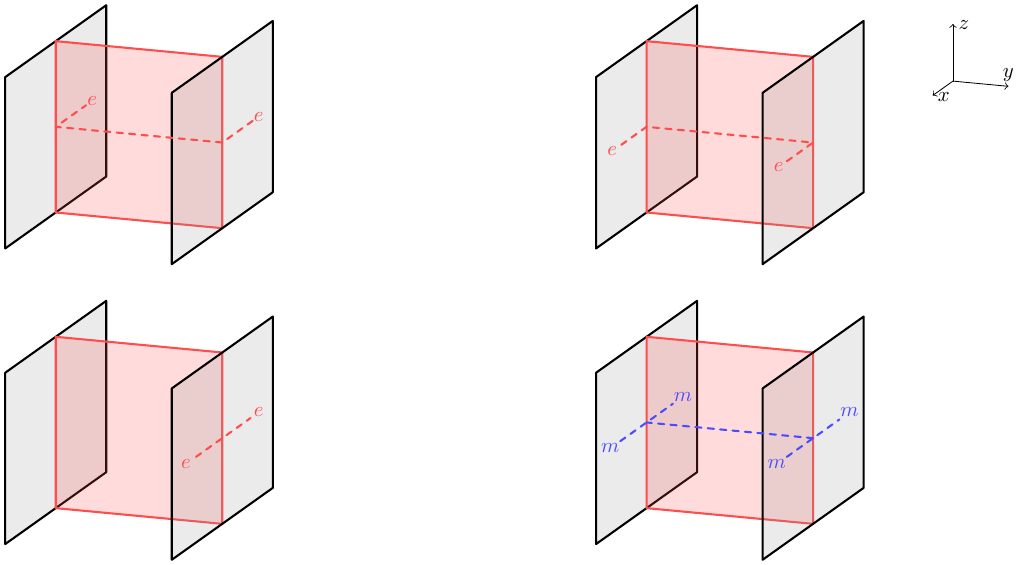}
    \centering
    \caption{
    Each grey layer is a surface code patch with the boundary conditions specified in Figure~\ref{fig:surface-code-bdry}. 
    We refer to the grey layers as $xz$-layers. The red layer corresponds to a topological defect emulating a $\overline{Z}\, \overline{Z}$ stabilizer between the two $xz$-layers. We refer to the red layers as $yz$-layers. The red layer condenses $\langle e_1^\ell e_2^\ell, e_1^r e_2^r, e_1^\ell e_1^r, m_1^\ell m_1^r m_2^\ell m_2^r  \rangle$.}
    \label{fig:ZZDefect}
\end{figure}

The topological defect of Figure \ref{fig:zzDefectNonlocal} is only local if the pair of surface code layers involved are within a constant distance of each other. 
Generalizing this construction to allow concatenation with codes that are not spatially local, and whose checks might involve qubits that are fart apart, requires an extra trick. 
This is provided by decomposing the above topological defect into a pair of topological defects joined by an additional layer of surface code that runs between the original surface code layers, see Figure~\ref{fig:ZZDefect}. 

\begin{figure}[H]
    \centering
    \includegraphics[page=4, scale=.9]{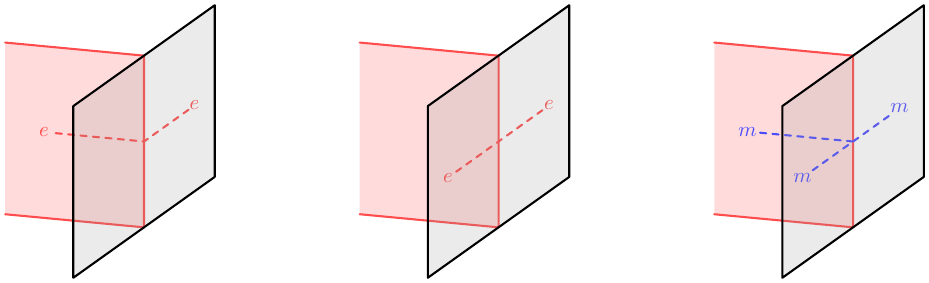}
    \caption{Topological defect trijunction (red line on the grey layer) between a $zy$ layer and an $xz$-layer that condenses $\langle e_1^\ell e_1^r, e_1^\ell e_{\hat{z}}, m_1^\ell m_1^r m_{\hat{z}} \rangle$. This elucidates how to realize the defect depicted in Figure \ref{fig:ZZDefect}. Here, both the red layer and the grey layer are surface codes, while their intersection line supports a nontrivial defect. Using the folding trick, the line where they intersect can also be seen as the gapped boundary of three surface codes.}
    \label{fig:zyLayerDefect}
\end{figure}

This results in a pair of topological defects where a triple of \TC topological orders meet, see Figure~\ref{fig:zyLayerDefect}. The new pair of line defects are specified by the condensed anyons $\langle e_1^\ell e_1^r, e_1^\ell e_{\hat{z}}, m_1^\ell m_1^r m_{\hat{z}} \rangle$ and $\langle e_2^\ell e_2^r, e_2^\ell e_{\hat{z}}, m_2^\ell m_2^r m_{\hat{z}} \rangle$ respectively, which are equivalent up to a relabelling of the anyons -- note that both these subgroups are Lagrangian, and therefore the new trivalent defect is guaranteed to exist.

\begin{figure}[H]
    \centering
    \includegraphics[page=65,scale=.9]{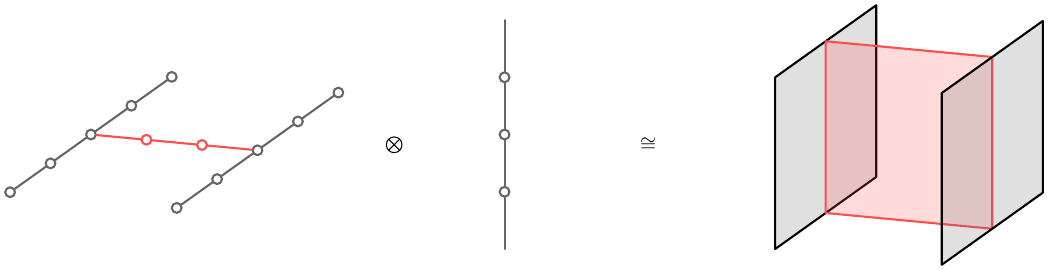}
    \caption{On the left of the ``$\cong$" sign are two classical codes on either side of ``$\otimes$" that are hypergraph-producted to yield the topological defect on the right of the ``$\cong$" sign. Each ``$\circ$" is a qubit, and every line between two qubits is a simple $ZZ$ check. This representation allows for a classical interpretation of the defect on the right. Following Ref.~\cite{baspin2023combinatorial}, the leftmost code can be viewed as a pair of repetition codes (grey) coupled by an effective long-range $ZZ$ (red). This enforces the repetition codes to either \emph{both} be in the all $0$ state, or in the all $1$ state. }
    \label{fig:zzDefectHGP}
\end{figure}

We can now state the main point of this section: the resultant defect, i.e. the red sheet in Figure \ref{fig:ZZDefect}, corresponds exactly to a lattice surgery patch \cite{horsman2012surface,vuillot2019code} connecting the two surface codes. By using the language of topological defects, we showed how to enforce a $\overline{Z} \, \overline{Z}$ on two surface codes. There is also a way to view this defect from a combinatorial lens by way of the hypergraph product construction~\cite{tillich2014quantum}, see Figure \ref{fig:zzDefectHGP}.

\subsection{Repetition code example} 

It is convenient at this point to illustrate our discussion with a specific example. 
Here we describe the layer code that corresponds to taking the 3-qubit repetition code $\mathcal{S} = \langle Z_1 Z_2, Z_2 Z_3 \rangle$ as input.
This layer code is given by two pairs of the defects from Figure~\ref{fig:zyLayerDefect} connecting together a triple of surface code layers as shown in  Figure~\ref{fig:RepetitionCode2}.
This layer code displays an example of the topological defect network structure inherent to all layer codes. 
That is, it is formed by layers of surface code, joined together by topological defect line junctions according to the Tanner graph of the input repetition code. 

\begin{figure}[H]
    \centering
    \includegraphics[page=5,scale=.9]{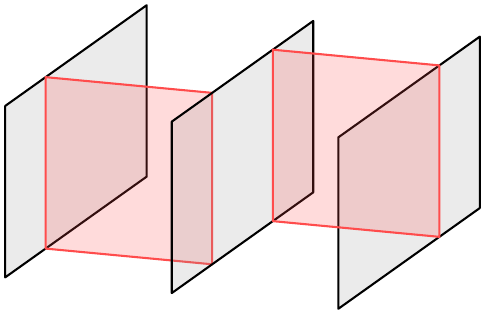}
    \caption{Topological defect network emulating the concatenation of a surface code and a repetition code. It is composed of three $xz$-layers, each representing a qubit of the repetition code, and two $yz$-layers, each representing a check of the repetition code.}
    \label{fig:RepetitionCode2}
\end{figure}

In this layer code, the $\hat{y}$ boundaries of the $yz$-oriented surface codes are taken to be $e$-condensing. 
With this choice, all boundaries ending on a $yz$-plane match and are $e$-condensing. 
Due to the defects, the ${Z}$ string operator in each layer provides an equivalent logical $\overline{Z}$ representative. 
On the other hand, the ${X}$ string operator in an individual layer no longer preserves the code space as it creates an additional excitation when it pierces the newly introduced line defects. 
A logical $\overline{X}$ representative of the concatenated code is found by connecting together the product of the $X$ string operators in the original three surface code layers with additional string operators that remove the excitations that appear at the junctions, see Figure~\ref{fig:RepetitionCodeLogical}. 

\begin{figure}[H]
    \centering
    \includegraphics[page=6,scale=.9]{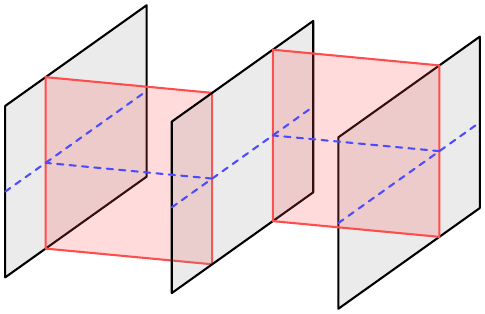}
    \caption{Topological defect network emulating the concatenation of a surface code and a repetition code. The network enforces that the logical $\overline{X}$ operator has the same structure as that of the concatenated code, without the need for high weight long range interactions.}
    \label{fig:RepetitionCodeLogical}
\end{figure}

\subsection{Generalizing the construction}

\label{par:Generalizing}
To generalize the above construction to Shor's code, and beyond to general CSS codes, we need to introduce several further features:
\begin{enumerate}
    \item 
    \label{item:Generalizing1}
    Defects that implement $k$-body $X$, or $Z$, stabilizer checks.
    \item 
    \label{item:Generalizing2}
    Line defects that resolve the crossing of $xy$ and $yz$ oriented surface code layers. 
    \item 
    \label{item:Generalizing3}
    Point defects where multiple line defects meet.
\end{enumerate}
Item 3.~above captures point defects where a line defect meets a boundary. This is because gapped boundaries of copies of the surface code can be viewed as line defects. 

\subsection{$k$-body terms.} 
\label{par:k-body-terms} 
To address the first two items above it is convenient to introduce an abstract picture for a line defect that captures the anyons it condenses while ignoring the homogeneous structure along the defect. 
For the $\langle e_1^\ell e_1^r, e_1^\ell e_{\hat{z}}, m_1^\ell m_1^r m_{\hat{z}} \rangle$ defect this picture is
\begin{figure}[H]
    \centering
    \includegraphics[page=2]{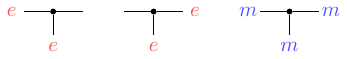}
\end{figure}
we henceforth refer this defect as the $\bullet$-defect. 
This defect can be generalized to an arbitrary number $k$ of incoming layers by joining together $(k-2)$ $\bullet$-defects with three incoming layers each. 
It is easy to verify that the order does not matter. 
Due to the symmetry of \TC under exchanging $e$ and $m$, there is a related defect with the condensation rules
\begin{figure}[H]
    \centering
    \includegraphics[page=3]{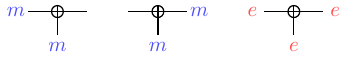}
\end{figure}
we refer to this as the $\oplus$-defect. 
Following the discussion for the $\bullet$-defect above, the $\oplus$-defect generalizes straightforwardly to any number of incoming layers. 

By combining the above defects we can implement the $k$-body generalization of the concatenated $Z$ stabilizer check. 
This is depicted in Figure~\ref{fig:DefectDecomp}, where we make use of the decomposition of a $k$-layer $\oplus$-defect into trivalent $\oplus$-defects to rearrange the network into one defect per $xy$-surface code layer, connected by an auxiliary layer of $yz$-surface code\footnote{This decomposition is closely related to Shor's scheme for syndrome extraction~\cite{shor1996fault}.}. 
An analogous construction for concatenated $X$ stabilizer checks is achieved by exchanging the roles of the $\bullet$-defects and $\oplus$-defects in Figure~\ref{fig:DefectDecomp}.

\begin{figure}[H]
    \centering
    \includegraphics[page=4, scale = 1.5]{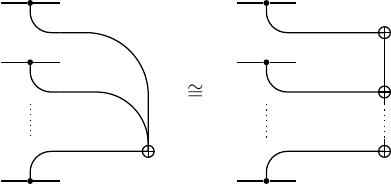}
    \caption{Decomposition of a defect corresponding to a single $k$-body check into a string of defects corresponding to at most $3$-body checks.}
    \label{fig:DefectDecomp}
\end{figure}

We have now answered item~(\ref{item:Generalizing1}) from the list in Section~\ref{par:Generalizing} as we have a recipe to implement $k$-body $X$, or $Z$, stabilizer checks on layers of surface code using defects that are local in 3D. 
The defects involved in a $Z$ stabilizer check are shown in Figure~\ref{fig:k-body-Z-checks}, and those involved in an $X$ stabilizer check are shown in Figure~\ref{fig:k-body-X-checks}.
\begin{figure}[H]
    \centering
    \includegraphics[page=7]{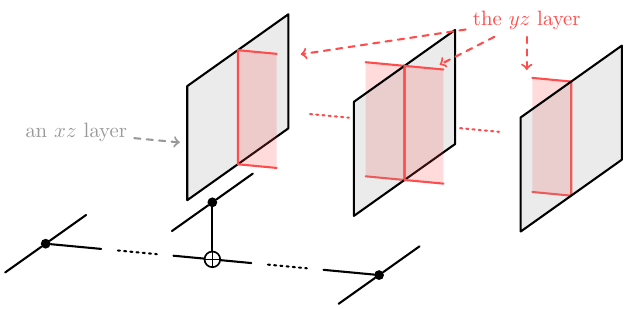}
    \caption{Defects associated with an arbitrary $k$-body $Z$-check. The network of surface code and defects corresponding to this check is referred to as a $yz$-layer. Such defects are illustrated as red layers in $yz$-planes below.}
    \label{fig:k-body-Z-checks}
\end{figure}
\begin{figure}[H]
    \centering
    \includegraphics[page=8]{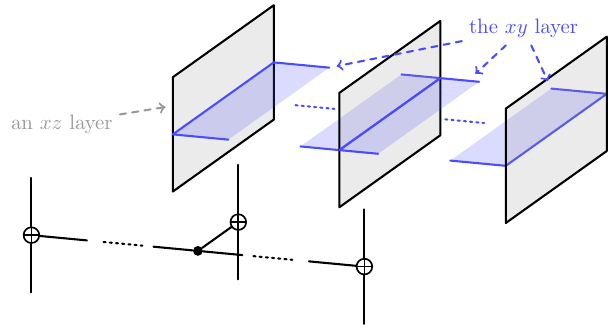}
    \caption{Defects associated with an arbitrary $k$-body $X$-check.
    The network of surface code and defects corresponding to this check is referred to as an $xy$-layer. Such defects are illustrated as blue layers in $xy$-planes below.}
    \label{fig:k-body-X-checks}
\end{figure}

These defects serve to introduce equivalences on $Z$-type, respectively $X$-type, string operators that correspond to multiplication by the associated stabilizer check operator. 
Additionally, they introduce an energy penalty for any $X$-type, respectively $Z$-type, string operator that passes through the defect. 
Only an even number of $X$-type, respectively $Z$-type, string operators can pass thorough a defect on the $xz$ surface code layers that it interacts with without leaving the codespace. 
This is because each such string operator deposits an excitation on the defect, which can be paired up and annihilated if their number is even, see Figure~\ref{fig:RepetitionCodeLogical}.

\subsection{Line defects at the intersection of $xy$ and $yz$-layers} 
\label{par:line-point-defects} 

Our construction for a general CSS code involves defects coupled to $xy$-layers that implement concatenated $X$ checks, as well as defects coupled to $yz$-layers that implement concatenated $Z$ checks. 
For $xy$ and $yz$-layers that do not interact with any common $xz$-layers, corresponding to $X$ and $Z$ checks that have nonoverlapping support, any lines of intersection along $\hat{y}$ are taken to be trivial crossings. 
On the other hand, for an $xy$ and $yz$-layer that intersect nontrivially on, necessarily an even number of, common $xz$-layers a trivial crossing along $\hat{y}$ does not suffice. 
To see why this is the case, we note that a $Z$-type string operator in an $xz$ surface code layer is equivalent to one, or two, $Z$-type string operators in any $yz$-layer that meets the $xz$-layer at a nontrivial defect.
If the $yz$-layer were to intersect trivially with all $xy$-layers, the $Z$-type string operators therein are free to pass across the full system while preserving the codespace, thereby implementing undesired logical string operators. 

Additionally, two lattice surgery patches supported on the same data surface code do not commute when these patches are, respectively, from a Z check and an X check.

To find appropriate $\hat{y}$ defect lines we introduce additional structure to the checks of the input CSS code. 
We consider qubits laid out along a line, and note there are 9 inequivalent configurations for an overlapping $X$ and $Z$ check in total.

\begin{figure}[H]
    \centering
    \includegraphics[scale = 0.4, page=5]{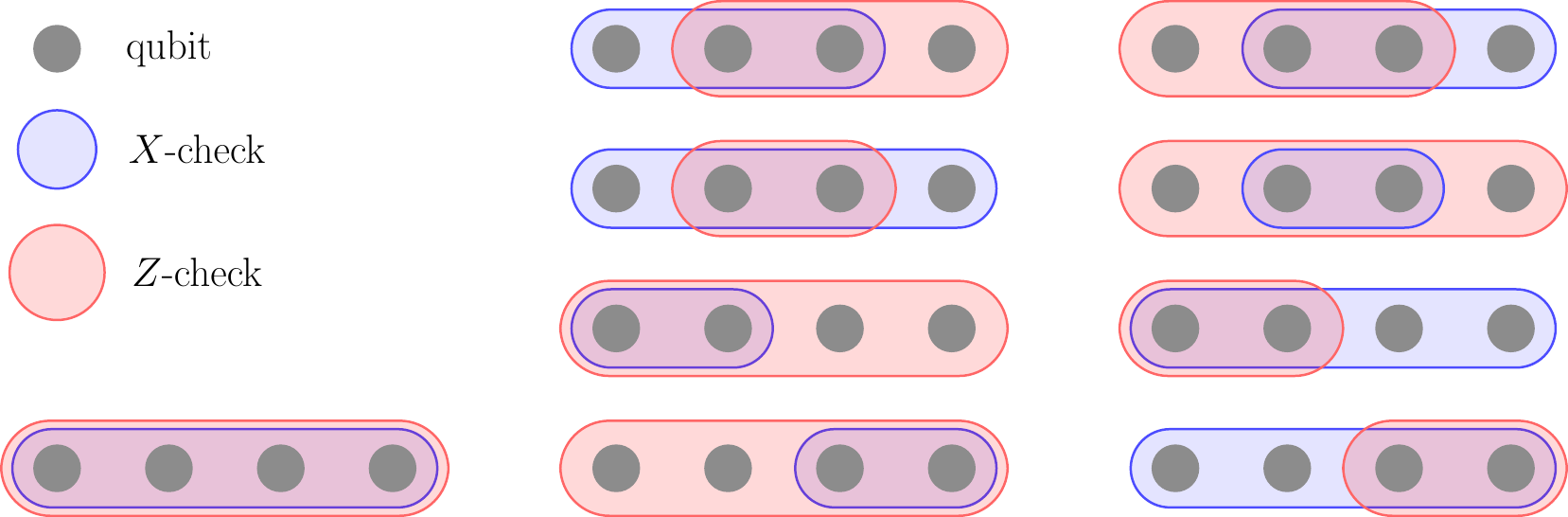}
\end{figure}

We remark that the following pair of single qubit overlaps are not possible as the $X$ and $Z$ checks commute and hence must share an even number of qubits.

\begin{figure}[H]
    \centering
    \includegraphics[scale = 0.4, page=6]{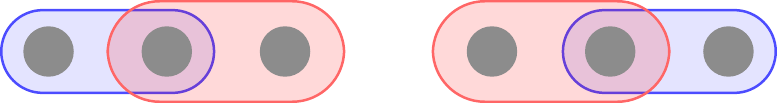}
\end{figure}

The extra structure we make use of in our defect network construction is a pair matching on the qubits that are contained within the support of each overlapping pair of $X$ and $Z$ checks. This pair matching is depicted by green lines in the example below.

\begin{figure}[H]
    \centering
    \includegraphics[scale = 0.4, page=7]{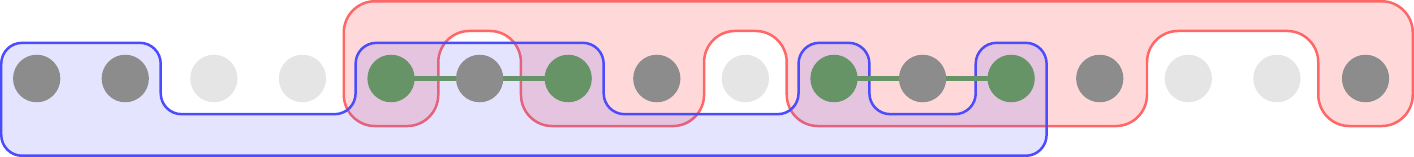}
\end{figure}

This pairing is defined by the linear ordering of the qubits in the general input CSS code that have been arranged along a line. 
The qubits in the intersection of the supports of an $X$ and $Z$ check are then grouped into pairs following this ordering, depicted by green lines above. 

The choice of $\hat{y}$ defect lines in the layer code relies on the additional pairing data. 
Along the $\hat{y}$ intersection line of an $xy$ and $yz$-layer we introduce nontrivial line defects that run between $xz$-layers corresponding to input qubits that have been paired, see Figure~\ref{fig:LineDefects}.  
The remaining unpaired $\hat{y}$ line segments are left as trivial crossing defects, see Figure~\ref{fig:LineDefects}.

\begin{figure}[H]
    \centering
    \subfloat
    {{\includegraphics[page=10,scale=.8]{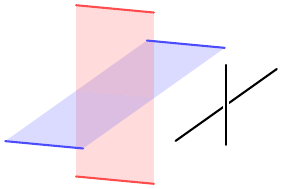}}}
    \qquad \qquad
    \subfloat
    {{\includegraphics[page=9,scale=.8]{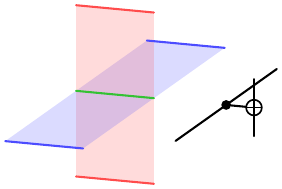} }}
    \caption{A trivial $\hat{y}$ line defect (left), and a nontrivial $\hat{y}$ line defect (right).}
    \label{fig:LineDefects}
\end{figure}

It is now easily verified that a $Z$ string operator within an $xz$-layer and its equivalent $Z$ string operators within $yz$-layers create equivalent excitations, up to local string operators. 
An analogous property holds for $X$ string operators. 
This resolves item~(\ref{item:Generalizing2}) from the list in Section~\ref{par:Generalizing}.

\subsection{Point defects} 
\label{sub:point-defects}

We now address item~(\ref{item:Generalizing3}) from the list in Section~\ref{par:Generalizing}. 
Recall that point defects in copies of the surface code are fixed by their eigenvalues under the codespace preserving, topologically nontrivial, operators supported on an annulus surrounding the point, see Section~\ref{sec:Background}. 
Below, we demonstrate a complete set of local topological operators whose eigenvalues fully specify the point defects that appear in the layer code construction based on the heuristic counting from Section~\ref{sec:Background}. 
Each of these operators is constant size and creates no excitations, and therefore will be included in the stabilizer group.
In Section~\ref{sec:LatticeModel} we find Pauli checks on a lattice model that realize these operators. 
We go on to show that they fully lift the local degeneracy as the neighborhood around each defect is correctable.

\subsubsection*{Bulk point defects}

We first consider point defects within the bulk of the layer code.
We remark that the list below does not contain trivial intersections, where layers simply pass through one another, as those decompose into standard annulus algebras of the individual \TC layers. 
\begin{enumerate}
    \item 
    \label{point:1}
    The first point defect configuration corresponds to an $xy$ and $yz$-layer ending on the same $xz$-layer. 
    There are four incoming half-planes of \TC topological order and hence four independent topological operators in the shell around the point suffice to lift all degeneracy:
    \begin{figure}[H]
        \centering
        \includegraphics[scale = 0.6, page=11]{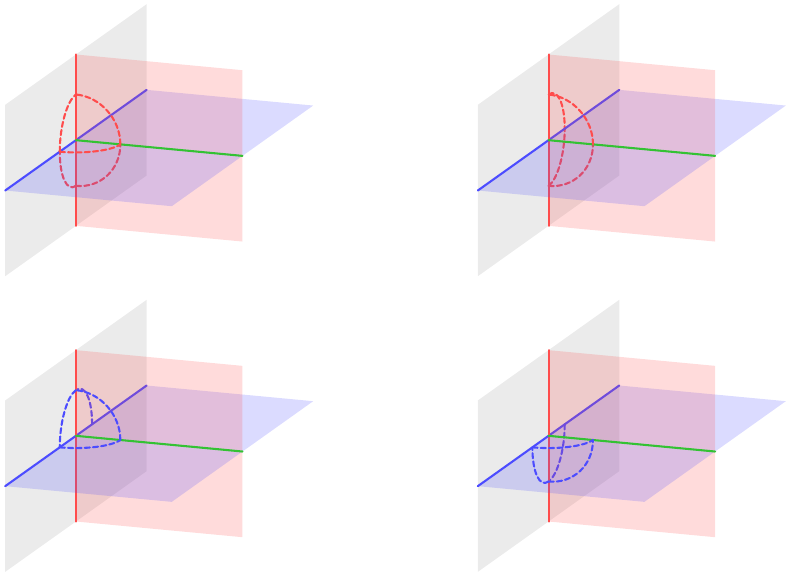}
    \end{figure}

    \item 
    \label{point:2}
    The second point defect is simply the reflection of the 1st through an $xz$-plane:
    \begin{figure}[H]
        \centering
        \includegraphics[scale = 0.6, page=12]{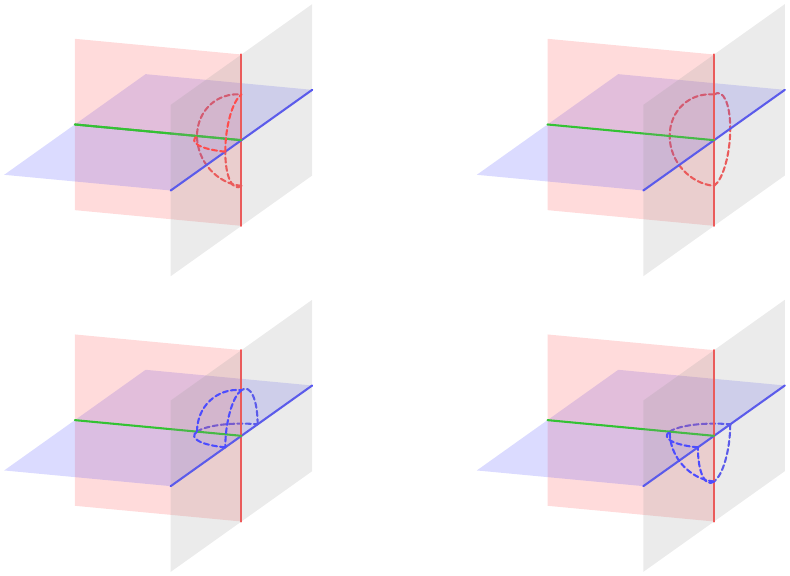}
    \end{figure}

    \item 
    \label{point:3}
    The third point defect configuration corresponds to an $xy$ and $yz$-layer interacting with a common $xz$-layer, but not ending there. 
    A nontrivial $\hat{y}$ defect ends on the point from above. 
    Here there are 6 incoming half-planes of \TC and hence we find 6 independent topological operators in the shell around the point:
    \begin{figure}[H]
        \centering
        \includegraphics[scale = 0.6, page=13]{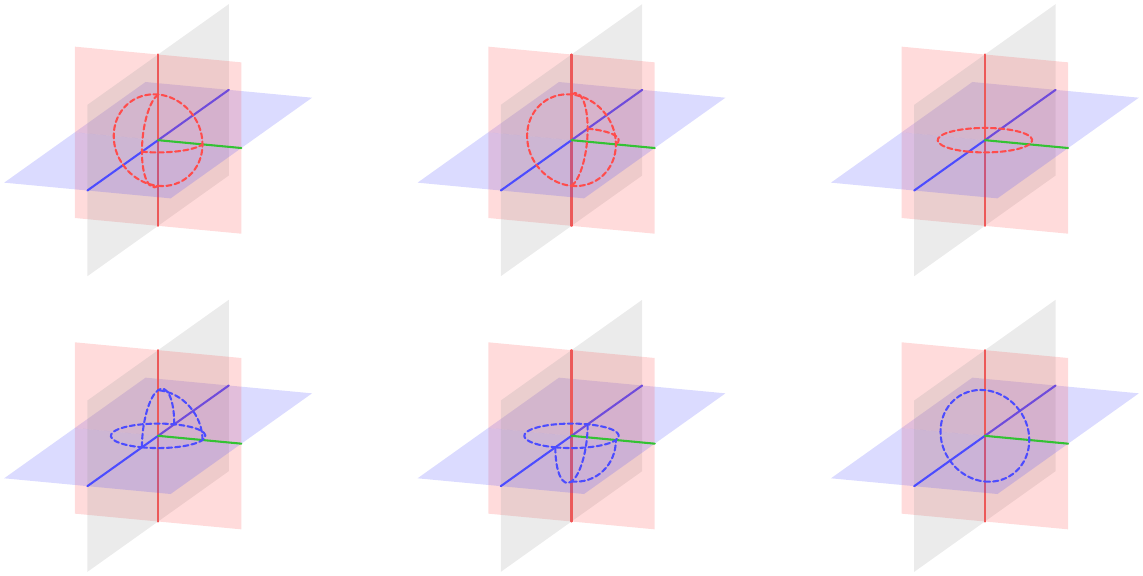}
    \end{figure}

    \item 
    \label{point:4}
    The fourth point defect configuration is a reflection of the 3rd through an $xz$-plane:
    \begin{figure}[H]
        \centering
        \includegraphics[scale = 0.6, page=14]{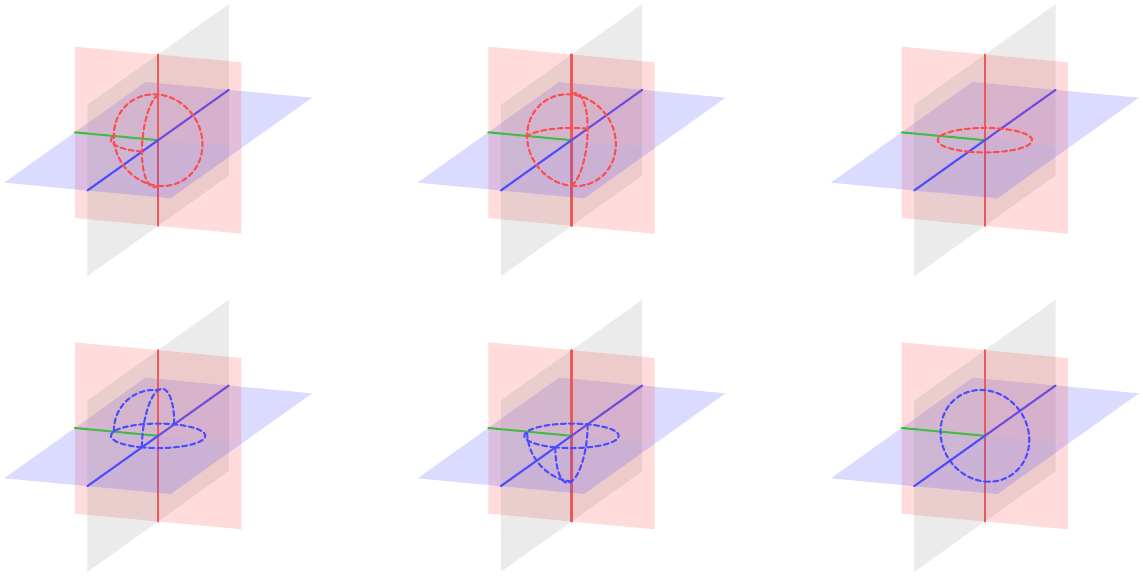}
    \end{figure}

    \item 
    \label{point:5}
    The fifth point defect configuration corresponds to an $xy$-layer interacting with an $xz$-layer that a $yz$-layer ends on from above. 
    This necessitates a nontrivial $\hat{y}$ defect also ending on the point from above. 
    There are 5 half-planes of \TC meeting at the point defect, and hence we find 5 independent topological operators in the shell:
    \begin{figure}[H]
        \centering
        \includegraphics[scale = 0.6, page=15]{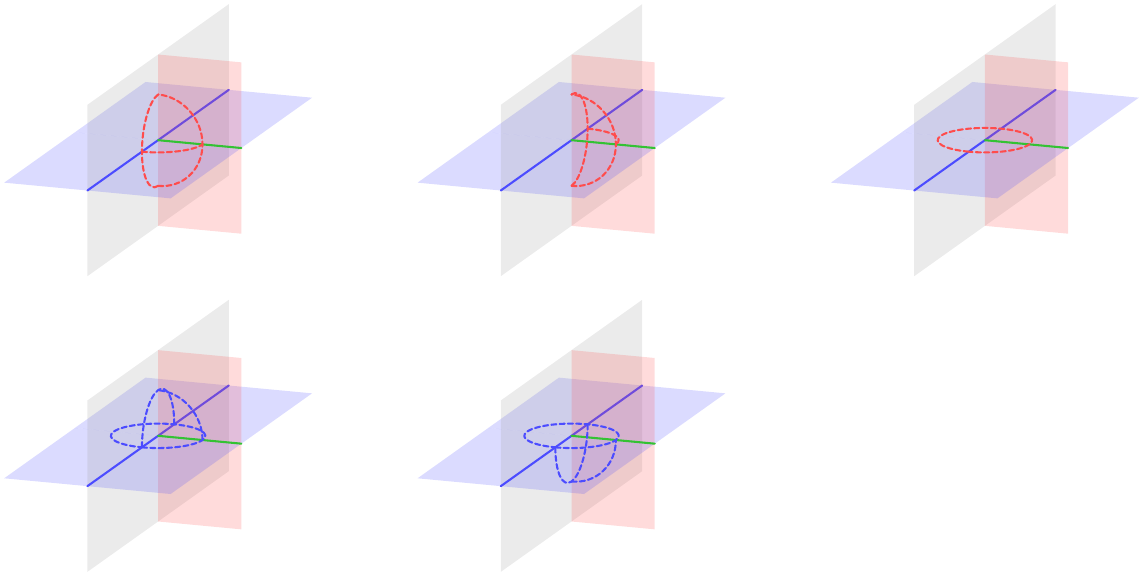}
    \end{figure}

    \item 
    \label{point:6}
    The sixth point defect configuration is a reflection of the fifth through and $xz$-plane:
    \begin{figure}[H]
        \centering
        \includegraphics[scale = 0.6, page=16]{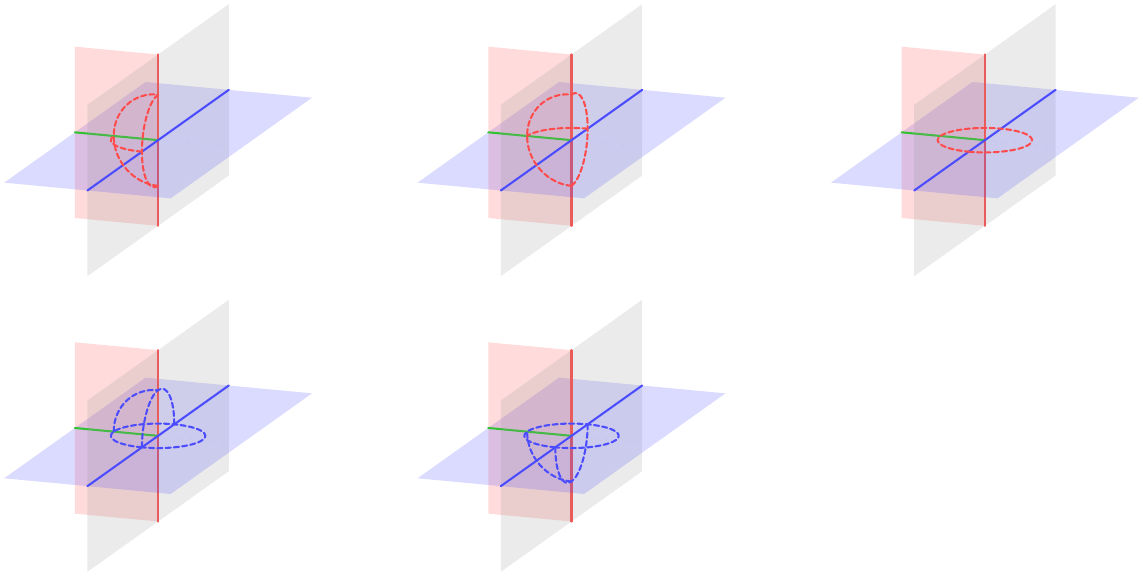}
    \end{figure}

    \item 
    \label{point:7}
    The seventh point defect configuration is similar to the fifth. 
    Here a $yz$-layer interacts with an $xz$-layer that an $xy$-layer ends on from above. 
    Again this requires a nontrivial $\hat{y}$ line defect to end on the point from above. 
    There are 5 half-planes of \TC meeting in this configuration, hence we find 4 independent topological operators: 
    \begin{figure}[H]
        \centering
        \includegraphics[scale = 0.6, page=17]{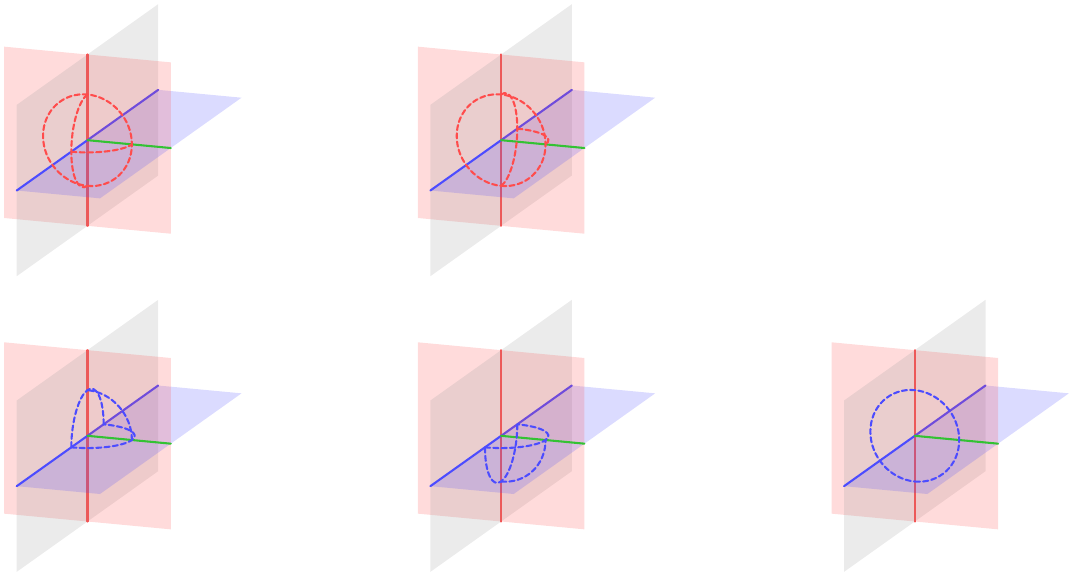}
    \end{figure}

    \item 
    \label{point:8}
    The eighth bulk configuration is simply the reflection of the seventh through an $xz$-plane:
    \begin{figure}[H]
        \centering
        \includegraphics[scale = 0.6, page=18]{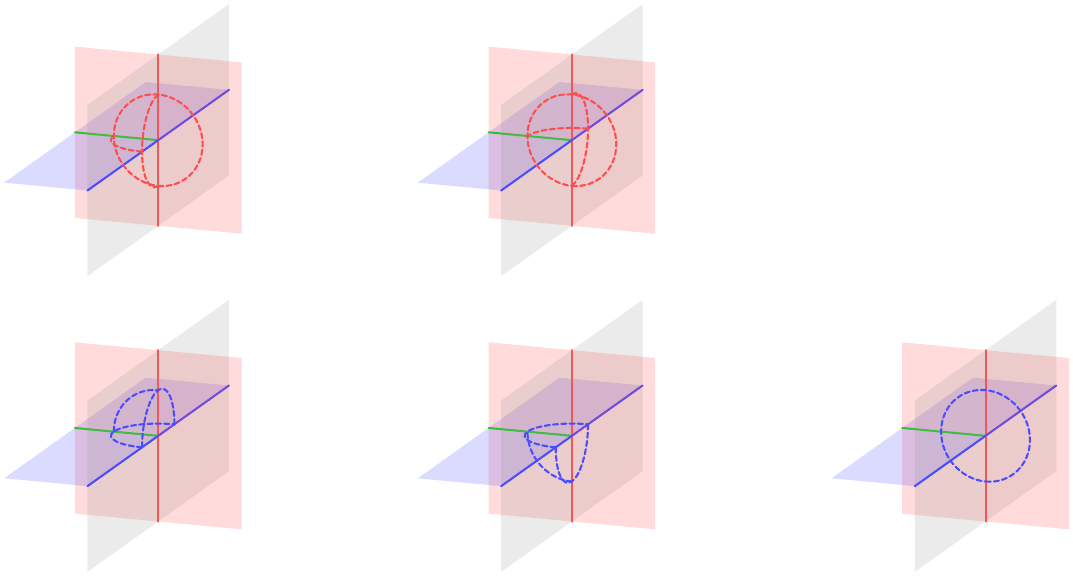}
    \end{figure}

    \item 
    \label{point:9}
    The ninth bulk point defect corresponds to a $yz$-layer that interacts with an $xz$-layer and is joined by a nontrivial defect to an $xy$-layer that simply passes through the $xz$-layer:  
    \begin{figure}[H]
        \centering
        \includegraphics[scale = 0.6, page=70]{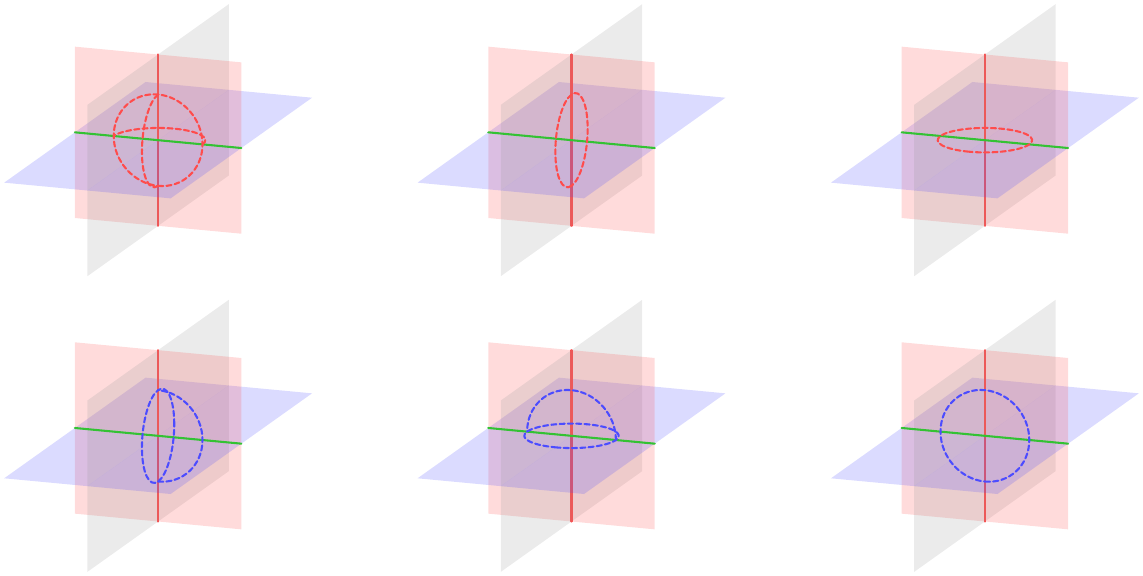}
    \end{figure}

    \item 
    \label{point:10}
    The tenth bulk point defect is similar to the ninth with the role of $xy$ and $yz$-layers reversed: 
    \begin{figure}[H]
        \centering
        \includegraphics[scale = 0.6, page=71]{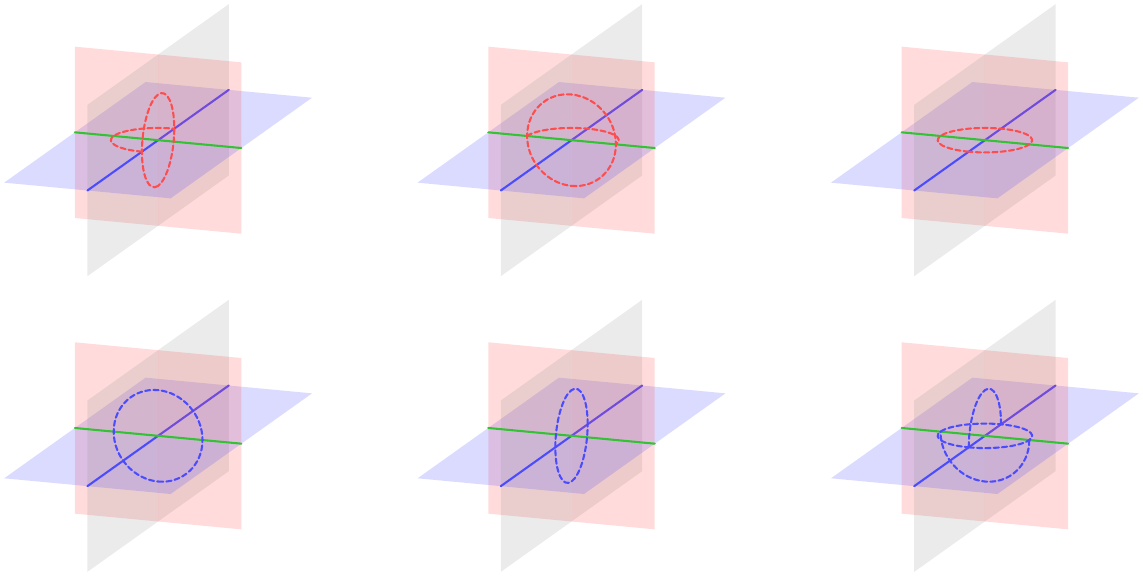}
    \end{figure}
\end{enumerate}
This concludes the enumeration of bulk point defects.

\subsubsection*{Boundary point defects}

We now consider point defects that sit on the boundary of the layer code. 
We remark that the list below does not contain the boundary point defect where two layers cross trivially, as that decomposes into a pair of standard surface code boundaries. 
\begin{enumerate}
    \item 
    \label{point:bdry1}
    The first boundary point defect configuration corresponds to the front $xy$ boundary of a $yz$-layer ending on an $xz$-layer.
    There are two half-planes of \TC meeting the boundary at this point defect and hence we find a pair of independent topological operators:
    \begin{figure}[H]
        \centering
        \includegraphics[scale = 0.5, page=19]{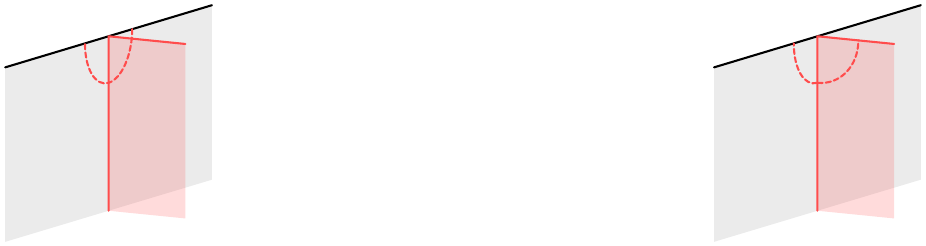}
    \end{figure}
    
    \item 
    \label{point:bdry2}
    The second boundary point defect configuration is a reflection of the first through an $xy$-plane:
    \begin{figure}[H]
        \centering
        \includegraphics[scale = 0.5, page=20]{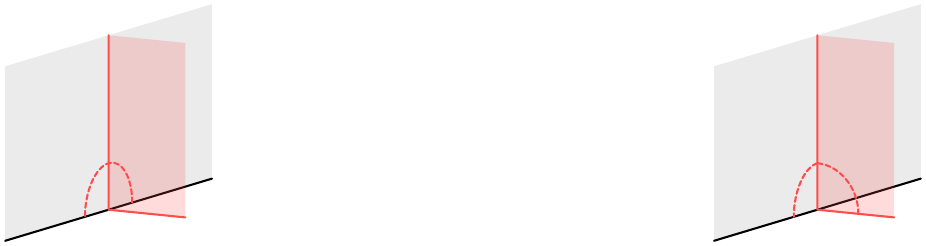}
    \end{figure}
    
    \item 
    \label{point:bdry3}
    The third boundary point defect is the reflection of the first through an $xz$-plane:
    \begin{figure}[H]
        \centering
        \includegraphics[scale = 0.5, page=21]{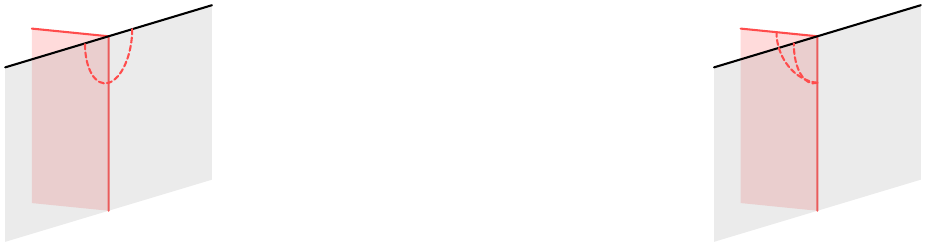}
    \end{figure}
    
    \item 
    \label{point:bdry4}
    The fourth boundary point defect is the reflection of the second through an $xz$-plane:
    \begin{figure}[H]
        \centering
        \includegraphics[scale = 0.5, page=22]{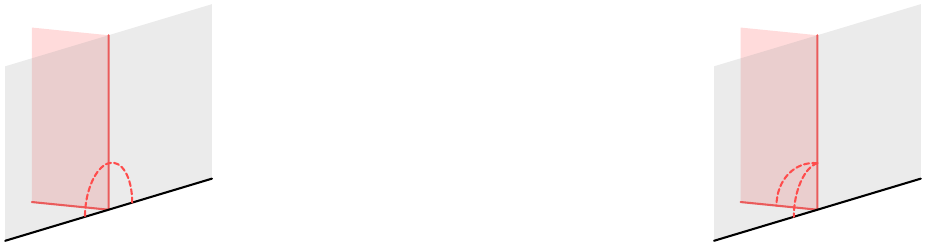}
    \end{figure}
    
    \item 
    \label{point:bdry5}
    The fifth boundary point defect configuration corresponds to the front $xy$ boundary of a $yz$-layer interacting with an $xz$-layer nontrivially.
    There are two half-planes of \TC ending at this point and hence we find a pair of independent topological operators: 
    \begin{figure}[H]
        \centering
        \includegraphics[scale = 0.5, page=23]{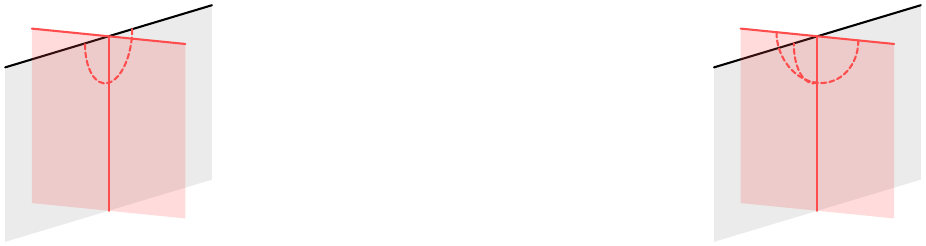}
    \end{figure}
    
    \item 
    \label{point:bdry6}
    The sixth boundary point defect is the reflection of the fifth through an $xy$-plane:
    \begin{figure}[H]
        \centering
        \includegraphics[scale = 0.5, page=24]{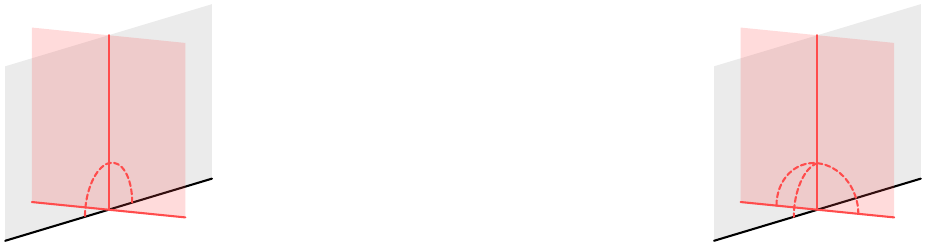}
    \end{figure}
    
    \item 
    \label{point:bdry7}
    The seventh boundary point defect configuration corresponds to the left $yz$ boundary of an $xy$-layer ending on an $xz$-layer. 
    There are two half-planes of \TC meeting at this point and hence we find a pair of independent topological operators: 
    \begin{figure}[H]
        \centering
        \includegraphics[scale = 0.5, page=25]{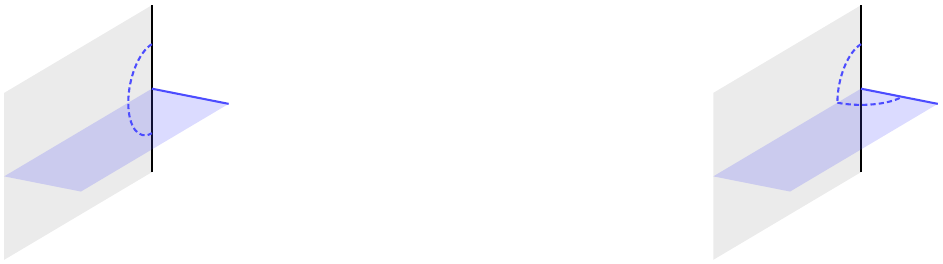}
    \end{figure}
    
    \item 
    \label{point:bdry8}
    The eighth boundary point defect configuration is the reflection of the seventh through an $yz$-plane:
    \begin{figure}[H]
        \centering
        \includegraphics[scale = 0.5, page=26]{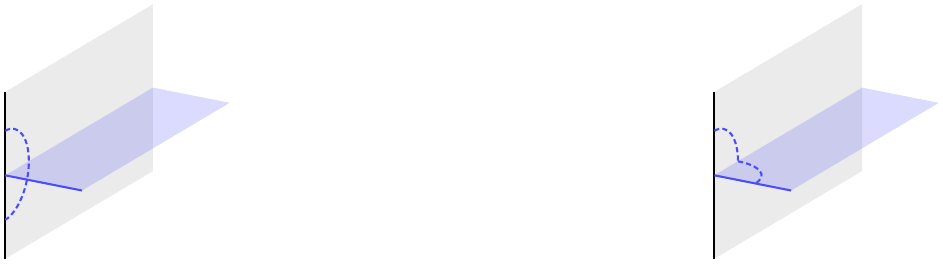}
    \end{figure}
    
    \item 
    \label{point:bdry9}
    The ninth boundary point defect configuration is the reflection of the seventh through an $xz$-plane:
    \begin{figure}[H]
        \centering
        \includegraphics[scale = 0.5, page=27]{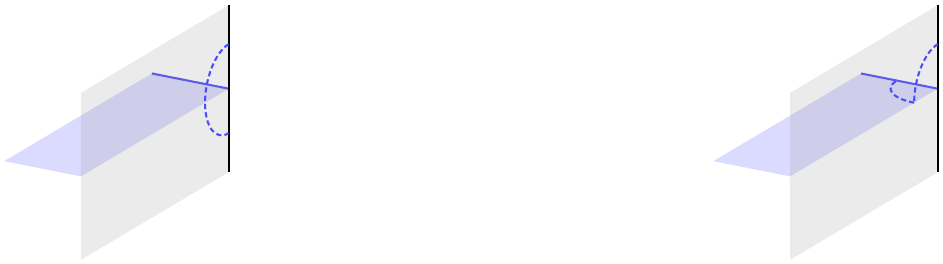}
    \end{figure}
    
    \item 
    \label{point:bdry10}
    The tenth boundary point defect configuration is the reflection of the 8th through an $xz$-plane:
    \begin{figure}[H]
        \centering
        \includegraphics[scale = 0.5, page=28]{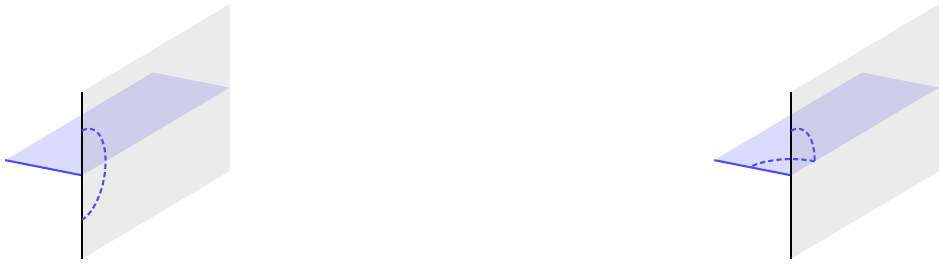}
    \end{figure}
    
    \item 
    \label{point:bdry11}
    The eleventh boundary point defect configuration corresponds to the left $yz$ boundary of an $xy$-layer interacting nontrivially with an $xz$-layer.
    There are two half-planes of \TC meeting the defect and accordingly we find a pair of independent topological operators:
    \begin{figure}[H]
        \centering
        \includegraphics[scale = 0.5, page=29]{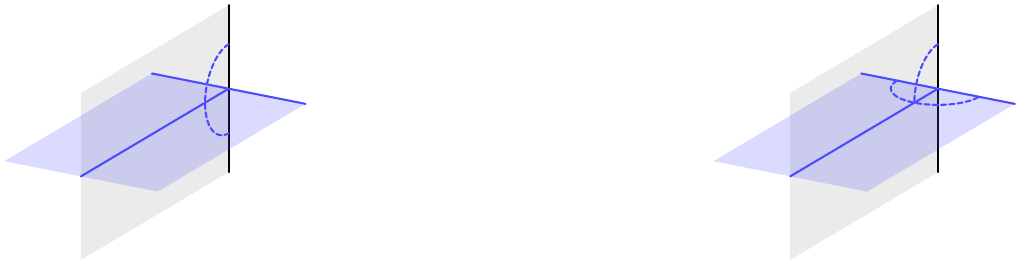}
    \end{figure}
    
    \item 
    \label{point:bdry12}
    The twelfth, and final, boundary point defect confugration is the reflection of the eleventh through a $yz$-plane:
    \begin{figure}[H]
        \centering
        \includegraphics[scale = 0.5, page=30]{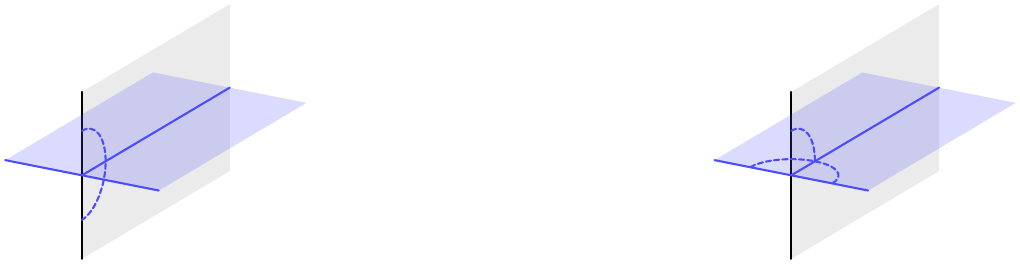}
    \end{figure}
\end{enumerate}

This concludes the enumeration of boundary point defects in the layer code construction.

\section{Code Properties} 
\label{sec:CodeProperties}

In this section we explain some of the essential features of the layer codes that are output by the construction above. 

\subsection{Code Parameters}
\label{sec:CodePropertiesParameters}

First, we explain the code parameters that are achieved by the layer code construction. 
The construction maps the code parameters of an arbitrary input CSS code to an output layer code as follows
\begin{align}
    [[n,k,d]] \mapsto [[\Theta(nn_Xn_Z),k,\Omega(\frac{1}{w}d \min(n_X,n_Z))]].
\end{align}
Here $n_X$ is the number of $X$ checks, $n_Z$ the number of $Z$ checks, and $w$ the maximum weight of the checks in the input code.
The output code is local in three-dimensional space and has checks of weight 6 or less.
When a good CSS LDPC code is taken as input, the output code achieves parameters $[[N,\Theta(N^{\frac{1}{3}}),\Theta(N^{\frac{2}{3}})]]=[[\Theta(L^3),\Theta(L),\Theta(L^2)]]$, where $L$ is the linear extent of the system.
These code parameters are optimal in the sense that their scaling saturates the BPT bound for three spatial dimensions. 

Below we discuss the scaling of the code parameters $[[N,K,D]]$ of the layer code that is produced when our construction is applied to an input code with parameters $[[n,k,d]]$. 

\subsubsection{Number of Physical Qubits: $N=\Theta(nn_Xn_Z)$}
    The scaling of the number of physical qubits in the output code is lower bounded by the number of physical qubits in $n$ layers of surface code with linear dimensions $c n_X \times c n_Z$. 
    Here $c$ is a constant that sets the superlattice spacing between the surface code layers involved in the construction. 
    On the other hand, $N$ is upper bounded by the number of physical qubits in a stack of surface codes along the three lattice directions of a $cn \times cn_X \times cn_Z$ cuboid. 
    
    This gives
    \begin{align}
        2 c^2 n n_X n_Z \leq N \leq 6 c^3 n n_X n_Z
    \end{align}
    and hence $N=\Theta(n n_X n_Z)$. 
    
\subsubsection{Number of Logical Qubits: $K=k$}
    The number of logical qubits in the output code is equal to the number of logical qubits in the input code. 
    To demonstrate this we consider a basis of representatives for the logical operator pairs in the input code 
    \begin{align}
        \Big\{ \overline{X}_i,\overline{Z}_i \,\Big|\, \{\overline{X}_i,\overline{Z}_i\}=0,\, [\overline{X}_i,\overline{Z}_j]=0, \text{ for } i\neq j,\, i,j=1,\dots,k  \Big\},
    \end{align}
    where the $\overline{X}_i$ operators are products of physical $X$ operators, and $\overline{Z}_i$ operators are products of physical $Z$ operators. 
    Here ${\{A,B\}=AB+BA}$ denotes the anticommutator and ${[A,B]=AB-BA}$ denotes the commutator. 
    For each $\overline{X}_i,\overline{Z}_i,$ logical operator pair in the input code there is a corresponding logical operator pair $\widetilde{X}_i,\widetilde{Z}_i,$ in the output layer code.
    
    We define a \textit{concatenated representative} for the $\overline{X}_i$ logical in the layer code by replacing each physical $X$ in the input code with an $X$ string operator across the corresponding $xz$ surface code layer. 
    We remark that this concatenated representative operator would be a logical operator in a straightforward concatenation of surface code with the input code.
    However this concatenated representative is \textit{not} a logical operator in the layer code, as the string operators violate stabilizers on the $yz$ $Z$-check layers that meet the relevant $xz$ qubit surface code layers at nontrivial junctions. 
    The number of such excitations created in each $yz$-layer must be even. This is because the incidence relation of the string operator in each $xz$ surface code layer with the $yz$-layers via nontrivial junctions matches the incidence relation of the corresponding input physical qubit with the input $Z$ checks. Hence the incidence of a concatenated representative operator on a $yz$-layer at nontrivial junctions matches the incidence of the relevant input logical on the associated input $Z$ check, which must be even as they commute.
    We can then pair up the excitations within each $yz$-layer and remove all nontrivial syndromes via $X$ string operators within $yz$-layers to form a logical operator $\widetilde{X}_i$ that we call a \textit{quasiconcatenated representative}. See Figure~\ref{fig:RepetitionCodeLogical} for a depiction of a quasiconcatenated logical operator in a repetition code example. 

    A quasiconcatenated representative for each $\widetilde{Z}_i$ logical is obtained similarly with $yz$-layers replaced by $xy$-layers, and $X$ string operators replaced by $Z$ string operators.
    This construction preserves the (anti-)commutation relations of the logical operators, and guarantees that $K\geq k$. 
    To see this we note that the first step of the process, mapping single qubit Pauli operators to string operators in layers of surface code, preserves commutation relations. 
    The second step of the process, in which excitations are paired up in $yz$ and $xy$-layers, does not affect the commutation relations. 
    This follows from the fact that the $X$ operators used to do the pairing within the $yz$-layers and the $Z$ operators used to do the pairing within the $xy$-layers can be chosen to have disjoint support. 
    
    To complete the proof that $K=k$ we must demonstrate both that there are no additional logical operators that have not been accounted for, and that none of the logical pairs in the output code we have counted are equivalent, up to multiplication with stabilizers.
    These statements follow from the characterization of the logical operators in Section~\ref{sec:Proofs}, summarized in Remark~\ref{rem:LogicalMapping}.

\subsubsection{Distance: $D=\Omega\big(\frac{1}{w}d \min(n_X,n_Z)\big)$}
    We now discuss a lower bound on the weight of the quasiconcatenated representative logical operators. The results in Section~\ref{sec:Proofs} establish that this lower bound holds for the scaling of the code distance, see Corollary~\ref{cor:DX}.
    
    We start by considering the quasiconcatenated representative logical operators $\widetilde{X}_i$ in the output layer code. 
    These operators are obtained from logical operators $\overline{X}_i$ in the input code by replacing each of their physical $X$ operators by $X$ string operators on the corresponding $xz$ surface code layer and applying additional string operators in the $yz$-layers to pair up the excitations that the string operators in the $xz$-layers create. 
    The weight of the quasiconcatenated representative is bounded from below by $d c n_Z$. 
    However, the addition of the $xy$-layers introduce new equivalence relations for these logical operators that allow string operators on $xz$-layers that meet an $xy$-layer at nontrivial junctions to be moved onto the $xy$-layer. 
    In the input code this corresponds to the associated logical operator $\overline{X}_i$ sharing support with part of an $X$ stabilizer check. 
    We assume that we have a minimum weight logical operator representative $\overline{X}_i$. 
    The  common support of this logical operator representative with any $X$ stabilizer must be $\leq \frac{w}{2}$, otherwise there would be an equivalent representative logical operator with lower weight thus contradicting the minimum weight assumption. 
    The worst-case weight reduction via the new equivalence relations due to $xy$-layer mentioned above can only result in the weight being multiplied by $\frac{2}{w}$.
    This corresponds to all of the quasiconcatenated representative's string operators in $xz$-layers being cleaned into $xy$-layers. 
    Hence the weight lower bound retains the same scaling with $n$, given by $\Omega(\frac{2}{w} d c n_Z)$, provided that the input code is LDPC.

    In the case that the input code is not LDPC, we remark that the $X$ distance of the output code still satisfies $\Omega(c n_Z)$. 
    This is simply because the logical operators remain stringlike in the output layer code.

    Analogous statements hold for the $Z$ logical operators of the output code.

    The characterization of logical operators in Section~\ref{sec:Proofs} implies that the lower bound on the logical weight introduced above holds for any equivalent logical operators in the output layer code, see Corollary~\ref{cor:DX}. This accounts for equivalence under all stabilizers in the layer code, going beyond the simple cleaning process outlined above.

\subsection{Energy barrier}
\label{sec:CodePropertiesEnergy}

In this subsection we demonstrate that the logical energy barrier of the output layer code $\Delta_{\text{out}}$ is lower bounded by a ratio of the logical energy barrier of the input code $\Delta_{\text{in}}$ to the product of maximum stabilizer weight $w$ with the maximum number of checks that act on a common qubit $\hat{w}$ in the input code
\begin{align}
    \frac{4  \Delta_{\text{in}}}{ w \hat{w}} \leq \Delta_{\text{out}}.
\end{align}

We first remark that the energy barrier of the output code is upper bounded by the energy barrier of the input code plus 1. 
This follows by considering, without loss of generality, a sequence of Pauli $X$ operators in the input code that attain its energy barrier and implement a nontrivial logical operator.
Each individual Pauli $X$ operator in this sequence can be mapped to a Pauli $X$ string operator in the output code that can be sequentially implemented at the cost of at most one additional unit of energy on top of the energy penalty of the input $X$ operator. 
As Pauli $X$ string operators are sequentially implemented on the layer code, we can maintain a maximum of a single excitation within each $yz$-layer. 
This can be achieved by using string operators to move any excitation within a $yz$-layer to the appropriate location that will immediately annihilate any additional excitation that would be created in that layer via sequential application of an $X$ string in some $xz$-layer. 
This process results in the implementation of a quasiconcatenated logical operator in the layer code with the claimed energy penalty, which is hence an upper bound on the energy barrier $\Delta_{\text{out}}$. 
A similar argument applies for Pauli $Z$ logical operators.

To derive a lower bound on the energy barrier we devise a method to convert any sequence of Pauli $Z$ operators in the layer code into a different sequence that can then be used to define a sequence of Pauli $Z$ operators on the input code. 
This conversion increases the energy penalty multiplicatively by at most a factor $w\hat{w}$, where $w$ is the max weight of a stabilizer check and $\hat{w}$ is the maximum number of checks that act on a common qubit in the input code. 
Hence for input codes that are LDPC, the energy barrier of the output layer code retains the same scaling with $n$ as the energy barrier of the input code. 

We now describe the conversion process for a multi-qubit Pauli $Z$ operator in a layer code.
First, all $e$ excitations within $yz$-layers are transformed into $e$ excitations on $xz$-layers via string operators along the $\hat{y}$ direction. 
There are two choices per $yz$-layer at this step, we take the one that produces the minimum number of excitations. 
This step increases the energy penalty multiplicatively by at most $\frac{w}{2}$. 
Next, we apply string operators to move all $e$ excitations within $xz$-layers to condense at a rough boundary on an external $xy$ face of the cuboid containing the layer code. 
Again there are two choices at this step for each $e$ excitation, we take the choice that produces the minimum number of excitations. 
This step increases the energy penalty multiplicatively by at most $\frac{\hat{w}}{2}$. 
At this point, all remaining $e$ excitations are on $xy$-layers. 
We then apply string operators to bring all $e$ excitations on each $xy$-layer to a common point. 
This step does not increase the energy penalty. 
The results in Section~\ref{sec:ProofsLogicals} establish that the Pauli operator produced by this procedure is equivalent to a quasiconcatenated operator that consists of up to one $Z$ string operator per $xz$-layer, connected to one another and to the remaining $e$ excitations by strings in $xy$-layers. 
Each such quasiconcatenated string operator defines a multi-qubit Pauli $Z$ operator on the input code by mapping each $Z$ string operator in an $xz$-layer to the $Z$ operator on the associated input physical qubit.

Via the above mapping, any sequence of Pauli $Z$ operators on the output code maps to a sequence of Pauli $Z$ operators on the input code. 
For an LDPC code, a local change to the operator on the output code maps to a local change to the operator on the input code\footnote{Here local refers to combinatorial $k$-locality for a constant $k$.}. 
This is because the operator conversion process for a syndrome on any layer only involves other layers that meet the original layer at a nontrivial junction. 
Hence, a local sequence that implements a logical operator on the output code produces a local sequence that implements a logical operator on the input code. 
In doing so the energy penalty on the output code is increased multiplicatively by at most a factor $\frac{1}{4} w \hat{w}$, and hence the energy barrier of the output code must satisfy the lower bound 
\begin{align}
    \frac{4  \Delta_{\text{in}}(n)}{ w \hat{w}} \leq \Delta_{\text{out}}(n) .
\end{align}
Here we are considering the less restrictive definition of the energy barrier that allows a sequence of local Pauli operators to be applied, see Section~\ref{sec:Background}. 
This result implies that the layer code construction preserves the scaling of the energy barrier with $n$ for LDPC codes. 
We remark, however, that since our construction increases the number of physical qubits from $n$ to $\Theta(n n_X n_Z)$ the ration of the energy barrier divided by the number of physical qubits is generically decreased by the layer construction.

\subsection{Stabilizer check relations}

In this subsection we demonstrate that relations between the stabilizer checks of the input code imply similar relations for the output code. 
This implies that the excitations in the output code inherit conservations laws from the excitations in the input code~\cite{kitaev2003fault,Brown2019Parallel,Brown2023Conservation}.

A relation $\mathcal{R}$ amongst $X$ checks on the input code 
\begin{align}
    \prod_{i\in \mathcal{R}} S^X_i = \mathds{1} ,
\end{align}
implies a similar relation on the output code. 
This follows by considering the product of the $X$ checks on $xy$-layers that correspond to the input checks appearing in the relation $\mathcal{R}$. 
The product of all $X$ checks on an $xy$-layer results in a quasiconcatenated representative of the associated $X$ check in the input code, see Figure~\ref{fig:QuasiconcatenatedStab}. 
This operator involves a product of $X$ string operators on each $xz$-layer that corresponds to a qubit in the support of the associated input $X$ check. These string operators are connected by string operators where the $xy$-layer intersects $yz$-layers at nontrivial junctions. 
Hence the product of all $X$ checks on all $xy$-layers involved in an input relation yields an even number of $X$ string operators on each $xz$-layer, together with additional connecting string operators on the $yz$-layers that are intersected nontrivially. 
We then pair up the $X$ string operators in each $xz$-layer, following the ordering of the $xy$-layers they originate from, and multiply by the $X$ checks in the $xz$-layer on the regions enclosed by the paired up strings. 
This leaves no nontrivial operators on the $xz$-layers, due to the smooth boundary conditions on the $yz$ boundaries of the cuboid. 
This step does create additional $X$ string operators on the $yz$-layers that intersect the regions between the paired up strings. 
These string operators, together with the other string operators in the $yz$-layers from the initial step, form $\mathbb{Z}_2$ boundaries of regions within the $yz$-layers. 
We then multiply by $X$ checks within these regions of the $yz$-layers to cancel out the string operator along their boundaries. 
This completes the trivialization of the operator, and hence we have constructed a relation inherited from any relation of the input code. 

A relation in the input code also implies the presence of certain stringlike stabilizers in the output code. 
Consider a relation $\mathcal{R}$ satisfied by $X$ checks
\begin{align}
     \prod_{i\in \mathcal{R}} S^X_i = \mathds{1} .
\end{align}
An $X$ string operator spanning an $xy$-layer in the $\hat{y}$ direction creates an $m$-excitation on each $xz$-layer that meets it at a nontrivial junction. 
The product of these string operators for all $xy$-layers corresponding to $X$ checks in a relation hence produces an even number of $m$-excitations on each $xz$-layer which can then be paired up via additional string operators in the $xz$-layers. 
To demonstrate that this is a stabilizer, we notice that this stringlike operator can be deformed through any $yz$-layer. 
The stringlike operator can hence be moved to a $yz$ boundary of the cuboid and condensed, i.e. trivialized via multiplication with smooth boundary stabilizers.

\section{Proof of Main Results} 
\label{sec:Proofs}

In this section we provide detailed proofs characterizing the structure of logical operators and errors in layer codes and establishing an $\Omega(n)$ energy barrier for a family of good CSS LDPC codes.
These proofs underpin the higher level arguments in Section~\ref{sec:CodeProperties}.

\subsection{Characterization of logical operators and errors}
\label{sec:ProofsLogicals}
In this section we present a number of lemmas that characterize the structure of logical operators and errors that appear in the output layer codes. 
These results are relied upon when establishing the code properties in the previous section. 
We focus our discussion on $X$-type operators, analogous results hold for $Z$-type operators. 

For convenience we introduce the notation $\mathcal{P}_{\text{layer}}^X$ for the set of $X$-type Pauli operators of the layer code, and $\mathcal{P}_{\text{input}}^X$ for the set of $X$-type Pauli operators of the input code. 
We use $\mathcal{N}_\text{layer}^X$ and $\mathcal{N}_\text{input}^X$ to denote the $X$-type normalizers of the layer code and the input code, respectively.
Similarly, we use $\mathcal{S}_\text{layer}^X$ and $\mathcal{S}_\text{input}^X$ to denote the $X$-type stabilizers of the layer code and the input code, respectively.

Our strategy to characterize $X$ logical operators is to consider dividing up elements of $\mathcal{N}_{\text{layer}}^X$ into segments that are each supported within a thin slab that contains a single $yz$-layer. 
The product of all slabs then reproduces the complete logical. 
A similar approach applies to dividing elements of $\mathcal{N}_{\text{layer}}^Z$ into thin slabs that each contain a single $xy$-layer.

\begin{figure}[H]
    \centering
    \includegraphics[page=66, scale=.95]{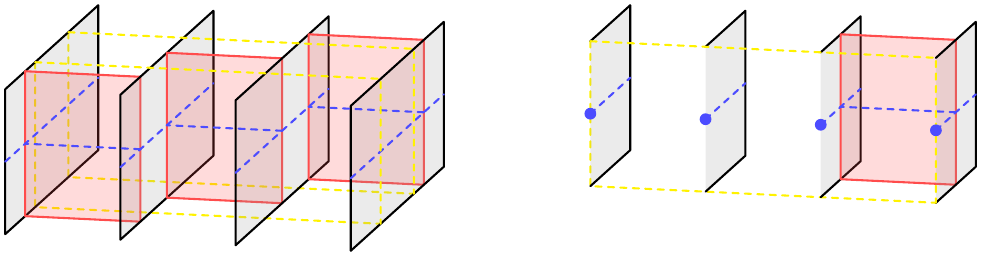}
    \caption{The layer code based on the 4-qubit repetition code divided into $yz$ slabs. On the left, yellow dotted lines indicate the boundaries between slabs that each contain a single $yz$-layer. The blue dotted lines depict a \textit{quasiconcatenated} $X$ logical operator. 
    On the right a single slab is depicted together with boundary excitations (blue circles) that are created by restricting the blue logical operator to the slab.}
    \label{fig:proofs-slabs}
\end{figure}

We now turn to the \textit{boundaries} between pairs of slabs. These slab-boundaries \textit{are not} physical boundaries, they are merely a useful bookkeeping device to track the excitations created by restricted logical operators. 
We first consider sets of $m$-excitations that are supported on the boundary between a pair of slabs. 
One way to create such $m$-excitations is by considering a logical $L \in \mathcal{N}_{\text{layer}}^X$ and a slab, then restricting the operator to that slab, written $L_\text{slab}$, to generate excitations at the boundary of said slab, see Figure \ref{fig:proofs-slabs}. 
We remark that these excitations must fall \emph{between} $yz$-layers. Such collections of $m$-excitations have natural equivalence relations under stringlike $X$ operators supported on the same slab-boundary.
Equivalently, an excitation supported on a slab-boundary can be moved along this boundary and fused with other excitations that live on the boundary in a way that is determined by the connectivity of the $xz$ and $xy$-layers.
We can make use of this equivalence relation to obtain a set of representative excitations that consist of at most a single excitation per $xz$ and $xy$-layer that intersects the slab-boundary. We call such a reduced form an $m$-configuration, see Figure~\ref{fig:m-configuration}.

\begin{definition}[$m$-configuration]
    An $m$-configuration is an equivalence class of sets of $m$-excitations. It is defined by a binary variable for each ${xz}$-layer, and for each segment of $xy$-layer that is bounded by a pair of nontrivial junctions to $xz$-layers. 
    Each binary variable specifies the parity of $m$-excitations on an  $xz$-layer or $xy$-segment. 
\end{definition}

\begin{figure}[H]
    \centering
    \includegraphics[page=69,scale=.9]{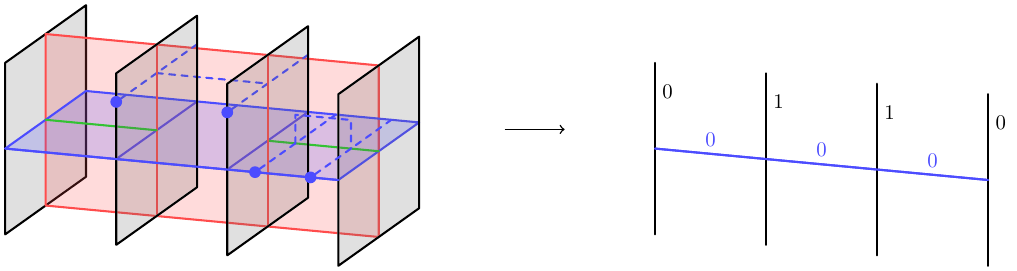}
    \caption{An example of an $m$-configuration on the boundary of a slab from the $[[4,2,2]]$ layer code. 
    On the left the dotted blue lines depict restricted logical operators and blue circles depict $m$-excitations on the slab-boundary. 
    On the right the $m$-confugration corresponding to the boundary excitations is shown. }
    \label{fig:m-configuration}
\end{figure}

Next we consider the equivalence of $m$-configurations induced on the same slab-boundary. 

\begin{definition}[Boundary equivalence]
\label{def:mConfigurationEquivalence}
    A pair of $m$-configurations are boundary-equivalent if they can be related by a sequence of condensations and creations of $m$-anyons at nontrivial junctions. 
\end{definition}

\begin{figure}[H]
    \centering
    \includegraphics[page=67]{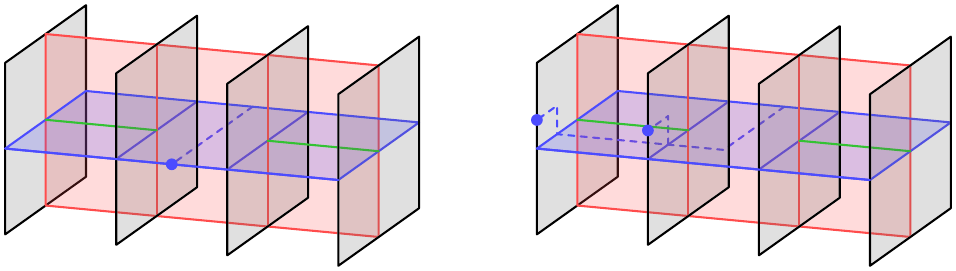}
    \caption{A pair of boundary-equivalent $m$-configurations on the $[[4,2,2]]$ layer code example.
    The $m$-configuration on the right maps to $X_1 X_2$ on the input $[[4,2,2]]$ code.}
    \label{fig:equivalent-configurations}
\end{figure}

\begin{remark}
    For any $m$-configuration there is a boundary-equivalent $m$-configuration that has no $m$-excitations on any $xy$-layers. 
    This is because the line defect junctions between $xy$ and $xz$-layers allow for the excitations to be carried to some $xz$-layer, see Figure \ref{fig:equivalent-configurations}.
\end{remark}

\begin{remark}
\label{rem:OutputInputMapping}
    An $m$-configuration that is supported exclusively on $xz$-layers can be formally mapped to an operator in $\mathcal{P}_{\text{input}}^X$. 
    This mapping simply sends such a configuration of $m$-excitations to a multi-qubit Pauli $X$ operator that is given by a product of single-qubit Pauli $X$ operators on each qubit of the input code that corresponds to an $xy$-layer with an odd number of $m$-excitations on it, see the right-most illustration in Figure \ref{fig:equivalent-configurations}. 
    The set of such $m$-configurations that are boundary-equivalent to the trivial configuration are one-to-one with elements of $\mathcal{S}_{\text{input}}^X$ under this mapping. 
\end{remark}

\begin{remark}
\label{rem:PartialStab}
    An $m$-configuration that is supported on a single $xy$-layer is boundary-equivalent to an $m$-configuration supported only on a set of $xz$-layers. Following Remark \ref{rem:OutputInputMapping}, the new $m$-configuration maps to a partial $X$ stabilizer of the input code and can be chosen such that the resulting operator in $\mathcal{P}_{\text{input}}^X$ has weight $\leq \frac{w}{2}$, see Figure \ref{fig:equivalent-configurations}. 
    Hence, any $m$-configuration of weight $h$ supported only on $xy$-layers is boundary-equivalent to an $m$-configuration supported only on $xz$-layers of weight $\leq \frac{w}{2} h$.
\end{remark}

\begin{remark}
\label{rem:EquivByStabs}
    For any pair of operators $L_a, L_b \in \mathcal{N}_\text{layer}^X$ that induce boundary-equivalent $m$-configurations $a$ and $b$ on a slab-boundary, respectively, there exists a stabilizer $S_{a} \in \mathcal{S}_\text{layer}^X$ such that $L_{a'} = S_{a} L_a$ induces an $m$-configuration $a'$ satisfying $a' = b$ at the same slab-boundary. The stabilizer $S_a$ can be found near the slab-boundary in question. 
    This follows by noting that any cluster of anyons that can condense locally on one side of a slab boundary has a mirror image cluster of anyons that can condense locally on the other side of the slab boundary. 
    The cluster of anyons and its mirror image can be created locally by operators that do not cross the slab-boundary.
    Each anyon in the cluster can then be annihilated with its mirror image via operators that do cross the slab-boundary.
    This process results in a stabilizer supported near the slab-boundary. This stabilizer is $S_a$ if the anyon cluster is chosen to be $a$. 
\end{remark}

The boundary-equivalence relations on $m$-configurations allow $m$-excitations on $xz$-layers to be moved onto $xy$-layers, which can reduce their weight. 
We now consider how much the weight of an $m$-configuration that is initially free of any excitations on $xy$-layers can be reduced by moving its $m$-excitations onto $xy$-layers. 

\begin{lemma}
\label{lem:MinmWeight}
The minimum weight of any boundary-equivalent $m$-configuration that maps to an element $P\in \mathcal{P}_\text{input}^X$ is $\geq\frac{2}{w} d_P$, where $d_P$ is the minimum weight of operators in $P\cdot\mathcal{S}_\text{input}^X$. 
\end{lemma}
\begin{proof}
Consider an $m$-configuration $a$ that maps, in the sense of Remark \ref{rem:OutputInputMapping}, to an element ${P\in \mathcal{P}_\text{input}^X}$, such that $d_P$ the weight of $P$ cannot be reduced by multiplication with elements of $\mathcal{S}_\text{input}$. 
Then $a$ has weight lower bounded by $\geq\frac{2}{w} d_P$.
This is because each $X$ check in the input code has maximum overlap $\leq\frac{w}{2}$ with the minimal weight $P$, otherwise multiplication with that $X$ check would reduce the weight further. 
Hence, the maximum cleaning that might be possible sends sets of $\frac{w}{2}$ $m$-excitations on disjoint $xz$-layers to single $m$-excitations on disjoint $xy$-layers. 
We now turn to other $m$-configurations that are boundary-equivalent to $a$. 
This includes the $m$-configurations that map to each element of $P\cdot\mathcal{S}_\text{input}^X$. 
From the weight bound in Remark~\ref{rem:PartialStab} it follows that any $m$-configuration with weight $h_{xy}$ on $xy$-layers and weight $h_{xz}$ on $xz$-layers is boundary-equivalent to an $m$-configuration of weight $\leq h_{xz}+\frac{w}{2}h_{xy}$ that is supported only on $xz$-layers. 
Due to the mapping from Remark \ref{rem:OutputInputMapping} this upper bound on the weight cannot be smaller than $d_P$ and hence $h_{xz}+\frac{w}{2}h_{xy}\geq d_P$. 
Without loss of generality we can assume $w\geq 2$.
From this we see that $h_{xy}+h_{xz}\geq \frac{2}{w} d_P$ for all $m$-configurations that are boundary-equivalent to $a$.
\end{proof}

Next, we turn our attention to the relationship between the $m$-configurations supported on either side of a slab -- as depicted in Figure \ref{fig:equivalent-configurations}.
\begin{lemma}
\label{lem:MSlabBdry}
    Consider an $X$ operator in a slab that creates no excitations on the $xy$-layers and $xz$-layers in the bulk of that slab. 
    The $m$-configuration created by this operator on one boundary of the slab is boundary-equivalent to the configuration created on the other boundary. 
\end{lemma}
\begin{proof}
    For the purposes of this argument we pick an abstract $xy$-\textit{plane} through the slab that does not coincide with a physical $xy$-layer, see Figure~\ref{fig:xyplane}. 
    Our aim is to clean the $X$ operator onto this $xy$-plane. 
    First, any excitations within the $yz$-layer contained in the slab can be moved onto the intersection of the $yz$-layer with the chosen $xy$-plane via $X$ string operators supported on the $yz$-layer. 
    This process creates no further excitations. 
    The modified $X$ operator can then be cleaned onto the single chosen $xy$-plane by multiplication with stabilizers in the layer code.
    The resulting cleaned operator is a product of string operators whose support is contained in the intersection of the chosen $xy$-plane with the $xz$-layers and the $yz$-layer in the slab, see Figure~\ref{fig:xyplane} for an example.
    Since the cleaned operator creates no excitations on the $xz$-layers within the slab, the checks from a given $xz$-layer force the operator to either have no support on that $xz$-layer or have support along the whole line of intersection between that $xz$-layer and the chosen $xy$-plane.
    This results in all $m$-excitations on one of the slab boundaries having a partner with the same $yz$ coordinates on the other boundary. 
\end{proof}

\begin{figure}[H]
    \centering
    \includegraphics[page=72]{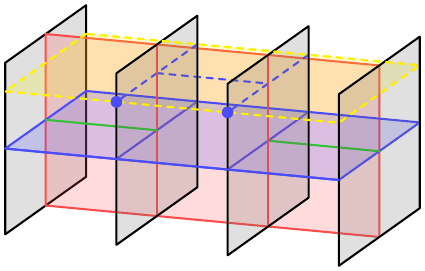}
    \caption{An example of an $X$ operator (blue dashed lines) that has been cleaned onto an $xy$-plane (yellow). We remark that the $xy$-plane is distinct from the $xy$-layer (blue) in the figure.  }
    \label{fig:xyplane}
\end{figure}

\begin{remark}
    We remark that the parity of the $m$-excitations left on the $yz$-layer in the slab matches the anticommutation relation of the associated many-qubit Pauli $X$ operator with the relevant $Z$ check in the input code, see Remark~\ref{rem:OutputInputMapping}.  
\end{remark}

\begin{lemma}
    Consider an $m$-configuration induced by an element of $\mathcal{N}_\text{layer}^X$. Then it has to be boundary-equivalent to an $m$-configuration that maps\footnote{The mapping is described in Remark~\ref{rem:OutputInputMapping}.} to an element of $\mathcal{N}_\text{input}^X$.
\end{lemma}
\begin{proof}
We begin by remarking that a straightforward mapping of $X$ operators on the input code to string operators on $xy$-planes that run across a slab results in an operator that creates a number of $m$-excitations on the $yz$-layer whose parity matches the stabilizer syndrome of the corresponding input $Z$ check. 
These $m$-excitations can all be fused at a common location in the $yz$-layer via string operators, producing at most a single $m$-excitation.

For there to be no $m$-excitation, the $m$-configuration on the boundary of that slab must be in a boundary-equivalence class that maps, in the sense of Remark~\ref{rem:OutputInputMapping}, to an $X$ operator on the input code that commutes with the relevant $Z$ check. 

Hence the restriction of an element of $\mathcal{N}_\text{layer}^X$ to any slab must create a configuration of excitations at the slab-boundary that corresponds to an element of $\mathcal{N}_\text{input}^X$. 
This is because the element of $\mathcal{N}_\text{layer}^X$ must create no excitation on any slab. 
\end{proof} 

To complete the characterization of the distance, we demonstrate below that an element of $\mathcal{N}_\text{layer}^X \setminus \mathcal{S}_\text{layer}^X$ must have support on all slabs.

\begin{figure}[H]
    \centering
    \includegraphics[page=68, scale=.95]{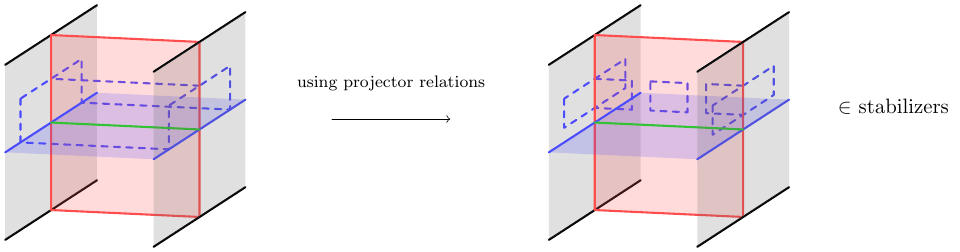}
    \caption{On the left we depict a slab from a layer code supporting an operator from $\mathcal{N}_\text{layer}^X$ depicted as a blue dashed line. 
    On the right we depict an equivalent element of $\mathcal{N}_\text{layer}^X$ that has been cleaned off the $xy$-layer via multiplication with stabilizers in the slab.}
    \label{fig:slab-is-stab}
\end{figure}

\begin{lemma}
\label{lem:slab-is-stab}
    An element of $\mathcal{N}_\text{layer}^X$ that is supported wholly within a slab must be in $\mathcal{S}_\text{layer}$. 
\end{lemma}
\begin{proof}
    An element $L \in \mathcal{N}_\text{layer}^X$ supported on a single slab is a product of string operators on the layers within the slab, joined at the nontrivial junctions. 
    The string operators can be cleaned off the $xy$-layers by multiplication with stabilizers, see Figure~\ref{fig:slab-is-stab}. 
    Within $xz$-layers and the $yz$-layer, $m$ string operators can be deformed along the $z$-direction via multiplication with stabilizers. 
    Hence $L$ can be cleaned onto a single $xy$-plane via stabilizer multiplication. 
    This must result in a trivial operator given the support and the fact that no excitations are created in the bulk or on the boundary of the slab.
\end{proof}

\begin{corollary}
    An element of $\mathcal{N}_\text{layer}^X \setminus \mathcal{S}_\text{layer}^X$ must have nontrivial support on all slabs. 
\end{corollary}
\begin{proof}
    Consider an operator $L \in \mathcal{N}_\text{layer}^X$ that acts trivially on some slab. 
    It can be written as a product of disjoint operators $L_\text{left}, L_\text{right} \in \mathcal{N}_\text{layer}^X$ on either side of that slab. By assumption both $L_\text{left}, L_\text{right},$ induce an empty $m$-configuration on the boundaries of that slab, i.e.~there are no $m$-excitations on these slabs. By Lemma~\ref{lem:MSlabBdry}, the excitations created by $L_\text{left}, L_\text{right},$ on all other slabs have to be boundary-equivalent with the empty $m$-configuration.
    
    Hence the $X$ operator within each slab can be multiplied by an $X$ operator on each slab-boundary to create a logical operator supported within the slab. 
    The resulting operator must be a stabilizer by Lemma \ref{lem:slab-is-stab}.
    Hence the $X$ logical, multiplied by some $X$ operators on the boundaries between slabs that create no excitations, is a product of stabilizers. 
    The additional $X$ operators between slabs are themselves stabilizers, which follows by considering them alone and shifting the boundaries of the slabs by a small amount, sufficient to bring the boundary operators wholly into single slabs. 
    Hence $L$ itself must be a product of stabilizers.
\end{proof}

\begin{remark}
\label{rem:LogicalMapping}
    The above lemmas establish a one-to-one correspondence between elements of $\mathcal{N}_\text{layer}^X$ and $\mathcal{N}_\text{input}^X$. 
    Mapping from $\mathcal{N}_\text{input}^X$ to $\mathcal{N}_\text{layer}^X$ can be achieved via the quasiconcatenated logical operators introduced in the previous section. 
    Mapping from $\mathcal{N}_\text{layer}^X$ to $\mathcal{N}_\text{input}^X$ can be achieved by truncating the layer code logical operator and mapping the configuration of $m$-excitations it creates at the boundary of a slab to an input $X$ logical following Remark~\ref{rem:OutputInputMapping}.
\end{remark}

\begin{figure}[H]
    \centering
    \includegraphics[page=73]{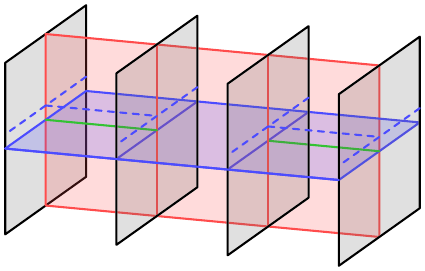}
    \caption{ A quasiconcatenated $X$-type stabilizer check (dashed blue lines) is given by the product of local $X$-type stabilizer checks on the $xy$-plane (blue). }
    \label{fig:QuasiconcatenatedStab}
\end{figure}

\begin{remark}
    The quasiconcatenated stabilizer checks play a special role in the mapping from $\mathcal{N}_\text{input}^X$ to $\mathcal{N}_\text{layer}^X$. 
    This is because the purpose of the $xy$-layers is to make the inference of quasiconcatenated stabilizer checks possible. 
    A quasiconcatenated check is equal to the product of $X$-type stabilizer checks on the corresponding $xy$-layer, See Figure~\ref{fig:QuasiconcatenatedStab}.
    Hence the measured value of the quasiconcatenated check can be inferred from local check measurements on the $xy$-layer. 
    Alternatively, the measured value of a \textit{concatenated} stabilizer check can be inferred by measuring all single qubit $X$ operators on the associated $xy$-layer. 
    This results in a deformation of the layer code. 
\end{remark}

\begin{corollary}
\label{cor:DX}
    The $X$ distance $d_X$ of a family of layer codes based on an $[[n,k,d]]$ family of input codes satisfies $D_X=\Omega\big(\frac{1}{w} d n_Z\big) $. 
    Here $n_Z$ is the number of $Z$ checks in the input code and $w$ is the maximum check weight of the input code. 
\end{corollary}
\begin{proof}
    The $n_Z$ factor in the scaling of the logical $X$ distance for the layer codes follows from the fact that nontrivial logical $X$ operators must have support on all slabs.
    The $\frac{1}{w}d$ factor in the scaling of the logical $X$ distance follows from the implication of Lemma~\ref{lem:MinmWeight} that restricting any representative of a nontrivial logical $X$ operator to a single slab results in a configuration of $m$-excitations on each boundary of that slab with minimum weight $\frac{2}{w} d$. 
    Moving $\frac{2}{w} d$ excitations into, and across, the slab from different points on the boundary requires an operator with weight $\Omega \big( \frac{2}{w} d \big) $.
\end{proof}

\begin{remark}
    The above distance scaling bound is not always tight. 
    This is particularly relevant for layer codes based on non-LDPC codes. 
    In particular, for non-LDPC codes whose maximum check weight grows faster than their distance, the simple lower bound $D_X=\Omega\big(n_Z\big)$ is more useful for the resulting layer code. 
    This bound follows from the fact that the minimum weight of any $m$-configuration in the boundary-equivalence class of an element $L \in \mathcal{N}_\text{input}^X\setminus \mathcal{S}_\text{input}^X$ must be at least $1$ even for non-LDPC codes. 
    If this bound is saturated the logical operators are stringlike. 
\end{remark}

\subsection{Energy barrier for the Leverrier-Zemor codes}
\label{sec:ProofsLZBarrier}

In this section we show that there exist good CSS LDPC codes with a linear energy barrier. 
For this purpose we focus on the Leverrier-Zemor (LZ) codes~\cite{leverrier2022quantum}.
Since $X$ and $Z$ operators can be treated separately in CSS codes, we use the notation $y \in \{0,1\}^n$ to specify an $n$-qubit $X$ or $Z$ operator depending on the context, $| \cdot |$ refers to the Hamming norm, and $\oplus$ denotes binary addition. 
We write $H_Z, H_X$ the $Z$ and $X$ parity-check matrices. Similarly we write $L_Z$ and $L_X$ the nontrivial logical operators of the code.
We state our argument explicitly for $X$ logical operators, and the $Z$ case follows identically. 
We start by giving a formal definition to the strict version of the energy barrier described in Section \ref{sec:intro-energy-barrier}.

\begin{definition}[Energy barrier]
    A code is said to have an energy barrier $\Delta$ if for any $y \in L_X$ and for any sequence $(y_i)_{i=1}^{i=l}$ of single-qubit operators such that $\bigoplus_{i=1}^{i=l} y_i = y$ there exists an index $e \leq l$ such that $|H_Z y'| \geq \Delta$, where $y' = \bigoplus_{i=1}^{i=e} y_i$, and similarly for $L_Z,H_X$.
\end{definition}

The LZ codes have the striking property that if $y$ has a small syndrome then either it is very far from, or very close to, being a trivial operator. The notion of being ``far" from the codespace is captured by the $| \cdot |_{S_X}$ (or $| \cdot |_{S_Z}$) norm:
\[
    | y |_{S_X} =  \min_{s \in S_X} |y \oplus s|
\]
and similarly for $| \cdot |_{S_Z}$.

It is easy to verify that for $\hat{x} \in \{0,1\}^{n}$, with $|\hat{x}| \leq 1$, then 
\[
\left | | y |_{S_X} - | y \oplus \hat{x}|_{S_X}\right | \leq 1,
\]
which we refer to as the \emph{continuity} of $| \cdot |_{S_X}$ .

\begin{lemma}[Property 1 of Ref.~\cite{anshu2023nlts}]
\label{lemma:lz}
    For the codes defined in Ref.~\cite{leverrier2022quantum} there exist constants $c_1,c_2,\delta_0$ such that for any $ 0 \leq \delta < \delta_0$, if ${c_1 \delta n \leq |y|_{S_X}  \leq c_2 n}$, then $|H_Z y| > \delta n_Z$, where $n_Z$ is the number of $Z$ checks.
\end{lemma}

\begin{lemma}
    The Leverrier-Zemor codes have an energy barrier scaling as $\Omega(n)$.
\end{lemma}
\begin{proof}
Consider an arbitrary sequence $(y_i)_{i=1}^{i=l}$ of single-qubit operators such that $\bigoplus_{i=1}^{i=l} y_i = y$ for some nontrivial logical operator $y$.  Note that by the definition of the distance, there exists an index $e$ such that
\begin{align}
    \left |\bigoplus_{i=1}^{i=e} y_i \right |_{S_X} &= \lfloor \min(c_2 n, d/2) \rfloor . 
\end{align}

We will write $y' = \bigoplus_{i=1}^{i=e} y_i $.
Given that $d\in \Omega(n)$, then there exists a constant $\delta' > 0$ independent of $n$ such that $c_1 \delta' n <  \lfloor \min(c_2 n, d/2) \rfloor $. That gives us
\[
 c_1 \delta' n \leq \left|y' \right|_{S_X} \leq c_2 n
\]
This can be inserted into Lemma~\ref{lemma:lz}, and we verify that 
\[
\left |H_Z \ y' \right| \in \Omega(n)
\]
\end{proof}

\section{Lattice Model} 
\label{sec:LatticeModel}

In order to build a physical model for the layer code construction introduced in Section~\ref{sec:LayerCodeConstruction}, we need to specify explicit checks for each topological defect that was used. 
To the best of our knowledge, there is no convenient ``algorithmic" way to obtain these lattice terms.
An effective general heuristic is to map each projector in the relevant annulus algebra to a set of stabilizer generators.
For layers of surface code, this process turns out to be simple and intuitive, see Figure \ref{fig:proj-to-term} for example.
We remark that these checks can be chosen so that the resulting layer code is CSS.

\begin{figure}[H]
    \centering
    \includegraphics[page=14]{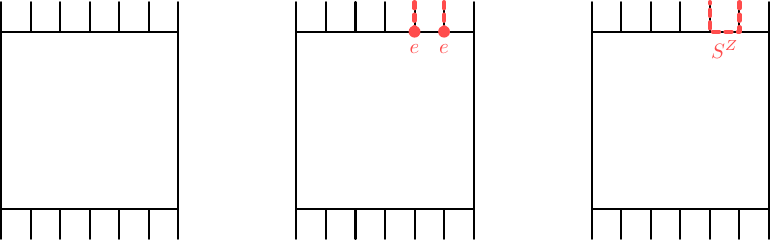}
    \caption{
    Illustration of how condensation rules can be used to derive lattice stabilizer checks. From left to right, a pair of $e$ anyons is created and then fused together, resulting in the operator in the rightmost figure. This operator creates no excitations, and therefore should be included in the stabilizer group.
    }
    \label{fig:proj-to-term}
\end{figure}

\subsection{Layer checks}

The stabilizer checks away from defects and boundaries are given by the canonical surface code checks, see Section~\ref{sec:Background}. 
The surface code layers have rough boundary conditions where they meet the $xy$-boundary of the cuboid that bounds the layer code. 
Similarly, they have smooth boundary conditions where they meet the $yz$-boundary of the cuboid, see Figure~\ref{fig:SCLayerBoundaries}. 
This choice of boundary conditions is such that the $\overline{X}$ logical operators are oriented along the $\hat{x}$ direction and similarly $\overline{Z}$ logical operators are oriented along $\hat{z}$, as depicted in Figure~\ref{fig:surface-code-bdry}.

\begin{figure}[H]
    \centering
    \includegraphics[scale = 1, page=59]{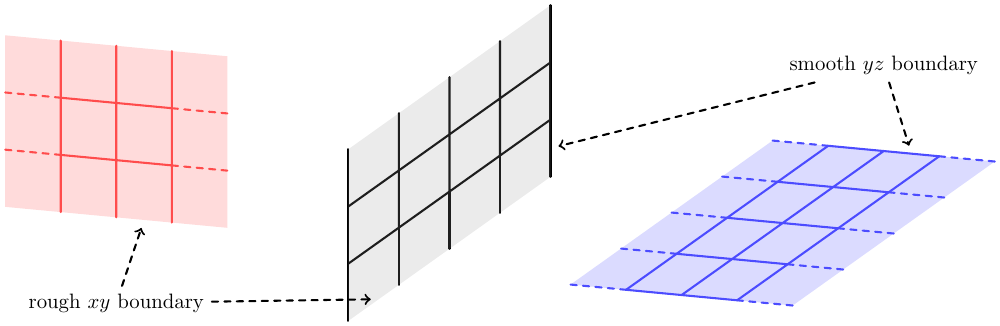}
    \caption{A $yz$-layer on the left with rough $xy$-boundaries. An $xz$-layer in the middle with both rough $xy$-boundaries and smooth $yz$-boundaries. An $xy$-layer on the right with smooth $yz$-boundaries.}
    \label{fig:SCLayerBoundaries}
\end{figure}

\begin{remark}
\label{rem:LatticeChecks}
    Below we explicitly list the novel lattice checks that appear when the layer code construction is applied to an arbitrary input CSS code. 
    For clarity, we list the minimal number of novel checks required to specify each defect. 
    For this reason we \textit{do not} list the checks near a defect that have already appeared when specifying the adjacent layers.
    In particular, we do not list the standard checks of the surface code layers and their rough and smooth boundaries below. 
    Furthermore, in the point defect subsection~\ref{sec:PointDefectChecks} we do not list the checks on adjacent line defects that have already appeared in the line defect subsection~\ref{sec:LineDefectChecks}.
    We remark that the novel checks at each defect replace some of the checks that would otherwise appear in the adjacent surface code layers. 
\end{remark}

\subsection{Line defect checks}
\label{sec:LineDefectChecks}

We now turn to the implementation of the topological line defects. 
Below, we list only those checks that differ from the standard surface code stabilizers. There are 3 types of $\hat{x}$-oriented junctions, 2 types of $\hat{y}$-oriented junctions and 3 types of $\hat{z}$-oriented junctions.  
\begin{enumerate}

    \item The first, \textit{trivial}, type of $\hat{y}$-junction is simply a crossing an ${xy}$ and a ${yz}$ surface code layer with no coupling between the layers. See Figure~\ref{fig:LineDefects}~(a). The layers meet along a line of $\hat{y}$ edges. 
    The stabilizer checks for the two layers separate and are given by the canonical surface code terms. 
    \begin{figure}[H]
        \centering
        \includegraphics[scale = 1, page=31]{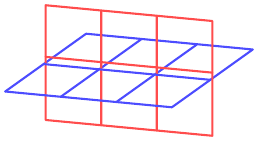}
    \end{figure}
    
    \item The second, \textit{nontrivial}, type of $\hat{y}$-junction has modified stabilizer checks that couple an ${xy}$ and a ${yz}$ surface code layer. See 
    Figure~\ref{fig:LineDefects}~(b).
    \begin{figure}[H]
        \centering
        \includegraphics[scale = 1, page=32]{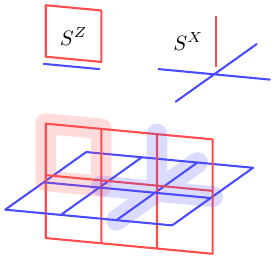}
    \end{figure}
    In this figure, and those below, the $S^Z$ stabilizer on the region highlighted in red has support on the edges that are depicted adjacent to it. The $S^X$ stabilizer on the region highlighted in blue is specified similarly. 
    
    \item The \emph{first} type of $\hat{z}$-junction occurs where a $yz$-layer meets its first $xz$-layer, see Figure~\ref{fig:k-body-Z-checks}. 
    The star checks are modified along the junction.
    \begin{figure}[H]
        \centering
        \includegraphics[scale = 1, page=33]{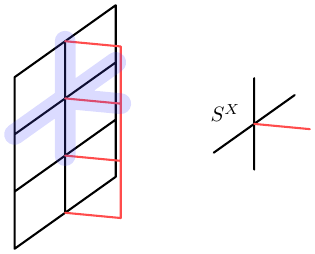}
    \end{figure}
    
    \item The second, \textit{middle}, type of $\hat{z}$-junction occurs where a $yz$-layer intersects an $xz$-layer, see Figure~\ref{fig:k-body-Z-checks}. 
    The star and plaquette checks are modified along the junction. 
    \begin{figure}[H]
        \centering
        \includegraphics[scale = 1, page=34]{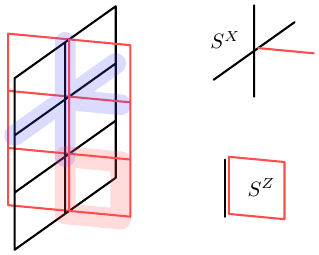}
    \end{figure}
    
    \item The third, \textit{last}, type of $\hat{z}$-junction occurs where a $yz$-layer meets its last $xz$-layer, see 
    Figure~\ref{fig:k-body-Z-checks}. 
    It is simply a reflection of the first type of $\hat{z}$-junction above. 
    \begin{figure}[H]
        \centering
        \includegraphics[scale = 1, page=35]{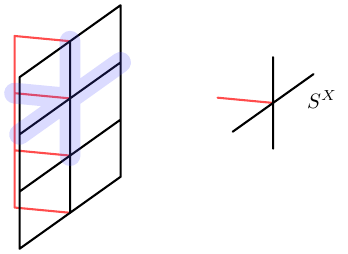}
    \end{figure}
    
    \item The \emph{first} type of $\hat{x}$-junction occurs where an $xy$-layer meets its first $xz$-layer, see 
    Figure~\ref{fig:k-body-X-checks}. 
    The star and plaquette checks are modified as follows along the junction.
    \begin{figure}[H]
        \centering
        \includegraphics[scale = 1, page=36]{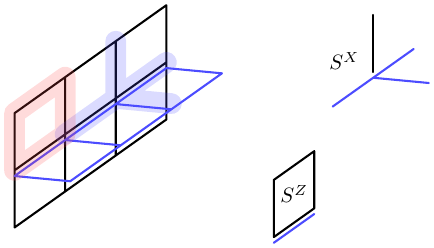}
    \end{figure}
    
    \item The second, \textit{middle}, type of $\hat{x}$-junction occurs where an $xy$-layer intersects an $xz$-layer, see Figure~\ref{fig:k-body-X-checks}. 
    The star and plaquette checks are modified along the junction. 
    \begin{figure}[H]
        \centering
        \includegraphics[scale = 1, page=37]{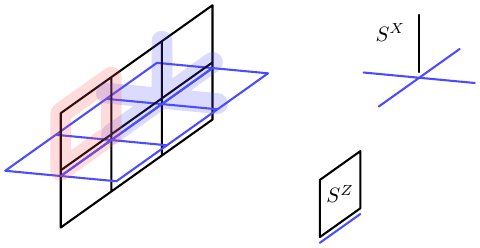}
    \end{figure}
    
    \item The third, \textit{last}, type of $\hat{x}$-junction occurs where an $xy$-layer meets its last $xz$-layer, see Figure~\ref{fig:k-body-X-checks}. 
    \begin{figure}[H]
        \centering
        \includegraphics[scale = 1, page=38]{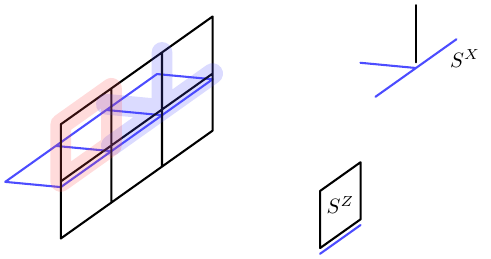}
    \end{figure}
    
\end{enumerate}
This concludes the enumeration of line defect lattice checks. 

\subsection{Point Defect Checks}

\label{sec:PointDefectChecks}

In addition to the above line defects, there are 10 types of point defects in the bulk and 12 types of point defects on the boundary of the cuboid. 

\subsubsection*{Bulk Point Defects}

The 10 point defects in the bulk are listed below.

\begin{enumerate}
    \item The first type of bulk point defect~(\ref{point:1}) occurs at the intersection of a first $\hat{x}$-junction, a first $\hat{z}$-junction and a nontrivial $\hat{y}$-junction. 
    The modified star checks at the point defect are shown directly below.
    \begin{figure}[H]
        \centering
        \includegraphics[scale = 1, page=39]{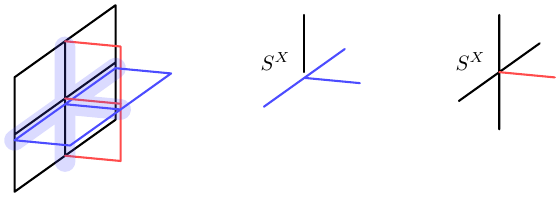}
    \end{figure}
    
    \item  The second type of bulk point defect~(\ref{point:2}) occurs at the intersection of a last $\hat{x}$-junction, a last $\hat{z}$-junction, and a nontrivial $\hat{y}$-junction. 
    This point defect is simply the reflection of the above through an $xz$-plane. 
    The modified star checks are shown directly below.
    \begin{figure}[H]
        \centering
        \includegraphics[scale = 1, page=40]{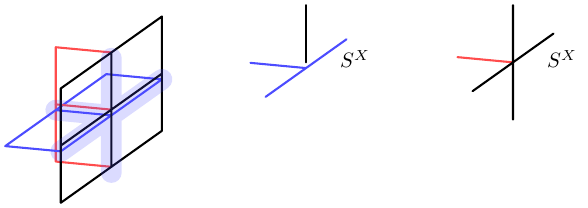}
    \end{figure}
    
    \item The third type of bulk point defect~(\ref{point:3}) occurs at the intersection of a middle $\hat{x}$-junction, a middle $\hat{z}$-junction, a trivial $\hat{y}$-junction from below and a nontrivial $\hat{y}$-junction from above.
    The modified checks are shown directly below. 
    \begin{figure}[H]
        \centering
        \includegraphics[scale = 1, page=41]{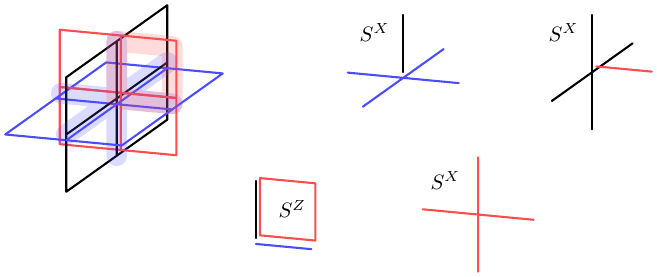}
    \end{figure}
  
    \item The fourth type of bulk point defect~(\ref{point:4}) occurs at the intersection of a middle $\hat{x}$-junction, a middle $\hat{z}$-junction, a nontrivial $\hat{y}$-junction from below and a trivial $\hat{y}$-junction from above.
    This junction only differs from the above by the exchange of the trivial and nontrivial $\hat{y}$-junctions. 
    The modified checks are shown directly below. 
    \begin{figure}[H]
        \centering
        \includegraphics[scale = 1, page=42]{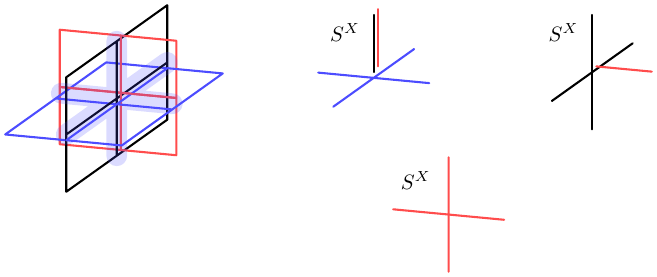}
    \end{figure}
    
    \item The fifth type of bulk point defect~(\ref{point:5}) occurs at the intersection of a middle $\hat{x}$-junction, a first $\hat{z}$-junction, and a nontrivial $\hat{y}$-junction from above. 
    The modified star checks on the point defect are shown directly below. 
    \begin{figure}[H]
        \centering
        \includegraphics[scale = 1, page=43]{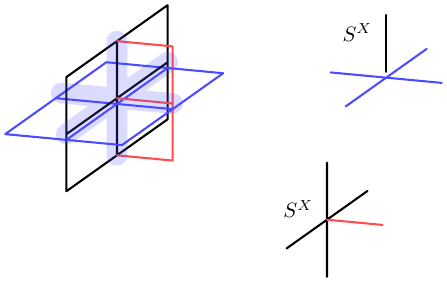}
    \end{figure}
    
    \item The sixth type of bulk point defect~(\ref{point:6}) occurs at the intersection of a middle $\hat{x}$-junction, a last $\hat{z}$-junction, and a nontrivial $\hat{y}$-junction from below. 
    The modified star checks on the point defect are shown directly below. 
    \begin{figure}[H]
        \centering
        \includegraphics[scale = 1, page=44]{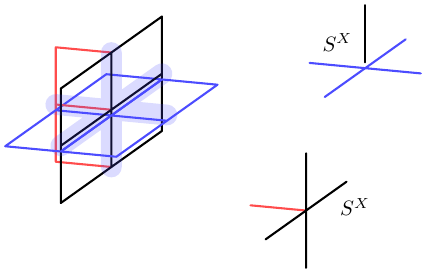}
    \end{figure}
    
    \item The seventh type of bulk point defect~(\ref{point:7}) occurs at the intersection of a first $\hat{x}$-junction, a middle $\hat{z}$-junction, and a nontrivial $\hat{y}$-junction from above. 
    The modified checks on the point defect are shown directly below. 
    \begin{figure}[H]
        \centering
        \includegraphics[scale = 1, page=45]{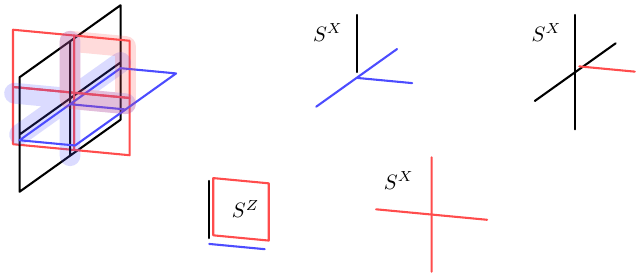}
    \end{figure}
    
    \item The eighth type of bulk point defect~(\ref{point:8}) occurs at the intersection of a last $\hat{x}$-junction, a middle $\hat{z}$-junction, and a nontrivial $\hat{y}$-junction from below. 
    The modified checks on the point defect are shown directly below.  
    \begin{figure}[H]
        \centering
        \includegraphics[scale = 1, page=46]{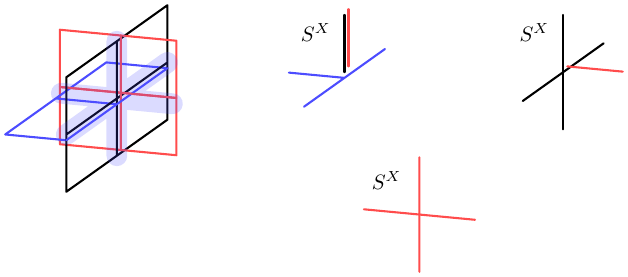}
    \end{figure}
    
    \item The ninth type of bulk point defect~(\ref{point:9}) occurs at the intersection of a trivial crossing $\hat{x}$-junction, a nontrivial middle $\hat{z}$-junction, and a nontrivial $\hat{y}$-junction that passes through the point defect. 
    \begin{figure}[H]
        \centering
        \includegraphics[scale = 1, page=75]{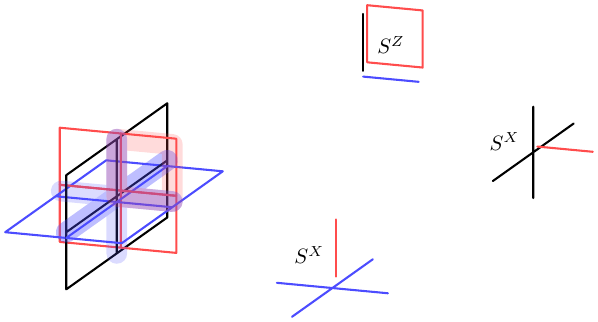}
    \end{figure}
    
    \item The tenth type of bulk point defect~(\ref{point:10}) is similar to the ninth type with the roles of $\hat{x}$ and $\hat{z}$ reversed. 
    \begin{figure}[H]
        \centering
        \includegraphics[scale = 1, page=74]{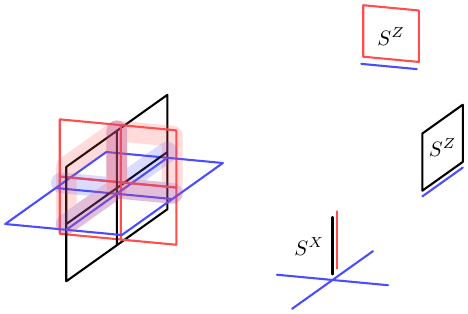}
    \end{figure}
\end{enumerate}

\subsubsection*{Boundary Point Defects}

The 12 boundary point defects are listed below.
\begin{enumerate}
    \item The first type of boundary point defect~(\ref{point:bdry1}) occurs where a first $\hat{z}$-junction meets the front $yz$ boundary. 
    The stabilizer checks are given by standard surface code terms in the presence of a rough boundary and the trivalent junction introduced above. 
    \begin{figure}[H]
        \centering
        \includegraphics[scale = 1, page=47]{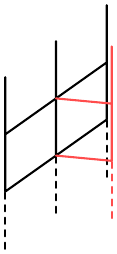}
    \end{figure}
    
    \item The second type of boundary point defect~(\ref{point:bdry2}) occurs where a first $\hat{z}$-junction meets the back $yz$ boundary. 
    The stabilizer checks are given by standard surface code terms in the presence of a rough boundary and the trivalent junction introduced above. 
    This is simply the reflection of the above point defect through a $yz$-plane. 
    \begin{figure}[H]
        \centering
        \includegraphics[scale = 1, page=48]{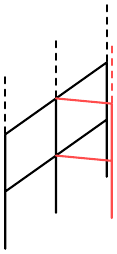}
    \end{figure}
    
    \item The third type of boundary point defect~(\ref{point:bdry3}) occurs where a last $\hat{z}$-junction meets the front $yz$ boundary. 
    The stabilizer checks are given by standard surface code terms in the presence of a rough boundary and the trivalent junction introduced above. 
    This is simply the reflection of the first boundary point defect through an $xz$-plane.
    \begin{figure}[H]
        \centering
        \includegraphics[scale = 1, page=49]{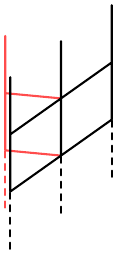}
    \end{figure}
    
    \item The fourth type of boundary point defect~(\ref{point:bdry4}) occurs where a last $\hat{z}$-junction meets the back $yz$ boundary. 
    The stabilizer checks are given by standard surface code terms in the presence of a rough boundary and the trivalent junction introduced above. 
    This is simply the reflection of the second boundary point defect through an $xz$-plane.
    \begin{figure}[H]
        \centering
        \includegraphics[scale = 1, page=50]{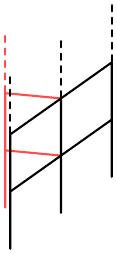}
    \end{figure}
    
    \item The fifth type of boundary point defect~(\ref{point:bdry5}) occurs where a middle $\hat{z}$-junction meets the front $yz$ boundary. 
    The stabilizer checks are given by standard surface code terms in the presence of a rough boundary and the middle $\hat{z}$-junction introduced above, together with a modified plaquette term shown directly below.
    \begin{figure}[H]
        \centering
        \includegraphics[scale = 1, page=51]{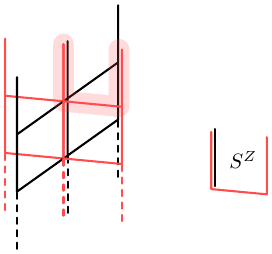}
    \end{figure}
    
    \item The sixth type of boundary point defect~(\ref{point:bdry6}) occurs where a middle $\hat{z}$-junction meets the back $yz$ boundary. 
    This is given by the reflection of the fifth boundary point defect through an $xy$-plane. 
    The stabilizer checks are given by standard surface code terms in the presence of a rough boundary and the middle $\hat{z}$-junction introduced above, together with a modified plaquette term shown directly below. 
    \begin{figure}[H]
        \centering
        \includegraphics[scale = 1, page=52]{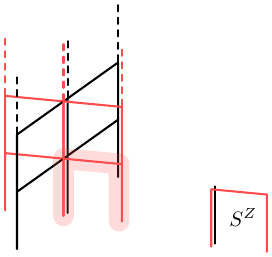}
    \end{figure}
    
    \item The seventh type of boundary point defect~(\ref{point:bdry7}) occurs where a first $\hat{x}$-junction meets the left $yz$ boundary. 
    The stabilizer checks are given by standard surface code terms with a smooth boundary and the $\hat{x}$-junction introduced above, together with modified boundary star checks. 
    \begin{figure}[H]
        \centering
        \includegraphics[scale = 1, page=53]{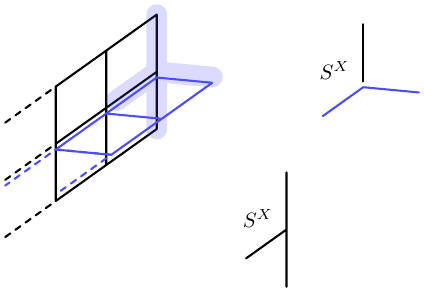}
    \end{figure}
    
    \item The eighth type of boundary point defect~(\ref{point:bdry8}) occurs where a first $\hat{x}$-junction meets the right $yz$ boundary. 
    This is simply the reflection of the above point defect through the $yz$-plane. 
    The stabilizer checks are given by standard surface code terms with a smooth boundary and the $\hat{x}$-junction introduced above, together with modified boundary star checks shown directly below. 
    \begin{figure}[H]
        \centering
        \includegraphics[scale = 1, page=54]{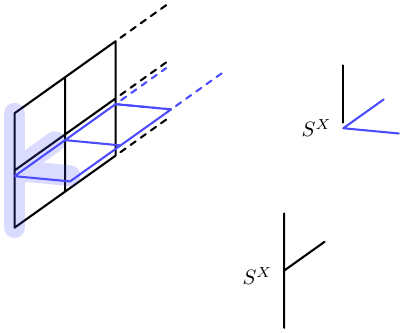}
    \end{figure}
    
    \item The ninth type of boundary point defect~(\ref{point:bdry9}) occurs where a last $\hat{x}$-junction meets the left $yz$ boundary. 
    The stabilizer checks are given by standard surface code terms with a smooth boundary and the $\hat{x}$-junction introduced above, together with modified boundary star checks shown directly below.
    \begin{figure}[H]
        \centering
        \includegraphics[scale = 1, page=55]{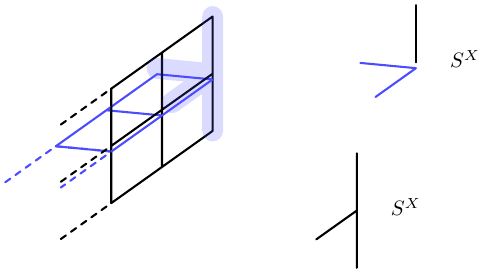}
    \end{figure}
    
    \item The tenth type of boundary point defect~(\ref{point:bdry10}) occurs where a last $\hat{x}$-junction meets the right $yz$ boundary. 
    This is given by reflecting the above point defect through the $yz$-plane. 
    The stabilizer checks are given by standard surface code terms with a smooth boundary and the $\hat{x}$-junction introduced above, together with modified boundary star checks shown directly below. 
    \begin{figure}[H]
        \centering
        \includegraphics[scale = 1, page=56]{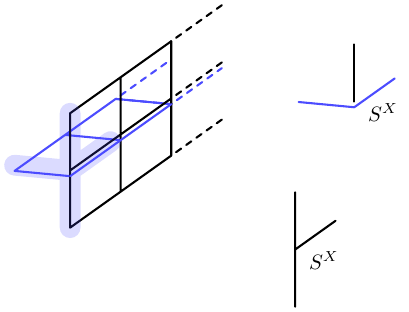}
    \end{figure}
    
    \item The eleventh type of boundary point defect~(\ref{point:bdry11}) occurs where a middle $\hat{x}$-junction meets the left $yz$ boundary. 
    The stabilizer checks are given by standard surface code terms with a smooth boundary and the middle $\hat{x}$-junction introduced above, together with modified boundary star checks shown directly below. 
    \begin{figure}[H]
        \centering
        \includegraphics[scale = 1, page=57]{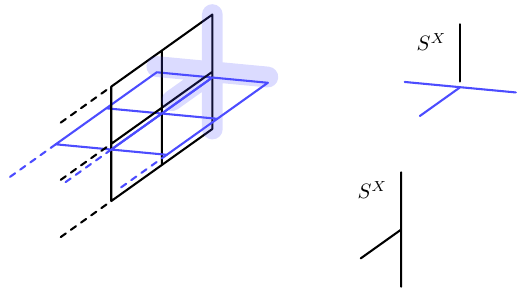}
    \end{figure}
    
    \item The twelfth type of boundary point defect~(\ref{point:bdry12}) occurs where a middle $\hat{x}$-junction meets the right $yz$ boundary. 
    This is given by reflecting the above point defect through a $yz$-plane. 
    The stabilizer checks are given by standard surface code terms with a smooth boundary and the middle $\hat{x}$-junction introduced above, together with modified boundary star checks shown directly below. 
    \begin{figure}[H]
        \centering
        \includegraphics[scale = 1, page=58]{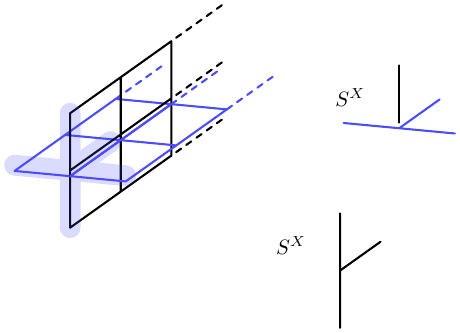}
    \end{figure}
\end{enumerate}

This concludes the enumeration of point defect lattice checks. We now proceed to show that for any logical operator $L \in\mathcal{N}_{\text{layer}}$ whose support overlaps with a ball of constant radius around any of the point defects, there exists an operator $K \in \mathcal{S}_{\text{layer}}$ such that $L' = SL$ no longer overlaps with that ball. When a point defect satisfies this condition, it is said to be correctable \cite{bravyi2009no, bravyi2010tradeoffs, flammia2017limits}. 

It was shown in Ref.~\cite{delfosse2013upper} that a subset $A$ of qubits is correctable if 
\begin{align}
    2 |A| = rank (H) + rank(H|_A) + rank(H|_{\overline{A}}),
\end{align}
where $H$ is the code's parity-check matrix, $H|_A$ is its restriction to the qubits in $A$, and $\overline{A}$ is the complement of $A$. Using this relation it can be verified that when introducing a point defect that preserves the number of qubits and the number of linearly independent checks in its neighborhood, then the correctability of the region of space affected by the point defect is preserved. We can leverage this property into a guarantee that the point defects we introduce are all correctable. Hence, they can be ignored in our discussion of the code properties such as the distance and the energy barrier.

\section{Discussion} 
\label{sec:Discussion}

In this work we have introduced a construction that takes as input an arbitrary $[[n,k,d]]$ CSS stabilizer code and produces as output a layer code, which is a $[[\Theta(nn_Xn_Z),k,\Omega(\frac{1}{w}dn)]]$ CSS stabilizer code that is local in three-dimensional space.
Here $n_X$ denotes the number of $X$ checks, $n_Z$ the number of $Z$ checks, and $w$ the maximum weight of the checks in the input code.
Each layer code has checks of weight 6 or less, and takes the form of a topological defect network. 
The name layer code reflects that these codes are a special variety of topological defect network that is built exclusively from layers of surface code, one for each physical qubit and stabilizer check of the input code, joined together by line defects with incidence relations inherited from the Tanner graph of the input code.  

Applying our construction to a family of good quantum CSS LDPC codes with parameters $[[n,\Theta(n),\Theta(n)]]]$ produces a family of layer codes with parameters $[[\Theta(n^3),\Theta(n),\Theta(n^2)]]$, which is optimal for a family of three-dimensional topological codes. 
While our construction applies to general input CSS codes that are not necessarily LDPC, we can only guarantee favorable scaling of the output code's distance when the maximum stabilizer weight of the input code is constant. 

We have shown that the layer codes output by our construction preserve the scaling of the energy barrier of an input LDPC code. 
Hence any layer code family that is based on a family of input codes with a polynomial energy barrier also has a polynomial energy barrier. 
We proved that the LDPC codes introduced in Ref.~\cite{leverrier2022quantum} have a linear energy barrier $\Theta(n)$, where $n$ is the number of physical qubits. 
Hence there exist layer codes with $\Theta(n^3)$ physical qubits and energy barrier $\Theta(n)$.

\acknowledgements

We thank Shin Ho Choe and Libor Caha for useful comments on the first version of this preprint. We are also grateful to Daniel Litinski and Anthony Leverrier for their precious insights when this project was still in its infancy.
DW's work on this project was supported by the Australian Research Council Discovery Early Career Research Award (DE220100625).

\bibliography{references.bib}

\begin{thebibliography}{134}%
\makeatletter
\providecommand \@ifxundefined [1]{%
 \@ifx{#1\undefined}
}%
\providecommand \@ifnum [1]{%
 \ifnum #1\expandafter \@firstoftwo
 \else \expandafter \@secondoftwo
 \fi
}%
\providecommand \@ifx [1]{%
 \ifx #1\expandafter \@firstoftwo
 \else \expandafter \@secondoftwo
 \fi
}%
\providecommand \natexlab [1]{#1}%
\providecommand \enquote  [1]{``#1''}%
\providecommand \bibnamefont  [1]{#1}%
\providecommand \bibfnamefont [1]{#1}%
\providecommand \citenamefont [1]{#1}%
\providecommand \href@noop [0]{\@secondoftwo}%
\providecommand \href [0]{\begingroup \@sanitize@url \@href}%
\providecommand \@href[1]{\@@startlink{#1}\@@href}%
\providecommand \@@href[1]{\endgroup#1\@@endlink}%
\providecommand \@sanitize@url [0]{\catcode `\\12\catcode `\$12\catcode
  `\&12\catcode `\#12\catcode `\^12\catcode `\_12\catcode `\%12\relax}%
\providecommand \@@startlink[1]{}%
\providecommand \@@endlink[0]{}%
\providecommand \url  [0]{\begingroup\@sanitize@url \@url }%
\providecommand \@url [1]{\endgroup\@href {#1}{\urlprefix }}%
\providecommand \urlprefix  [0]{URL }%
\providecommand \Eprint [0]{\href }%
\providecommand \doibase [0]{https://doi.org/}%
\providecommand \selectlanguage [0]{\@gobble}%
\providecommand \bibinfo  [0]{\@secondoftwo}%
\providecommand \bibfield  [0]{\@secondoftwo}%
\providecommand \translation [1]{[#1]}%
\providecommand \BibitemOpen [0]{}%
\providecommand \bibitemStop [0]{}%
\providecommand \bibitemNoStop [0]{.\EOS\space}%
\providecommand \EOS [0]{\spacefactor3000\relax}%
\providecommand \BibitemShut  [1]{\csname bibitem#1\endcsname}%
\let\auto@bib@innerbib\@empty
\bibitem [{\citenamefont {Feynman}(1982)}]{Feynman1982}%
  \BibitemOpen
  \bibfield  {author} {\bibinfo {author} {\bibfnamefont {R.~P.}\ \bibnamefont
  {Feynman}},\ }\bibfield  {title} {\bibinfo {title} {{Simulating physics with
  computers}},\ }\bibfield  {journal} {\bibinfo  {journal} {International
  Journal of Theoretical Physics}\ }\textbf {\bibinfo {volume} {21}},\ \href
  {https://doi.org/10.1007/BF02650179} {10.1007/BF02650179} (\bibinfo {year}
  {1982})\BibitemShut {NoStop}%
\bibitem [{\citenamefont {Manin}(1980)}]{manin1980computable}%
  \BibitemOpen
  \bibfield  {author} {\bibinfo {author} {\bibfnamefont {Y.}~\bibnamefont
  {Manin}},\ }\bibfield  {title} {\bibinfo {title} {Computable and
  uncomputable},\ }\href@noop {} {\bibfield  {journal} {\bibinfo  {journal}
  {Sovetskoye Radio, Moscow}\ }\textbf {\bibinfo {volume} {128}},\ \bibinfo
  {pages} {28} (\bibinfo {year} {1980})}\BibitemShut {NoStop}%
\bibitem [{\citenamefont {Benioff}(1980)}]{Benioff1980}%
  \BibitemOpen
  \bibfield  {author} {\bibinfo {author} {\bibfnamefont {P.}~\bibnamefont
  {Benioff}},\ }\bibfield  {title} {\bibinfo {title} {{The computer as a
  physical system: A microscopic quantum mechanical Hamiltonian model of
  computers as represented by Turing machines}},\ }\bibfield  {journal}
  {\bibinfo  {journal} {Journal of Statistical Physics}\ }\textbf {\bibinfo
  {volume} {22}},\ \href {https://doi.org/10.1007/BF01011339}
  {10.1007/BF01011339} (\bibinfo {year} {1980})\BibitemShut {NoStop}%
\bibitem [{\citenamefont {Deutsch}(1985)}]{Deutsch1985}%
  \BibitemOpen
  \bibfield  {author} {\bibinfo {author} {\bibfnamefont {D.}~\bibnamefont
  {Deutsch}},\ }\bibfield  {title} {\bibinfo {title} {{QUANTUM THEORY, THE
  CHURCH-TURING PRINCIPLE AND THE UNIVERSAL QUANTUM COMPUTER.}},\ }\bibfield
  {journal} {\bibinfo  {journal} {Proceedings of The Royal Society of London,
  Series A: Mathematical and Physical Sciences}\ }\textbf {\bibinfo {volume}
  {400}},\ \href {https://doi.org/10.1098/rspa.1985.0070}
  {10.1098/rspa.1985.0070} (\bibinfo {year} {1985})\BibitemShut {NoStop}%
\bibitem [{\citenamefont {AI}\ and\ \citenamefont
  {Collaborators}(2019)}]{Arute2019}%
  \BibitemOpen
  \bibfield  {author} {\bibinfo {author} {\bibfnamefont {G.~Q.}\ \bibnamefont
  {AI}}\ and\ \bibinfo {author} {\bibnamefont {Collaborators}},\ }\bibfield
  {title} {\bibinfo {title} {{Quantum supremacy using a programmable
  superconducting processor}},\ }\bibfield  {journal} {\bibinfo  {journal}
  {Nature}\ }\textbf {\bibinfo {volume} {574}},\ \href
  {https://doi.org/10.1038/s41586-019-1666-5} {10.1038/s41586-019-1666-5}
  (\bibinfo {year} {2019})\BibitemShut {NoStop}%
\bibitem [{\citenamefont {Shor}(1995)}]{Shor1995}%
  \BibitemOpen
  \bibfield  {author} {\bibinfo {author} {\bibfnamefont {P.~W.}\ \bibnamefont
  {Shor}},\ }\bibfield  {title} {\bibinfo {title} {{Scheme for reducing
  decoherence in quantum computer memory}},\ }\href
  {https://doi.org/10.1103/PhysRevA.52.R2493} {\bibfield  {journal} {\bibinfo
  {journal} {Physical Review A}\ }\textbf {\bibinfo {volume} {52}},\ \bibinfo
  {pages} {R2493} (\bibinfo {year} {1995})}\BibitemShut {NoStop}%
\bibitem [{\citenamefont {Steane}(1996{\natexlab{a}})}]{Steane1996}%
  \BibitemOpen
  \bibfield  {author} {\bibinfo {author} {\bibfnamefont {A.~M.}\ \bibnamefont
  {Steane}},\ }\bibfield  {title} {\bibinfo {title} {{Error correcting codes in
  quantum theory}},\ }\bibfield  {journal} {\bibinfo  {journal} {Physical
  Review Letters}\ }\textbf {\bibinfo {volume} {77}},\ \href
  {https://doi.org/10.1103/PhysRevLett.77.793} {10.1103/PhysRevLett.77.793}
  (\bibinfo {year} {1996}{\natexlab{a}})\BibitemShut {NoStop}%
\bibitem [{\citenamefont {Shor}(1996)}]{shor1996fault}%
  \BibitemOpen
  \bibfield  {author} {\bibinfo {author} {\bibfnamefont {P.~W.}\ \bibnamefont
  {Shor}},\ }\bibfield  {title} {\bibinfo {title} {Fault-tolerant quantum
  computation},\ }in\ \href@noop {} {\emph {\bibinfo {booktitle} {Proceedings
  of 37th Conference on Foundations of Computer Science}}}\ (\bibinfo
  {organization} {IEEE},\ \bibinfo {year} {1996})\ pp.\ \bibinfo {pages}
  {56--65}\BibitemShut {NoStop}%
\bibitem [{\citenamefont {Gottesman}(1997)}]{gottesman1997stabilizer}%
  \BibitemOpen
  \bibfield  {author} {\bibinfo {author} {\bibfnamefont {D.}~\bibnamefont
  {Gottesman}},\ }\bibfield  {title} {\bibinfo {title} {Stabilizer codes and
  quantum error correction},\ }\href@noop {} {\bibfield  {journal} {\bibinfo
  {journal} {arXiv preprint quant-ph/9705052}\ } (\bibinfo {year}
  {1997})}\BibitemShut {NoStop}%
\bibitem [{\citenamefont {Aharonov}\ and\ \citenamefont
  {Ben-Or}(1997)}]{aharonov1997fault}%
  \BibitemOpen
  \bibfield  {author} {\bibinfo {author} {\bibfnamefont {D.}~\bibnamefont
  {Aharonov}}\ and\ \bibinfo {author} {\bibfnamefont {M.}~\bibnamefont
  {Ben-Or}},\ }\bibfield  {title} {\bibinfo {title} {Fault-tolerant quantum
  computation with constant error},\ }in\ \href@noop {} {\emph {\bibinfo
  {booktitle} {Proceedings of the twenty-ninth annual ACM symposium on Theory
  of computing}}}\ (\bibinfo {organization} {ACM},\ \bibinfo {year} {1997})\
  pp.\ \bibinfo {pages} {176--188}\BibitemShut {NoStop}%
\bibitem [{\citenamefont {Knill}\ \emph {et~al.}(1998)\citenamefont {Knill},
  \citenamefont {Laflamme},\ and\ \citenamefont
  {Zurek}}]{knill1998resilientQC}%
  \BibitemOpen
  \bibfield  {author} {\bibinfo {author} {\bibfnamefont {E.}~\bibnamefont
  {Knill}}, \bibinfo {author} {\bibfnamefont {R.}~\bibnamefont {Laflamme}},\
  and\ \bibinfo {author} {\bibfnamefont {W.~H.}\ \bibnamefont {Zurek}},\
  }\bibfield  {title} {\bibinfo {title} {Resilient quantum computation},\
  }\href@noop {} {\bibfield  {journal} {\bibinfo  {journal} {Science}\ }\textbf
  {\bibinfo {volume} {279}},\ \bibinfo {pages} {342} (\bibinfo {year}
  {1998})}\BibitemShut {NoStop}%
\bibitem [{\citenamefont {Kitaev}(1997)}]{kitaev1997quantum}%
  \BibitemOpen
  \bibfield  {author} {\bibinfo {author} {\bibfnamefont {A.~Y.}\ \bibnamefont
  {Kitaev}},\ }\bibfield  {title} {\bibinfo {title} {Quantum computations:
  algorithms and error correction},\ }\href@noop {} {\bibfield  {journal}
  {\bibinfo  {journal} {Russian Mathematical Surveys}\ }\textbf {\bibinfo
  {volume} {52}},\ \bibinfo {pages} {1191} (\bibinfo {year}
  {1997})}\BibitemShut {NoStop}%
\bibitem [{\citenamefont {Preskill}(1998)}]{preskill1998reliable}%
  \BibitemOpen
  \bibfield  {author} {\bibinfo {author} {\bibfnamefont {J.}~\bibnamefont
  {Preskill}},\ }\bibfield  {title} {\bibinfo {title} {{Reliable quantum
  computers}},\ }\href {https://doi.org/10.1098/rspa.1998.0167} {\bibfield
  {journal} {\bibinfo  {journal} {Proceedings of the Royal Society A:
  Mathematical, Physical and Engineering Sciences}\ }\textbf {\bibinfo {volume}
  {454}},\ \bibinfo {pages} {385} (\bibinfo {year} {1998})},\ \Eprint
  {https://arxiv.org/abs/9705031} {arXiv:9705031 [quant-ph]} \BibitemShut
  {NoStop}%
\bibitem [{\citenamefont {Preskill}(1997)}]{Preskill1997fault}%
  \BibitemOpen
  \bibfield  {author} {\bibinfo {author} {\bibfnamefont {J.}~\bibnamefont
  {Preskill}},\ }\bibfield  {title} {\bibinfo {title} {{Fault-tolerant quantum
  computation}},\ }in\ \href {https://doi.org/10.1109/SFCS.1996.548464} {\emph
  {\bibinfo {booktitle} {Proceedings of 37th Conference on Foundations of
  Computer Science}}}\ (\bibinfo  {publisher} {WORLD SCIENTIFIC},\ \bibinfo
  {year} {1997})\ pp.\ \bibinfo {pages} {56--65},\ \Eprint
  {https://arxiv.org/abs/9712048} {arXiv:9712048 [quant-ph]} \BibitemShut
  {NoStop}%
\bibitem [{\citenamefont {Kitaev}(2003)}]{kitaev2003fault}%
  \BibitemOpen
  \bibfield  {author} {\bibinfo {author} {\bibfnamefont {A.~Y.}\ \bibnamefont
  {Kitaev}},\ }\bibfield  {title} {\bibinfo {title} {Fault-tolerant quantum
  computation by anyons},\ }\href@noop {} {\bibfield  {journal} {\bibinfo
  {journal} {Annals of Physics}\ }\textbf {\bibinfo {volume} {303}},\ \bibinfo
  {pages} {2} (\bibinfo {year} {2003})}\BibitemShut {NoStop}%
\bibitem [{\citenamefont {Bravyi}\ and\ \citenamefont
  {Kitaev}(1998)}]{bravyi1998quantum}%
  \BibitemOpen
  \bibfield  {author} {\bibinfo {author} {\bibfnamefont {S.}~\bibnamefont
  {Bravyi}}\ and\ \bibinfo {author} {\bibfnamefont {A.~Y.}\ \bibnamefont
  {Kitaev}},\ }\bibfield  {title} {\bibinfo {title} {Quantum codes on a lattice
  with boundary},\ }\href@noop {} {\bibfield  {journal} {\bibinfo  {journal}
  {arXiv preprint quant-ph/9811052}\ } (\bibinfo {year} {1998})}\BibitemShut
  {NoStop}%
\bibitem [{\citenamefont {Dennis}\ \emph {et~al.}(2002)\citenamefont {Dennis},
  \citenamefont {Kitaev}, \citenamefont {Landahl},\ and\ \citenamefont
  {Preskill}}]{dennis2002topological}%
  \BibitemOpen
  \bibfield  {author} {\bibinfo {author} {\bibfnamefont {E.}~\bibnamefont
  {Dennis}}, \bibinfo {author} {\bibfnamefont {A.}~\bibnamefont {Kitaev}},
  \bibinfo {author} {\bibfnamefont {A.}~\bibnamefont {Landahl}},\ and\ \bibinfo
  {author} {\bibfnamefont {J.}~\bibnamefont {Preskill}},\ }\bibfield  {title}
  {\bibinfo {title} {Topological quantum memory},\ }\href@noop {} {\bibfield
  {journal} {\bibinfo  {journal} {Journal of Mathematical Physics}\ }\textbf
  {\bibinfo {volume} {43}},\ \bibinfo {pages} {4452} (\bibinfo {year}
  {2002})}\BibitemShut {NoStop}%
\bibitem [{\citenamefont {Raussendorf}\ and\ \citenamefont
  {Harrington}(2007)}]{raussendorf2007fault}%
  \BibitemOpen
  \bibfield  {author} {\bibinfo {author} {\bibfnamefont {R.}~\bibnamefont
  {Raussendorf}}\ and\ \bibinfo {author} {\bibfnamefont {J.}~\bibnamefont
  {Harrington}},\ }\bibfield  {title} {\bibinfo {title} {Fault-tolerant quantum
  computation with high threshold in two dimensions},\ }\href@noop {}
  {\bibfield  {journal} {\bibinfo  {journal} {Physical Review Letters}\
  }\textbf {\bibinfo {volume} {98}},\ \bibinfo {pages} {190504} (\bibinfo
  {year} {2007})}\BibitemShut {NoStop}%
\bibitem [{\citenamefont {Raussendorf}\ \emph {et~al.}(2007)\citenamefont
  {Raussendorf}, \citenamefont {Harrington},\ and\ \citenamefont
  {Goyal}}]{raussendorf2007topological}%
  \BibitemOpen
  \bibfield  {author} {\bibinfo {author} {\bibfnamefont {R.}~\bibnamefont
  {Raussendorf}}, \bibinfo {author} {\bibfnamefont {J.}~\bibnamefont
  {Harrington}},\ and\ \bibinfo {author} {\bibfnamefont {K.}~\bibnamefont
  {Goyal}},\ }\bibfield  {title} {\bibinfo {title} {Topological fault-tolerance
  in cluster state quantum computation},\ }\href@noop {} {\bibfield  {journal}
  {\bibinfo  {journal} {New Journal of Physics}\ }\textbf {\bibinfo {volume}
  {9}},\ \bibinfo {pages} {199} (\bibinfo {year} {2007})}\BibitemShut {NoStop}%
\bibitem [{\citenamefont {Postler}\ \emph {et~al.}(2022)\citenamefont
  {Postler}, \citenamefont {Heu$\beta$en}, \citenamefont {Pogorelov},
  \citenamefont {Rispler}, \citenamefont {Feldker}, \citenamefont {Meth},
  \citenamefont {Marciniak}, \citenamefont {Stricker}, \citenamefont
  {Ringbauer}, \citenamefont {Blatt}, \citenamefont {Schindler}, \citenamefont
  {M{\"{u}}ller},\ and\ \citenamefont {Monz}}]{Postler2022}%
  \BibitemOpen
  \bibfield  {author} {\bibinfo {author} {\bibfnamefont {L.}~\bibnamefont
  {Postler}}, \bibinfo {author} {\bibfnamefont {S.}~\bibnamefont
  {Heu$\beta$en}}, \bibinfo {author} {\bibfnamefont {I.}~\bibnamefont
  {Pogorelov}}, \bibinfo {author} {\bibfnamefont {M.}~\bibnamefont {Rispler}},
  \bibinfo {author} {\bibfnamefont {T.}~\bibnamefont {Feldker}}, \bibinfo
  {author} {\bibfnamefont {M.}~\bibnamefont {Meth}}, \bibinfo {author}
  {\bibfnamefont {C.~D.}\ \bibnamefont {Marciniak}}, \bibinfo {author}
  {\bibfnamefont {R.}~\bibnamefont {Stricker}}, \bibinfo {author}
  {\bibfnamefont {M.}~\bibnamefont {Ringbauer}}, \bibinfo {author}
  {\bibfnamefont {R.}~\bibnamefont {Blatt}}, \bibinfo {author} {\bibfnamefont
  {P.}~\bibnamefont {Schindler}}, \bibinfo {author} {\bibfnamefont
  {M.}~\bibnamefont {M{\"{u}}ller}},\ and\ \bibinfo {author} {\bibfnamefont
  {T.}~\bibnamefont {Monz}},\ }\bibfield  {title} {\bibinfo {title}
  {{Demonstration of fault-tolerant universal quantum gate operations}},\
  }\bibfield  {journal} {\bibinfo  {journal} {Nature}\ }\textbf {\bibinfo
  {volume} {605}},\ \href {https://doi.org/10.1038/s41586-022-04721-1}
  {10.1038/s41586-022-04721-1} (\bibinfo {year} {2022})\BibitemShut {NoStop}%
\bibitem [{\citenamefont {Krinner}\ \emph {et~al.}(2022)\citenamefont
  {Krinner}, \citenamefont {Lacroix}, \citenamefont {Remm}, \citenamefont {{Di
  Paolo}}, \citenamefont {Genois}, \citenamefont {Leroux}, \citenamefont
  {Hellings}, \citenamefont {Lazar}, \citenamefont {Swiadek}, \citenamefont
  {Herrmann}, \citenamefont {Norris}, \citenamefont {Andersen}, \citenamefont
  {M{\"{u}}ller}, \citenamefont {Blais}, \citenamefont {Eichler},\ and\
  \citenamefont {Wallraff}}]{Krinner2022}%
  \BibitemOpen
  \bibfield  {author} {\bibinfo {author} {\bibfnamefont {S.}~\bibnamefont
  {Krinner}}, \bibinfo {author} {\bibfnamefont {N.}~\bibnamefont {Lacroix}},
  \bibinfo {author} {\bibfnamefont {A.}~\bibnamefont {Remm}}, \bibinfo {author}
  {\bibfnamefont {A.}~\bibnamefont {{Di Paolo}}}, \bibinfo {author}
  {\bibfnamefont {E.}~\bibnamefont {Genois}}, \bibinfo {author} {\bibfnamefont
  {C.}~\bibnamefont {Leroux}}, \bibinfo {author} {\bibfnamefont
  {C.}~\bibnamefont {Hellings}}, \bibinfo {author} {\bibfnamefont
  {S.}~\bibnamefont {Lazar}}, \bibinfo {author} {\bibfnamefont
  {F.}~\bibnamefont {Swiadek}}, \bibinfo {author} {\bibfnamefont
  {J.}~\bibnamefont {Herrmann}}, \bibinfo {author} {\bibfnamefont {G.~J.}\
  \bibnamefont {Norris}}, \bibinfo {author} {\bibfnamefont {C.~K.}\
  \bibnamefont {Andersen}}, \bibinfo {author} {\bibfnamefont {M.}~\bibnamefont
  {M{\"{u}}ller}}, \bibinfo {author} {\bibfnamefont {A.}~\bibnamefont {Blais}},
  \bibinfo {author} {\bibfnamefont {C.}~\bibnamefont {Eichler}},\ and\ \bibinfo
  {author} {\bibfnamefont {A.}~\bibnamefont {Wallraff}},\ }\bibfield  {title}
  {\bibinfo {title} {{Realizing repeated quantum error correction in a
  distance-three surface code}},\ }\bibfield  {journal} {\bibinfo  {journal}
  {Nature}\ }\textbf {\bibinfo {volume} {605}},\ \href
  {https://doi.org/10.1038/s41586-022-04566-8} {10.1038/s41586-022-04566-8}
  (\bibinfo {year} {2022})\BibitemShut {NoStop}%
\bibitem [{\citenamefont {AI}(2023)}]{Acharya2023}%
  \BibitemOpen
  \bibfield  {author} {\bibinfo {author} {\bibfnamefont {G.~Q.}\ \bibnamefont
  {AI}},\ }\bibfield  {title} {\bibinfo {title} {{Suppressing quantum errors by
  scaling a surface code logical qubit}},\ }\bibfield  {journal} {\bibinfo
  {journal} {Nature}\ }\textbf {\bibinfo {volume} {614}},\ \href
  {https://doi.org/10.1038/s41586-022-05434-1} {10.1038/s41586-022-05434-1}
  (\bibinfo {year} {2023})\BibitemShut {NoStop}%
\bibitem [{\citenamefont {Gottesman}(2014)}]{gottesman2014fault}%
  \BibitemOpen
  \bibfield  {author} {\bibinfo {author} {\bibfnamefont {D.}~\bibnamefont
  {Gottesman}},\ }\bibfield  {title} {\bibinfo {title} {Fault-tolerant quantum
  computation with constant overhead},\ }\href@noop {} {\bibfield  {journal}
  {\bibinfo  {journal} {Quantum Information \& Computation}\ }\textbf {\bibinfo
  {volume} {14}},\ \bibinfo {pages} {1338} (\bibinfo {year}
  {2014})}\BibitemShut {NoStop}%
\bibitem [{\citenamefont {Tillich}\ and\ \citenamefont
  {Z{\'e}mor}(2014)}]{tillich2014quantum}%
  \BibitemOpen
  \bibfield  {author} {\bibinfo {author} {\bibfnamefont {J.-P.}\ \bibnamefont
  {Tillich}}\ and\ \bibinfo {author} {\bibfnamefont {G.}~\bibnamefont
  {Z{\'e}mor}},\ }\bibfield  {title} {\bibinfo {title} {Quantum {LDPC} codes
  with positive rate and minimum distance proportional to the square root of
  the blocklength},\ }\href@noop {} {\bibfield  {journal} {\bibinfo  {journal}
  {IEEE Transactions on Information Theory}\ }\textbf {\bibinfo {volume}
  {60}},\ \bibinfo {pages} {1193} (\bibinfo {year} {2014})}\BibitemShut
  {NoStop}%
\bibitem [{\citenamefont {Leverrier}\ \emph {et~al.}(2015)\citenamefont
  {Leverrier}, \citenamefont {Tillich},\ and\ \citenamefont
  {Z{\'e}mor}}]{leverrier2015quantum}%
  \BibitemOpen
  \bibfield  {author} {\bibinfo {author} {\bibfnamefont {A.}~\bibnamefont
  {Leverrier}}, \bibinfo {author} {\bibfnamefont {J.-P.}\ \bibnamefont
  {Tillich}},\ and\ \bibinfo {author} {\bibfnamefont {G.}~\bibnamefont
  {Z{\'e}mor}},\ }\bibfield  {title} {\bibinfo {title} {Quantum expander
  codes},\ }in\ \href@noop {} {\emph {\bibinfo {booktitle} {Foundations of
  Computer Science (FOCS), 2015 IEEE 56th Annual Symposium on}}}\ (\bibinfo
  {organization} {IEEE},\ \bibinfo {year} {2015})\ pp.\ \bibinfo {pages}
  {810--824}\BibitemShut {NoStop}%
\bibitem [{\citenamefont {Breuckmann}\ and\ \citenamefont
  {Eberhardt}(2021)}]{breuckmann2021ldpc}%
  \BibitemOpen
  \bibfield  {author} {\bibinfo {author} {\bibfnamefont {N.~P.}\ \bibnamefont
  {Breuckmann}}\ and\ \bibinfo {author} {\bibfnamefont {J.}~\bibnamefont
  {Eberhardt}},\ }\bibfield  {title} {\bibinfo {title} {{LDPC} quantum codes},\
  }\href@noop {} {\bibfield  {journal} {\bibinfo  {journal} {arXiv preprint
  arXiv:2103.06309}\ } (\bibinfo {year} {2021})}\BibitemShut {NoStop}%
\bibitem [{\citenamefont {Wegner}(1971)}]{Wegner1971}%
  \BibitemOpen
  \bibfield  {author} {\bibinfo {author} {\bibfnamefont {F.~J.}\ \bibnamefont
  {Wegner}},\ }\bibfield  {title} {\bibinfo {title} {{Duality in generalized
  Ising models and phase transitions without local order parameters}},\
  }\bibfield  {journal} {\bibinfo  {journal} {Journal of Mathematical Physics}\
  }\textbf {\bibinfo {volume} {12}},\ \href {https://doi.org/10.1063/1.1665530}
  {10.1063/1.1665530} (\bibinfo {year} {1971})\BibitemShut {NoStop}%
\bibitem [{\citenamefont {Kogut}(1979)}]{Kogut1979}%
  \BibitemOpen
  \bibfield  {author} {\bibinfo {author} {\bibfnamefont {J.~B.}\ \bibnamefont
  {Kogut}},\ }\bibfield  {title} {\bibinfo {title} {{An introduction to lattice
  gauge theory and spin systems}},\ }\href
  {https://doi.org/10.1103/RevModPhys.51.659} {\bibfield  {journal} {\bibinfo
  {journal} {Reviews of Modern Physics}\ }\textbf {\bibinfo {volume} {51}},\
  \bibinfo {pages} {659} (\bibinfo {year} {1979})}\BibitemShut {NoStop}%
\bibitem [{\citenamefont {Dijkgraaf}\ and\ \citenamefont
  {Witten}(1990)}]{DijkgraafWitten}%
  \BibitemOpen
  \bibfield  {author} {\bibinfo {author} {\bibfnamefont {R.}~\bibnamefont
  {Dijkgraaf}}\ and\ \bibinfo {author} {\bibfnamefont {E.}~\bibnamefont
  {Witten}},\ }\bibfield  {title} {\bibinfo {title} {{Topological gauge
  theories and group cohomology}},\ }\href {https://doi.org/10.1007/BF02096988}
  {\bibfield  {journal} {\bibinfo  {journal} {Communications in Mathematical
  Physics}\ }\textbf {\bibinfo {volume} {129}},\ \bibinfo {pages} {393}
  (\bibinfo {year} {1990})}\BibitemShut {NoStop}%
\bibitem [{\citenamefont {Kitaev}(2006)}]{kitaev2006anyons}%
  \BibitemOpen
  \bibfield  {author} {\bibinfo {author} {\bibfnamefont {A.}~\bibnamefont
  {Kitaev}},\ }\bibfield  {title} {\bibinfo {title} {{Anyons in an exactly
  solved model and beyond}},\ }\href
  {https://doi.org/10.1016/j.aop.2005.10.005} {\bibfield  {journal} {\bibinfo
  {journal} {Annals of Physics}\ }\textbf {\bibinfo {volume} {321}},\ \bibinfo
  {pages} {2} (\bibinfo {year} {2006})},\ \Eprint
  {https://arxiv.org/abs/0506438} {arXiv:0506438 [cond-mat]} \BibitemShut
  {NoStop}%
\bibitem [{\citenamefont {Hastings}\ \emph {et~al.}(2020)\citenamefont
  {Hastings}, \citenamefont {Haah},\ and\ \citenamefont
  {O’Donnell}}]{hastings2020fiber}%
  \BibitemOpen
  \bibfield  {author} {\bibinfo {author} {\bibfnamefont {M.~B.}\ \bibnamefont
  {Hastings}}, \bibinfo {author} {\bibfnamefont {J.}~\bibnamefont {Haah}},\
  and\ \bibinfo {author} {\bibfnamefont {R.}~\bibnamefont {O’Donnell}},\
  }\bibfield  {title} {\bibinfo {title} {Fiber bundle codes: Breaking the
  {$N^{1/2}\text{poly}\log(N)$} barrier for quantum {LDPC} codes},\ }\href@noop
  {} {\bibfield  {journal} {\bibinfo  {journal} {arXiv preprint
  arXiv:2009.03921}\ } (\bibinfo {year} {2020})}\BibitemShut {NoStop}%
\bibitem [{\citenamefont {Panteleev}\ and\ \citenamefont
  {Kalachev}(2020)}]{panteleev2020quantum}%
  \BibitemOpen
  \bibfield  {author} {\bibinfo {author} {\bibfnamefont {P.}~\bibnamefont
  {Panteleev}}\ and\ \bibinfo {author} {\bibfnamefont {G.}~\bibnamefont
  {Kalachev}},\ }\bibfield  {title} {\bibinfo {title} {Quantum {LDPC} codes
  with almost linear minimum distance},\ }\href@noop {} {\bibfield  {journal}
  {\bibinfo  {journal} {arXiv preprint arXiv:2012.04068}\ } (\bibinfo {year}
  {2020})}\BibitemShut {NoStop}%
\bibitem [{\citenamefont {Breuckmann}\ and\ \citenamefont
  {Eberhardt}(2020)}]{breuckmann2020balanced}%
  \BibitemOpen
  \bibfield  {author} {\bibinfo {author} {\bibfnamefont {N.~P.}\ \bibnamefont
  {Breuckmann}}\ and\ \bibinfo {author} {\bibfnamefont {J.~N.}\ \bibnamefont
  {Eberhardt}},\ }\bibfield  {title} {\bibinfo {title} {Balanced product
  quantum codes},\ }\href@noop {} {\bibfield  {journal} {\bibinfo  {journal}
  {arXiv preprint arXiv:2012.09271}\ } (\bibinfo {year} {2020})}\BibitemShut
  {NoStop}%
\bibitem [{\citenamefont {Panteleev}\ and\ \citenamefont
  {Kalachev}(2022)}]{Panteleev2022}%
  \BibitemOpen
  \bibfield  {author} {\bibinfo {author} {\bibfnamefont {P.}~\bibnamefont
  {Panteleev}}\ and\ \bibinfo {author} {\bibfnamefont {G.}~\bibnamefont
  {Kalachev}},\ }\bibfield  {title} {\bibinfo {title} {{Asymptotically good
  Quantum and locally testable classical LDPC codes}},\ }in\ \href
  {https://doi.org/10.1145/3519935.3520017} {\emph {\bibinfo {booktitle}
  {Proceedings of the Annual ACM Symposium on Theory of Computing}}}\ (\bibinfo
  {year} {2022})\BibitemShut {NoStop}%
\bibitem [{\citenamefont {Leverrier}\ and\ \citenamefont
  {Zémor}(2022)}]{leverrier2022quantum}%
  \BibitemOpen
  \bibfield  {author} {\bibinfo {author} {\bibfnamefont {A.}~\bibnamefont
  {Leverrier}}\ and\ \bibinfo {author} {\bibfnamefont {G.}~\bibnamefont
  {Zémor}},\ }\href {https://doi.org/10.48550/ARXIV.2202.13641} {\bibinfo
  {title} {Quantum tanner codes}} (\bibinfo {year} {2022})\BibitemShut
  {NoStop}%
\bibitem [{\citenamefont {Dinur}\ \emph {et~al.}(2023)\citenamefont {Dinur},
  \citenamefont {Hsieh}, \citenamefont {Lin},\ and\ \citenamefont
  {Vidick}}]{Dinur2023}%
  \BibitemOpen
  \bibfield  {author} {\bibinfo {author} {\bibfnamefont {I.}~\bibnamefont
  {Dinur}}, \bibinfo {author} {\bibfnamefont {M.~H.}\ \bibnamefont {Hsieh}},
  \bibinfo {author} {\bibfnamefont {T.~C.}\ \bibnamefont {Lin}},\ and\ \bibinfo
  {author} {\bibfnamefont {T.}~\bibnamefont {Vidick}},\ }\bibfield  {title}
  {\bibinfo {title} {{Good Quantum LDPC Codes with Linear Time Decoders}},\
  }in\ \href {https://doi.org/10.1145/3564246.3585101} {\emph {\bibinfo
  {booktitle} {Proceedings of the Annual ACM Symposium on Theory of
  Computing}}}\ (\bibinfo {year} {2023})\BibitemShut {NoStop}%
\bibitem [{\citenamefont {Bravyi}\ \emph
  {et~al.}(2010{\natexlab{a}})\citenamefont {Bravyi}, \citenamefont {Poulin},\
  and\ \citenamefont {Terhal}}]{bravyi2010tradeoffs}%
  \BibitemOpen
  \bibfield  {author} {\bibinfo {author} {\bibfnamefont {S.}~\bibnamefont
  {Bravyi}}, \bibinfo {author} {\bibfnamefont {D.}~\bibnamefont {Poulin}},\
  and\ \bibinfo {author} {\bibfnamefont {B.}~\bibnamefont {Terhal}},\
  }\bibfield  {title} {\bibinfo {title} {Tradeoffs for reliable quantum
  information storage in 2{D} systems},\ }\href@noop {} {\bibfield  {journal}
  {\bibinfo  {journal} {Physical Review Letters}\ }\textbf {\bibinfo {volume}
  {104}},\ \bibinfo {pages} {050503} (\bibinfo {year}
  {2010}{\natexlab{a}})}\BibitemShut {NoStop}%
\bibitem [{\citenamefont {Bravyi}\ and\ \citenamefont
  {Terhal}(2009)}]{bravyi2009no}%
  \BibitemOpen
  \bibfield  {author} {\bibinfo {author} {\bibfnamefont {S.}~\bibnamefont
  {Bravyi}}\ and\ \bibinfo {author} {\bibfnamefont {B.}~\bibnamefont
  {Terhal}},\ }\bibfield  {title} {\bibinfo {title} {A no-go theorem for a
  two-dimensional self-correcting quantum memory based on stabilizer codes},\
  }\href@noop {} {\bibfield  {journal} {\bibinfo  {journal} {New Journal of
  Physics}\ }\textbf {\bibinfo {volume} {11}},\ \bibinfo {pages} {043029}
  (\bibinfo {year} {2009})}\BibitemShut {NoStop}%
\bibitem [{\citenamefont {Haah}(2021)}]{Haah2021}%
  \BibitemOpen
  \bibfield  {author} {\bibinfo {author} {\bibfnamefont {J.}~\bibnamefont
  {Haah}},\ }\bibfield  {title} {\bibinfo {title} {{A degeneracy bound for
  homogeneous topological order}},\ }\bibfield  {journal} {\bibinfo  {journal}
  {SciPost Physics}\ }\textbf {\bibinfo {volume} {10}},\ \href
  {https://doi.org/10.21468/SCIPOSTPHYS.10.1.011}
  {10.21468/SCIPOSTPHYS.10.1.011} (\bibinfo {year} {2021})\BibitemShut
  {NoStop}%
\bibitem [{\citenamefont {Calderbank}\ and\ \citenamefont
  {Shor}(1996)}]{calderbank1996good}%
  \BibitemOpen
  \bibfield  {author} {\bibinfo {author} {\bibfnamefont {A.~R.}\ \bibnamefont
  {Calderbank}}\ and\ \bibinfo {author} {\bibfnamefont {P.~W.}\ \bibnamefont
  {Shor}},\ }\bibfield  {title} {\bibinfo {title} {Good quantum
  error-correcting codes exist},\ }\href@noop {} {\bibfield  {journal}
  {\bibinfo  {journal} {Physical Review A}\ }\textbf {\bibinfo {volume} {54}},\
  \bibinfo {pages} {1098} (\bibinfo {year} {1996})}\BibitemShut {NoStop}%
\bibitem [{\citenamefont {Steane}(1996{\natexlab{b}})}]{Steane1996Simple}%
  \BibitemOpen
  \bibfield  {author} {\bibinfo {author} {\bibfnamefont {A.~M.}\ \bibnamefont
  {Steane}},\ }\bibfield  {title} {\bibinfo {title} {{Simple quantum
  error-correcting codes}},\ }\href {https://doi.org/10.1103/PhysRevA.54.4741}
  {\bibfield  {journal} {\bibinfo  {journal} {Physical Review A - Atomic,
  Molecular, and Optical Physics}\ }\textbf {\bibinfo {volume} {54}},\ \bibinfo
  {pages} {4741} (\bibinfo {year} {1996}{\natexlab{b}})},\ \Eprint
  {https://arxiv.org/abs/9605021} {arXiv:9605021 [quant-ph]} \BibitemShut
  {NoStop}%
\bibitem [{\citenamefont {Slagle}\ \emph {et~al.}(2019)\citenamefont {Slagle},
  \citenamefont {Aasen},\ and\ \citenamefont {Williamson}}]{Slagle2018}%
  \BibitemOpen
  \bibfield  {author} {\bibinfo {author} {\bibfnamefont {K.}~\bibnamefont
  {Slagle}}, \bibinfo {author} {\bibfnamefont {D.}~\bibnamefont {Aasen}},\ and\
  \bibinfo {author} {\bibfnamefont {D.}~\bibnamefont {Williamson}},\ }\bibfield
   {title} {\bibinfo {title} {{Foliated field theory and string-membrane-net
  condensation picture of fracton order}},\ }\bibfield  {journal} {\bibinfo
  {journal} {SciPost Physics}\ }\textbf {\bibinfo {volume} {6}},\ \href
  {https://doi.org/10.21468/scipostphys.6.4.043} {10.21468/scipostphys.6.4.043}
  (\bibinfo {year} {2019}),\ \Eprint {https://arxiv.org/abs/1812.01613}
  {arXiv:1812.01613} \BibitemShut {NoStop}%
\bibitem [{\citenamefont {Aasen}\ \emph {et~al.}(2020)\citenamefont {Aasen},
  \citenamefont {Bulmash}, \citenamefont {Prem}, \citenamefont {Slagle},\ and\
  \citenamefont {Williamson}}]{Aasen2020}%
  \BibitemOpen
  \bibfield  {author} {\bibinfo {author} {\bibfnamefont {D.}~\bibnamefont
  {Aasen}}, \bibinfo {author} {\bibfnamefont {D.}~\bibnamefont {Bulmash}},
  \bibinfo {author} {\bibfnamefont {A.}~\bibnamefont {Prem}}, \bibinfo {author}
  {\bibfnamefont {K.}~\bibnamefont {Slagle}},\ and\ \bibinfo {author}
  {\bibfnamefont {D.~J.}\ \bibnamefont {Williamson}},\ }\bibfield  {title}
  {\bibinfo {title} {{Topological defect networks for fractons of all types}},\
  }\bibfield  {journal} {\bibinfo  {journal} {Physical Review Research}\
  }\textbf {\bibinfo {volume} {2}},\ \href
  {https://doi.org/10.1103/physrevresearch.2.043165}
  {10.1103/physrevresearch.2.043165} (\bibinfo {year} {2020}),\ \Eprint
  {https://arxiv.org/abs/2002.05166} {arXiv:2002.05166} \BibitemShut {NoStop}%
\bibitem [{\citenamefont {Song}\ \emph
  {et~al.}(2023{\natexlab{a}})\citenamefont {Song}, \citenamefont {Dua},
  \citenamefont {Shirley},\ and\ \citenamefont {Williamson}}]{Song2021}%
  \BibitemOpen
  \bibfield  {author} {\bibinfo {author} {\bibfnamefont {Z.}~\bibnamefont
  {Song}}, \bibinfo {author} {\bibfnamefont {A.}~\bibnamefont {Dua}}, \bibinfo
  {author} {\bibfnamefont {W.}~\bibnamefont {Shirley}},\ and\ \bibinfo {author}
  {\bibfnamefont {D.~J.}\ \bibnamefont {Williamson}},\ }\bibfield  {title}
  {\bibinfo {title} {Topological defect network representations of fracton
  stabilizer codes},\ }\href
  {https://doi.org/10.1103/PRXQUANTUM.4.010304/FIGURES/8/MEDIUM} {\bibfield
  {journal} {\bibinfo  {journal} {PRX Quantum}\ }\textbf {\bibinfo {volume}
  {4}},\ \bibinfo {pages} {010304} (\bibinfo {year}
  {2023}{\natexlab{a}})}\BibitemShut {NoStop}%
\bibitem [{\citenamefont {Bravyi}(2011)}]{bravyi2011subsystem}%
  \BibitemOpen
  \bibfield  {author} {\bibinfo {author} {\bibfnamefont {S.}~\bibnamefont
  {Bravyi}},\ }\bibfield  {title} {\bibinfo {title} {Subsystem codes with
  spatially local generators},\ }\bibfield  {journal} {\bibinfo  {journal}
  {Physical Review A}\ }\textbf {\bibinfo {volume} {83}},\ \href
  {https://doi.org/10.1103/physreva.83.012320} {10.1103/physreva.83.012320}
  (\bibinfo {year} {2011})\BibitemShut {NoStop}%
\bibitem [{\citenamefont {Bacon}\ \emph {et~al.}(2015)\citenamefont {Bacon},
  \citenamefont {Flammia}, \citenamefont {Harrow},\ and\ \citenamefont
  {Shi}}]{bacon2015sparse}%
  \BibitemOpen
  \bibfield  {author} {\bibinfo {author} {\bibfnamefont {D.}~\bibnamefont
  {Bacon}}, \bibinfo {author} {\bibfnamefont {S.~T.}\ \bibnamefont {Flammia}},
  \bibinfo {author} {\bibfnamefont {A.~W.}\ \bibnamefont {Harrow}},\ and\
  \bibinfo {author} {\bibfnamefont {J.}~\bibnamefont {Shi}},\ }\bibfield
  {title} {\bibinfo {title} {Sparse quantum codes from quantum circuits},\ }in\
  \href@noop {} {\emph {\bibinfo {booktitle} {Proceedings of the forty-seventh
  annual ACM symposium on Theory of Computing}}}\ (\bibinfo {year} {2015})\
  pp.\ \bibinfo {pages} {327--334}\BibitemShut {NoStop}%
\bibitem [{\citenamefont {Baspin}\ and\ \citenamefont
  {Williamson}()}]{WireCodes}%
  \BibitemOpen
  \bibfield  {author} {\bibinfo {author} {\bibfnamefont {N.}~\bibnamefont
  {Baspin}}\ and\ \bibinfo {author} {\bibfnamefont {D.~J.}\ \bibnamefont
  {Williamson}},\ }\href@noop {} {\bibinfo {title} {Wire codes, \textit{in
  preparation}}}\BibitemShut {NoStop}%
\bibitem [{\citenamefont {Baspin}(2023)}]{baspin2023combinatorial}%
  \BibitemOpen
  \bibfield  {author} {\bibinfo {author} {\bibfnamefont {N.}~\bibnamefont
  {Baspin}},\ }\href@noop {} {\bibinfo {title} {On combinatorial structures in
  linear codes}} (\bibinfo {year} {2023}),\ \Eprint
  {https://arxiv.org/abs/2309.16411} {arXiv:2309.16411 [cs.IT]} \BibitemShut
  {NoStop}%
\bibitem [{\citenamefont {Bombin}\ and\ \citenamefont
  {Martin-Delgado}(2007{\natexlab{a}})}]{Bombin2007}%
  \BibitemOpen
  \bibfield  {author} {\bibinfo {author} {\bibfnamefont {H.}~\bibnamefont
  {Bombin}}\ and\ \bibinfo {author} {\bibfnamefont {M.~A.}\ \bibnamefont
  {Martin-Delgado}},\ }\bibfield  {title} {\bibinfo {title} {{Exact topological
  quantum order in D = 3 and beyond: Branyons and brane-net condensates}},\
  }\href {https://doi.org/10.1103/PhysRevB.75.075103} {\bibfield  {journal}
  {\bibinfo  {journal} {Physical Review B}\ }\textbf {\bibinfo {volume} {75}},\
  \bibinfo {pages} {075103} (\bibinfo {year} {2007}{\natexlab{a}})}\BibitemShut
  {NoStop}%
\bibitem [{\citenamefont {Bomb{\'{i}}n}(2015)}]{Bombin2015}%
  \BibitemOpen
  \bibfield  {author} {\bibinfo {author} {\bibfnamefont {H.}~\bibnamefont
  {Bomb{\'{i}}n}},\ }\bibfield  {title} {\bibinfo {title} {{Gauge color codes:
  Optimal transversal gates and gauge fixing in topological stabilizer
  codes}},\ }\bibfield  {journal} {\bibinfo  {journal} {New Journal of
  Physics}\ }\textbf {\bibinfo {volume} {17}},\ \href
  {https://doi.org/10.1088/1367-2630/17/8/083002}
  {10.1088/1367-2630/17/8/083002} (\bibinfo {year} {2015}),\ \Eprint
  {https://arxiv.org/abs/1311.0879} {arXiv:1311.0879} \BibitemShut {NoStop}%
\bibitem [{\citenamefont {Brown}\ \emph {et~al.}(2016)\citenamefont {Brown},
  \citenamefont {Loss}, \citenamefont {Pachos}, \citenamefont {Self},\ and\
  \citenamefont {Wootton}}]{Brown2016}%
  \BibitemOpen
  \bibfield  {author} {\bibinfo {author} {\bibfnamefont {B.~J.}\ \bibnamefont
  {Brown}}, \bibinfo {author} {\bibfnamefont {D.}~\bibnamefont {Loss}},
  \bibinfo {author} {\bibfnamefont {J.~K.}\ \bibnamefont {Pachos}}, \bibinfo
  {author} {\bibfnamefont {C.~N.}\ \bibnamefont {Self}},\ and\ \bibinfo
  {author} {\bibfnamefont {J.~R.}\ \bibnamefont {Wootton}},\ }\bibfield
  {title} {\bibinfo {title} {Quantum memories at finite temperature},\
  }\bibfield  {journal} {\bibinfo  {journal} {Reviews of Modern Physics}\
  }\textbf {\bibinfo {volume} {88}},\ \href
  {https://doi.org/10.1103/revmodphys.88.045005} {10.1103/revmodphys.88.045005}
  (\bibinfo {year} {2016})\BibitemShut {NoStop}%
\bibitem [{\citenamefont {Castelnovo}\ and\ \citenamefont
  {Chamon}(2007)}]{Castelnovo2007}%
  \BibitemOpen
  \bibfield  {author} {\bibinfo {author} {\bibfnamefont {C.}~\bibnamefont
  {Castelnovo}}\ and\ \bibinfo {author} {\bibfnamefont {C.}~\bibnamefont
  {Chamon}},\ }\bibfield  {title} {\bibinfo {title} {{Entanglement and
  topological entropy of the toric code at finite temperature}},\ }\bibfield
  {journal} {\bibinfo  {journal} {Physical Review B - Condensed Matter and
  Materials Physics}\ }\textbf {\bibinfo {volume} {76}},\ \href
  {https://doi.org/10.1103/PhysRevB.76.184442} {10.1103/PhysRevB.76.184442}
  (\bibinfo {year} {2007}),\ \Eprint {https://arxiv.org/abs/0704.3616}
  {arXiv:0704.3616} \BibitemShut {NoStop}%
\bibitem [{\citenamefont {Castelnovo}\ and\ \citenamefont
  {Chamon}(2008)}]{Castelnovo2008}%
  \BibitemOpen
  \bibfield  {author} {\bibinfo {author} {\bibfnamefont {C.}~\bibnamefont
  {Castelnovo}}\ and\ \bibinfo {author} {\bibfnamefont {C.}~\bibnamefont
  {Chamon}},\ }\bibfield  {title} {\bibinfo {title} {{Topological order in a
  three-dimensional toric code at finite temperature}},\ }\bibfield  {journal}
  {\bibinfo  {journal} {Physical Review B - Condensed Matter and Materials
  Physics}\ }\textbf {\bibinfo {volume} {78}},\ \href
  {https://doi.org/10.1103/PhysRevB.78.155120} {10.1103/PhysRevB.78.155120}
  (\bibinfo {year} {2008}),\ \Eprint {https://arxiv.org/abs/0804.3591}
  {arXiv:0804.3591} \BibitemShut {NoStop}%
\bibitem [{\citenamefont
  {Yoshida}(2011{\natexlab{a}})}]{yoshida2011classification}%
  \BibitemOpen
  \bibfield  {author} {\bibinfo {author} {\bibfnamefont {B.}~\bibnamefont
  {Yoshida}},\ }\bibfield  {title} {\bibinfo {title} {{Classification of
  quantum phases and topology of logical operators in an exactly solved model
  of quantum codes}},\ }\href {https://doi.org/10.1016/j.aop.2010.10.009}
  {\bibfield  {journal} {\bibinfo  {journal} {Annals of Physics}\ }\textbf
  {\bibinfo {volume} {326}},\ \bibinfo {pages} {15} (\bibinfo {year}
  {2011}{\natexlab{a}})},\ \Eprint {https://arxiv.org/abs/1007.4601}
  {arXiv:1007.4601} \BibitemShut {NoStop}%
\bibitem [{\citenamefont
  {Yoshida}(2011{\natexlab{b}})}]{Yoshida2011Feasibility}%
  \BibitemOpen
  \bibfield  {author} {\bibinfo {author} {\bibfnamefont {B.}~\bibnamefont
  {Yoshida}},\ }\bibfield  {title} {\bibinfo {title} {{Feasibility of
  self-correcting quantum memory and thermal stability of topological order}},\
  }\bibfield  {journal} {\bibinfo  {journal} {Annals of Physics}\ }\textbf
  {\bibinfo {volume} {326}},\ \href {https://doi.org/10.1016/j.aop.2011.06.001}
  {10.1016/j.aop.2011.06.001} (\bibinfo {year}
  {2011}{\natexlab{b}})\BibitemShut {NoStop}%
\bibitem [{\citenamefont {Chamon}(2005)}]{chamon2005quantum}%
  \BibitemOpen
  \bibfield  {author} {\bibinfo {author} {\bibfnamefont {C.}~\bibnamefont
  {Chamon}},\ }\bibfield  {title} {\bibinfo {title} {{Quantum glassiness in
  strongly correlated clean systems: An example of topological
  overprotection}},\ }\href {https://doi.org/10.1103/PhysRevLett.94.040402}
  {\bibfield  {journal} {\bibinfo  {journal} {Physical Review Letters}\
  }\textbf {\bibinfo {volume} {94}},\ \bibinfo {pages} {40402} (\bibinfo {year}
  {2005})},\ \Eprint {https://arxiv.org/abs/0404182} {arXiv:0404182 [cond-mat]}
  \BibitemShut {NoStop}%
\bibitem [{\citenamefont {Haah}(2011)}]{haah2011fractal}%
  \BibitemOpen
  \bibfield  {author} {\bibinfo {author} {\bibfnamefont {J.}~\bibnamefont
  {Haah}},\ }\bibfield  {title} {\bibinfo {title} {Local stabilizer codes in
  three dimensions without string logical operators},\ }\bibfield  {journal}
  {\bibinfo  {journal} {Physical Review A}\ }\textbf {\bibinfo {volume} {83}},\
  \href {https://doi.org/10.1103/physreva.83.042330}
  {10.1103/physreva.83.042330} (\bibinfo {year} {2011})\BibitemShut {NoStop}%
\bibitem [{\citenamefont {Walker}\ and\ \citenamefont
  {Wang}(2012)}]{walker2012}%
  \BibitemOpen
  \bibfield  {author} {\bibinfo {author} {\bibfnamefont {K.}~\bibnamefont
  {Walker}}\ and\ \bibinfo {author} {\bibfnamefont {Z.}~\bibnamefont {Wang}},\
  }\bibfield  {title} {\bibinfo {title} {{(3+1)-TQFTs and topological
  insulators}},\ }\href {https://doi.org/10.1007/s11467-011-0194-z} {\bibfield
  {journal} {\bibinfo  {journal} {Frontiers of Physics}\ }\textbf {\bibinfo
  {volume} {7}},\ \bibinfo {pages} {150} (\bibinfo {year} {2012})},\ \Eprint
  {https://arxiv.org/abs/1104.2632} {arXiv:1104.2632} \BibitemShut {NoStop}%
\bibitem [{\citenamefont {Kim}(2012)}]{kim20123d}%
  \BibitemOpen
  \bibfield  {author} {\bibinfo {author} {\bibfnamefont {I.~H.}\ \bibnamefont
  {Kim}},\ }\bibfield  {title} {\bibinfo {title} {{$3$D local qupit quantum
  code without string logical operator}},\ }\href
  {http://arxiv.org/abs/1202.0052} {\bibfield  {journal} {\bibinfo  {journal}
  {arXiv}\ ,\ \bibinfo {pages} {9}} (\bibinfo {year} {2012})},\ \Eprint
  {https://arxiv.org/abs/1202.0052} {arXiv:1202.0052} \BibitemShut {NoStop}%
\bibitem [{\citenamefont {Yoshida}(2013)}]{yoshida2013exotic}%
  \BibitemOpen
  \bibfield  {author} {\bibinfo {author} {\bibfnamefont {B.}~\bibnamefont
  {Yoshida}},\ }\bibfield  {title} {\bibinfo {title} {{Exotic topological order
  in fractal spin liquids}},\ }\href
  {https://doi.org/10.1103/PhysRevB.88.125122} {\bibfield  {journal} {\bibinfo
  {journal} {Physical Review B - Condensed Matter and Materials Physics}\
  }\textbf {\bibinfo {volume} {88}},\ \bibinfo {pages} {125122} (\bibinfo
  {year} {2013})},\ \Eprint {https://arxiv.org/abs/1302.6248} {arXiv:1302.6248}
  \BibitemShut {NoStop}%
\bibitem [{\citenamefont {Brell}(2016)}]{brell2014proposal}%
  \BibitemOpen
  \bibfield  {author} {\bibinfo {author} {\bibfnamefont {C.~G.}\ \bibnamefont
  {Brell}},\ }\bibfield  {title} {\bibinfo {title} {A proposal for
  self-correcting stabilizer quantum memories in 3 dimensions (or slightly
  less)},\ }\href {https://doi.org/10.1088/1367-2630/18/1/013050} {\bibfield
  {journal} {\bibinfo  {journal} {New Journal of Physics}\ }\textbf {\bibinfo
  {volume} {18}},\ \bibinfo {pages} {13050} (\bibinfo {year}
  {2016})}\BibitemShut {NoStop}%
\bibitem [{\citenamefont {Vijay}\ \emph {et~al.}(2015)\citenamefont {Vijay},
  \citenamefont {Haah},\ and\ \citenamefont {Fu}}]{PhysRevB.92.235136}%
  \BibitemOpen
  \bibfield  {author} {\bibinfo {author} {\bibfnamefont {S.}~\bibnamefont
  {Vijay}}, \bibinfo {author} {\bibfnamefont {J.}~\bibnamefont {Haah}},\ and\
  \bibinfo {author} {\bibfnamefont {L.}~\bibnamefont {Fu}},\ }\bibfield
  {title} {\bibinfo {title} {{A new kind of topological quantum order: A
  dimensional hierarchy of quasiparticles built from stationary excitations}},\
  }\href {https://doi.org/10.1103/PhysRevB.92.235136} {\bibfield  {journal}
  {\bibinfo  {journal} {Physical Review B - Condensed Matter and Materials
  Physics}\ }\textbf {\bibinfo {volume} {92}},\ \bibinfo {pages} {235136}
  (\bibinfo {year} {2015})},\ \Eprint {https://arxiv.org/abs/1505.02576}
  {arXiv:1505.02576} \BibitemShut {NoStop}%
\bibitem [{\citenamefont {Vijay}\ \emph {et~al.}(2016)\citenamefont {Vijay},
  \citenamefont {Haah},\ and\ \citenamefont {Fu}}]{Vijay2016}%
  \BibitemOpen
  \bibfield  {author} {\bibinfo {author} {\bibfnamefont {S.}~\bibnamefont
  {Vijay}}, \bibinfo {author} {\bibfnamefont {J.}~\bibnamefont {Haah}},\ and\
  \bibinfo {author} {\bibfnamefont {L.}~\bibnamefont {Fu}},\ }\bibfield
  {title} {\bibinfo {title} {{Fracton Topological Order, Generalized Lattice
  Gauge Theory and Duality}},\ }\bibfield  {journal} {\bibinfo  {journal}
  {Physical Review B}\ }\href {https://doi.org/10.1103/physrevb.94.235157}
  {10.1103/physrevb.94.235157} (\bibinfo {year} {2016}),\ \Eprint
  {https://arxiv.org/abs/1603.04442} {arXiv:1603.04442} \BibitemShut {NoStop}%
\bibitem [{\citenamefont {Williamson}(2016)}]{Williamson2016}%
  \BibitemOpen
  \bibfield  {author} {\bibinfo {author} {\bibfnamefont {D.~J.}\ \bibnamefont
  {Williamson}},\ }\bibfield  {title} {\bibinfo {title} {{Fractal symmetries:
  Ungauging the cubic code}},\ }\bibfield  {journal} {\bibinfo  {journal}
  {Physical Review B}\ }\textbf {\bibinfo {volume} {94}},\ \href
  {https://doi.org/10.1103/physrevb.94.155128} {10.1103/physrevb.94.155128}
  (\bibinfo {year} {2016}),\ \Eprint {https://arxiv.org/abs/1603.05182}
  {arXiv:1603.05182} \BibitemShut {NoStop}%
\bibitem [{\citenamefont {Bravyi}\ \emph
  {et~al.}(2010{\natexlab{b}})\citenamefont {Bravyi}, \citenamefont
  {Leemhuis},\ and\ \citenamefont {Terhal}}]{bravyi2011topological}%
  \BibitemOpen
  \bibfield  {author} {\bibinfo {author} {\bibfnamefont {S.}~\bibnamefont
  {Bravyi}}, \bibinfo {author} {\bibfnamefont {B.}~\bibnamefont {Leemhuis}},\
  and\ \bibinfo {author} {\bibfnamefont {B.~M.}\ \bibnamefont {Terhal}},\
  }\bibfield  {title} {\bibinfo {title} {{Topological order in an exactly
  solvable 3D spin model}},\ }\href {https://doi.org/10.1016/j.aop.2010.11.002}
  {\bibfield  {journal} {\bibinfo  {journal} {Ann. Phys.}\ }\textbf {\bibinfo
  {volume} {326}},\ \bibinfo {pages} {839} (\bibinfo {year}
  {2010}{\natexlab{b}})},\ \Eprint {https://arxiv.org/abs/1006.4871}
  {arXiv:1006.4871} \BibitemShut {NoStop}%
\bibitem [{\citenamefont {Michnicki}(2012)}]{michnicki2012quantum}%
  \BibitemOpen
  \bibfield  {author} {\bibinfo {author} {\bibfnamefont {K.}~\bibnamefont
  {Michnicki}},\ }\href {https://doi.org/10.48550/ARXIV.1208.3496} {\bibinfo
  {title} {3-d quantum stabilizer codes with a power law energy barrier}}
  (\bibinfo {year} {2012})\BibitemShut {NoStop}%
\bibitem [{\citenamefont {Bravyi}\ and\ \citenamefont
  {Haah}(2011)}]{PhysRevLett.107.150504}%
  \BibitemOpen
  \bibfield  {author} {\bibinfo {author} {\bibfnamefont {S.}~\bibnamefont
  {Bravyi}}\ and\ \bibinfo {author} {\bibfnamefont {J.}~\bibnamefont {Haah}},\
  }\bibfield  {title} {\bibinfo {title} {{Energy landscape of 3D spin
  hamiltonians with topological order}},\ }\href
  {https://doi.org/10.1103/PhysRevLett.107.150504} {\bibfield  {journal}
  {\bibinfo  {journal} {Physical Review Letters}\ }\textbf {\bibinfo {volume}
  {107}},\ \bibinfo {pages} {150504} (\bibinfo {year} {2011})},\ \Eprint
  {https://arxiv.org/abs/1105.4159} {arXiv:1105.4159} \BibitemShut {NoStop}%
\bibitem [{\citenamefont {Bravyi}\ and\ \citenamefont
  {Haah}(2013)}]{Bravyi2013}%
  \BibitemOpen
  \bibfield  {author} {\bibinfo {author} {\bibfnamefont {S.}~\bibnamefont
  {Bravyi}}\ and\ \bibinfo {author} {\bibfnamefont {J.}~\bibnamefont {Haah}},\
  }\bibfield  {title} {\bibinfo {title} {Quantum self-correction in the 3d
  cubic code model},\ }\bibfield  {journal} {\bibinfo  {journal} {Physical
  Review Letters}\ }\textbf {\bibinfo {volume} {111}},\ \href
  {https://doi.org/10.1103/physrevlett.111.200501}
  {10.1103/physrevlett.111.200501} (\bibinfo {year} {2013})\BibitemShut
  {NoStop}%
\bibitem [{\citenamefont {Williamson}\ and\ \citenamefont
  {Wang}(2017)}]{williamson2016hamiltonian}%
  \BibitemOpen
  \bibfield  {author} {\bibinfo {author} {\bibfnamefont {D.~J.}\ \bibnamefont
  {Williamson}}\ and\ \bibinfo {author} {\bibfnamefont {Z.}~\bibnamefont
  {Wang}},\ }\bibfield  {title} {\bibinfo {title} {{Hamiltonian models for
  topological phases of matter in three spatial dimensions}},\ }\href
  {https://doi.org/10.1016/j.aop.2016.12.018} {\bibfield  {journal} {\bibinfo
  {journal} {Annals of Physics}\ }\textbf {\bibinfo {volume} {377}},\ \bibinfo
  {pages} {311} (\bibinfo {year} {2017})},\ \Eprint
  {https://arxiv.org/abs/1606.07144} {arXiv:1606.07144} \BibitemShut {NoStop}%
\bibitem [{\citenamefont {Brown}\ and\ \citenamefont
  {Williamson}(2020)}]{Brown2019Parallel}%
  \BibitemOpen
  \bibfield  {author} {\bibinfo {author} {\bibfnamefont {B.~J.}\ \bibnamefont
  {Brown}}\ and\ \bibinfo {author} {\bibfnamefont {D.~J.}\ \bibnamefont
  {Williamson}},\ }\bibfield  {title} {\bibinfo {title} {{Parallelized quantum
  error correction with fracton topological codes}},\ }\href
  {https://doi.org/10.1103/physrevresearch.2.013303} {\bibfield  {journal}
  {\bibinfo  {journal} {Physical Review Research}\ }\textbf {\bibinfo {volume}
  {2}},\ \bibinfo {pages} {1} (\bibinfo {year} {2020})},\ \Eprint
  {https://arxiv.org/abs/1901.08061} {arXiv:1901.08061} \BibitemShut {NoStop}%
\bibitem [{\citenamefont {Weinstein}\ \emph {et~al.}(2019)\citenamefont
  {Weinstein}, \citenamefont {Ortiz},\ and\ \citenamefont
  {Nussinov}}]{Weinstein2019}%
  \BibitemOpen
  \bibfield  {author} {\bibinfo {author} {\bibfnamefont {Z.}~\bibnamefont
  {Weinstein}}, \bibinfo {author} {\bibfnamefont {G.}~\bibnamefont {Ortiz}},\
  and\ \bibinfo {author} {\bibfnamefont {Z.}~\bibnamefont {Nussinov}},\
  }\bibfield  {title} {\bibinfo {title} {{Universality Classes of Stabilizer
  Code Hamiltonians}},\ }\bibfield  {journal} {\bibinfo  {journal} {Physical
  Review Letters}\ }\textbf {\bibinfo {volume} {123}},\ \href
  {https://doi.org/10.1103/PhysRevLett.123.230503}
  {10.1103/PhysRevLett.123.230503} (\bibinfo {year} {2019})\BibitemShut
  {NoStop}%
\bibitem [{\citenamefont {Devakul}\ and\ \citenamefont
  {Williamson}(2021)}]{Devakul2020b}%
  \BibitemOpen
  \bibfield  {author} {\bibinfo {author} {\bibfnamefont {T.}~\bibnamefont
  {Devakul}}\ and\ \bibinfo {author} {\bibfnamefont {D.~J.}\ \bibnamefont
  {Williamson}},\ }\bibfield  {title} {\bibinfo {title} {{Fractalizing quantum
  codes}},\ }\bibfield  {journal} {\bibinfo  {journal} {Quantum}\ }\textbf
  {\bibinfo {volume} {5}},\ \href {https://doi.org/10.22331/q-2021-04-22-438}
  {10.22331/q-2021-04-22-438} (\bibinfo {year} {2021}),\ \Eprint
  {https://arxiv.org/abs/2009.01252} {arXiv:2009.01252} \BibitemShut {NoStop}%
\bibitem [{\citenamefont {Zhu}\ \emph {et~al.}(2022)\citenamefont {Zhu},
  \citenamefont {Jochym-O'Connor},\ and\ \citenamefont {Dua}}]{Zhu2022}%
  \BibitemOpen
  \bibfield  {author} {\bibinfo {author} {\bibfnamefont {G.}~\bibnamefont
  {Zhu}}, \bibinfo {author} {\bibfnamefont {T.}~\bibnamefont
  {Jochym-O'Connor}},\ and\ \bibinfo {author} {\bibfnamefont {A.}~\bibnamefont
  {Dua}},\ }\bibfield  {title} {\bibinfo {title} {{Topological Order, Quantum
  Codes, and Quantum Computation on Fractal Geometries}},\ }\bibfield
  {journal} {\bibinfo  {journal} {PRX Quantum}\ }\textbf {\bibinfo {volume}
  {3}},\ \href {https://doi.org/10.1103/PRXQuantum.3.030338}
  {10.1103/PRXQuantum.3.030338} (\bibinfo {year} {2022})\BibitemShut {NoStop}%
\bibitem [{\citenamefont {Aitchison}\ \emph {et~al.}(2023)\citenamefont
  {Aitchison}, \citenamefont {Bulmash}, \citenamefont {Dua}, \citenamefont
  {Doherty},\ and\ \citenamefont {Williamson}}]{aitchison2023no}%
  \BibitemOpen
  \bibfield  {author} {\bibinfo {author} {\bibfnamefont {C.~T.}\ \bibnamefont
  {Aitchison}}, \bibinfo {author} {\bibfnamefont {D.}~\bibnamefont {Bulmash}},
  \bibinfo {author} {\bibfnamefont {A.}~\bibnamefont {Dua}}, \bibinfo {author}
  {\bibfnamefont {A.~C.}\ \bibnamefont {Doherty}},\ and\ \bibinfo {author}
  {\bibfnamefont {D.~J.}\ \bibnamefont {Williamson}},\ }\bibfield  {title}
  {\bibinfo {title} {No strings attached: Boundaries and defects in the cubic
  code},\ }\href@noop {} {\bibfield  {journal} {\bibinfo  {journal} {arXiv
  preprint arXiv:2308.00138}\ } (\bibinfo {year} {2023})}\BibitemShut {NoStop}%
\bibitem [{\citenamefont {Dua}\ \emph {et~al.}(2023)\citenamefont {Dua},
  \citenamefont {Jochym-O'Connor},\ and\ \citenamefont {Zhu}}]{Dua2023Fractal}%
  \BibitemOpen
  \bibfield  {author} {\bibinfo {author} {\bibfnamefont {A.}~\bibnamefont
  {Dua}}, \bibinfo {author} {\bibfnamefont {T.}~\bibnamefont
  {Jochym-O'Connor}},\ and\ \bibinfo {author} {\bibfnamefont {G.}~\bibnamefont
  {Zhu}},\ }\bibfield  {title} {\bibinfo {title} {Quantum error correction with
  fractal topological codes},\ }\href
  {https://doi.org/10.22331/q-2023-09-26-1122} {\bibfield  {journal} {\bibinfo
  {journal} {Quantum}\ }\textbf {\bibinfo {volume} {7}},\ \bibinfo {pages}
  {1122} (\bibinfo {year} {2023})}\BibitemShut {NoStop}%
\bibitem [{\citenamefont {Haah}(2014)}]{haah2014bifurcation}%
  \BibitemOpen
  \bibfield  {author} {\bibinfo {author} {\bibfnamefont {J.}~\bibnamefont
  {Haah}},\ }\bibfield  {title} {\bibinfo {title} {{Bifurcation in entanglement
  renormalization group flow of a gapped spin model}},\ }\href
  {https://doi.org/10.1103/PhysRevB.89.075119} {\bibfield  {journal} {\bibinfo
  {journal} {Physical Review B - Condensed Matter and Materials Physics}\
  }\textbf {\bibinfo {volume} {89}},\ \bibinfo {pages} {75119} (\bibinfo {year}
  {2014})},\ \Eprint {https://arxiv.org/abs/1310.4507} {arXiv:1310.4507}
  \BibitemShut {NoStop}%
\bibitem [{\citenamefont {Vijay}\ and\ \citenamefont
  {Fu}(2017)}]{vijay2017generalization}%
  \BibitemOpen
  \bibfield  {author} {\bibinfo {author} {\bibfnamefont {S.}~\bibnamefont
  {Vijay}}\ and\ \bibinfo {author} {\bibfnamefont {L.}~\bibnamefont {Fu}},\
  }\bibfield  {title} {\bibinfo {title} {{A Generalization of Non-Abelian
  Anyons in Three Dimensions}},\ }\href {http://arxiv.org/abs/1706.07070}
  {\bibfield  {journal} {\bibinfo  {journal} {preprint}\ } (\bibinfo {year}
  {2017})},\ \Eprint {https://arxiv.org/abs/1706.07070} {arXiv:1706.07070}
  \BibitemShut {NoStop}%
\bibitem [{\citenamefont {Ma}\ \emph {et~al.}(2017)\citenamefont {Ma},
  \citenamefont {Lake}, \citenamefont {Chen},\ and\ \citenamefont
  {Hermele}}]{PhysRevB.95.245126}%
  \BibitemOpen
  \bibfield  {author} {\bibinfo {author} {\bibfnamefont {H.}~\bibnamefont
  {Ma}}, \bibinfo {author} {\bibfnamefont {E.}~\bibnamefont {Lake}}, \bibinfo
  {author} {\bibfnamefont {X.}~\bibnamefont {Chen}},\ and\ \bibinfo {author}
  {\bibfnamefont {M.}~\bibnamefont {Hermele}},\ }\bibfield  {title} {\bibinfo
  {title} {{Fracton topological order via coupled layers}},\ }\href
  {https://doi.org/10.1103/PhysRevB.95.245126} {\bibfield  {journal} {\bibinfo
  {journal} {Physical Review B}\ }\textbf {\bibinfo {volume} {95}},\ \bibinfo
  {pages} {245126} (\bibinfo {year} {2017})},\ \Eprint
  {https://arxiv.org/abs/1701.00747} {arXiv:1701.00747} \BibitemShut {NoStop}%
\bibitem [{\citenamefont {Dua}\ \emph {et~al.}(2019{\natexlab{a}})\citenamefont
  {Dua}, \citenamefont {Kim}, \citenamefont {Cheng},\ and\ \citenamefont
  {Williamson}}]{Dua2019}%
  \BibitemOpen
  \bibfield  {author} {\bibinfo {author} {\bibfnamefont {A.}~\bibnamefont
  {Dua}}, \bibinfo {author} {\bibfnamefont {I.~H.}\ \bibnamefont {Kim}},
  \bibinfo {author} {\bibfnamefont {M.}~\bibnamefont {Cheng}},\ and\ \bibinfo
  {author} {\bibfnamefont {D.~J.}\ \bibnamefont {Williamson}},\ }\bibfield
  {title} {\bibinfo {title} {{Sorting topological stabilizer models in three
  dimensions}},\ }\bibfield  {journal} {\bibinfo  {journal} {Physical Review
  B}\ }\textbf {\bibinfo {volume} {100}},\ \href
  {https://doi.org/10.1103/PhysRevB.100.155137} {10.1103/PhysRevB.100.155137}
  (\bibinfo {year} {2019}{\natexlab{a}}),\ \Eprint
  {https://arxiv.org/abs/1908.08049} {arXiv:1908.08049} \BibitemShut {NoStop}%
\bibitem [{\citenamefont {Dua}\ \emph {et~al.}(2019{\natexlab{b}})\citenamefont
  {Dua}, \citenamefont {Williamson}, \citenamefont {Haah},\ and\ \citenamefont
  {Cheng}}]{Dua2019c}%
  \BibitemOpen
  \bibfield  {author} {\bibinfo {author} {\bibfnamefont {A.}~\bibnamefont
  {Dua}}, \bibinfo {author} {\bibfnamefont {D.~J.}\ \bibnamefont {Williamson}},
  \bibinfo {author} {\bibfnamefont {J.}~\bibnamefont {Haah}},\ and\ \bibinfo
  {author} {\bibfnamefont {M.}~\bibnamefont {Cheng}},\ }\bibfield  {title}
  {\bibinfo {title} {{Compactifying fracton stabilizer models}},\ }\bibfield
  {journal} {\bibinfo  {journal} {Physical Review B}\ }\textbf {\bibinfo
  {volume} {99}},\ \href {https://doi.org/10.1103/physrevb.99.245135}
  {10.1103/physrevb.99.245135} (\bibinfo {year} {2019}{\natexlab{b}}),\ \Eprint
  {https://arxiv.org/abs/1903.12246} {arXiv:1903.12246} \BibitemShut {NoStop}%
\bibitem [{\citenamefont {Pai}\ and\ \citenamefont {Hermele}(2019)}]{Pai2019}%
  \BibitemOpen
  \bibfield  {author} {\bibinfo {author} {\bibfnamefont {S.}~\bibnamefont
  {Pai}}\ and\ \bibinfo {author} {\bibfnamefont {M.}~\bibnamefont {Hermele}},\
  }\bibfield  {title} {\bibinfo {title} {{Fracton fusion and statistics}},\
  }\bibfield  {journal} {\bibinfo  {journal} {Physical Review B}\ }\textbf
  {\bibinfo {volume} {100}},\ \href
  {https://doi.org/10.1103/PhysRevB.100.195136} {10.1103/PhysRevB.100.195136}
  (\bibinfo {year} {2019}),\ \Eprint {https://arxiv.org/abs/1903.11625}
  {arXiv:1903.11625} \BibitemShut {NoStop}%
\bibitem [{\citenamefont {Tantivasadakarn}\ and\ \citenamefont
  {Vijay}(2019)}]{Tantivasadakarn2019}%
  \BibitemOpen
  \bibfield  {author} {\bibinfo {author} {\bibfnamefont {N.}~\bibnamefont
  {Tantivasadakarn}}\ and\ \bibinfo {author} {\bibfnamefont {S.}~\bibnamefont
  {Vijay}},\ }\bibfield  {title} {\bibinfo {title} {{Searching for Fracton
  Orders via Symmetry Defect Condensation}},\ }\href
  {http://arxiv.org/abs/1912.02826} {\bibfield  {journal} {\bibinfo  {journal}
  {preprint}\ } (\bibinfo {year} {2019})},\ \Eprint
  {https://arxiv.org/abs/1912.02826} {arXiv:1912.02826} \BibitemShut {NoStop}%
\bibitem [{\citenamefont {Tantivasadakarn}\ \emph
  {et~al.}(2021{\natexlab{a}})\citenamefont {Tantivasadakarn}, \citenamefont
  {Ji},\ and\ \citenamefont {Vijay}}]{Tantivasadakarn2021}%
  \BibitemOpen
  \bibfield  {author} {\bibinfo {author} {\bibfnamefont {N.}~\bibnamefont
  {Tantivasadakarn}}, \bibinfo {author} {\bibfnamefont {W.}~\bibnamefont
  {Ji}},\ and\ \bibinfo {author} {\bibfnamefont {S.}~\bibnamefont {Vijay}},\
  }\bibfield  {title} {\bibinfo {title} {{Hybrid fracton phases: Parent orders
  for liquid and nonliquid quantum phases}},\ }\bibfield  {journal} {\bibinfo
  {journal} {Physical Review B}\ }\textbf {\bibinfo {volume} {103}},\ \href
  {https://doi.org/10.1103/PhysRevB.103.245136} {10.1103/PhysRevB.103.245136}
  (\bibinfo {year} {2021}{\natexlab{a}}),\ \Eprint
  {https://arxiv.org/abs/2102.09555} {arXiv:2102.09555} \BibitemShut {NoStop}%
\bibitem [{\citenamefont {Tantivasadakarn}\ \emph
  {et~al.}(2021{\natexlab{b}})\citenamefont {Tantivasadakarn}, \citenamefont
  {Ji},\ and\ \citenamefont {Vijay}}]{Tantivasadakarn2021Nonabelian}%
  \BibitemOpen
  \bibfield  {author} {\bibinfo {author} {\bibfnamefont {N.}~\bibnamefont
  {Tantivasadakarn}}, \bibinfo {author} {\bibfnamefont {W.}~\bibnamefont
  {Ji}},\ and\ \bibinfo {author} {\bibfnamefont {S.}~\bibnamefont {Vijay}},\
  }\bibfield  {title} {\bibinfo {title} {{Non-Abelian hybrid fracton orders}},\
  }\bibfield  {journal} {\bibinfo  {journal} {Physical Review B}\ }\textbf
  {\bibinfo {volume} {104}},\ \href
  {https://doi.org/10.1103/PhysRevB.104.115117} {10.1103/PhysRevB.104.115117}
  (\bibinfo {year} {2021}{\natexlab{b}}),\ \Eprint
  {https://arxiv.org/abs/2106.03842} {arXiv:2106.03842} \BibitemShut {NoStop}%
\bibitem [{\citenamefont {Song}\ \emph
  {et~al.}(2023{\natexlab{b}})\citenamefont {Song}, \citenamefont
  {Tantivasadakarn}, \citenamefont {Shirley},\ and\ \citenamefont
  {Hermele}}]{Song2023}%
  \BibitemOpen
  \bibfield  {author} {\bibinfo {author} {\bibfnamefont {H.}~\bibnamefont
  {Song}}, \bibinfo {author} {\bibfnamefont {N.}~\bibnamefont
  {Tantivasadakarn}}, \bibinfo {author} {\bibfnamefont {W.}~\bibnamefont
  {Shirley}},\ and\ \bibinfo {author} {\bibfnamefont {M.}~\bibnamefont
  {Hermele}},\ }\bibfield  {title} {\bibinfo {title} {{Fracton
  Self-Statistics}},\ }\href {http://arxiv.org/abs/2304.00028} {\bibfield
  {journal} {\bibinfo  {journal} {preprint}\ } (\bibinfo {year}
  {2023}{\natexlab{b}})},\ \Eprint {https://arxiv.org/abs/2304.00028}
  {arXiv:2304.00028} \BibitemShut {NoStop}%
\bibitem [{\citenamefont {Bulmash}\ and\ \citenamefont
  {Iadecola}(2019)}]{Bulmash2018}%
  \BibitemOpen
  \bibfield  {author} {\bibinfo {author} {\bibfnamefont {D.}~\bibnamefont
  {Bulmash}}\ and\ \bibinfo {author} {\bibfnamefont {T.}~\bibnamefont
  {Iadecola}},\ }\bibfield  {title} {\bibinfo {title} {{Braiding and gapped
  boundaries in fracton topological phases}},\ }\bibfield  {journal} {\bibinfo
  {journal} {Physical Review B}\ }\textbf {\bibinfo {volume} {99}},\ \href
  {https://doi.org/10.1103/PhysRevB.99.125132} {10.1103/PhysRevB.99.125132}
  (\bibinfo {year} {2019}),\ \Eprint {https://arxiv.org/abs/1810.00012}
  {arXiv:1810.00012} \BibitemShut {NoStop}%
\bibitem [{\citenamefont {Song}\ \emph {et~al.}(2019)\citenamefont {Song},
  \citenamefont {Prem}, \citenamefont {Huang},\ and\ \citenamefont
  {Martin-Delgado}}]{song2018twisted}%
  \BibitemOpen
  \bibfield  {author} {\bibinfo {author} {\bibfnamefont {H.}~\bibnamefont
  {Song}}, \bibinfo {author} {\bibfnamefont {A.}~\bibnamefont {Prem}}, \bibinfo
  {author} {\bibfnamefont {S.~J.}\ \bibnamefont {Huang}},\ and\ \bibinfo
  {author} {\bibfnamefont {M.~A.}\ \bibnamefont {Martin-Delgado}},\ }\bibfield
  {title} {\bibinfo {title} {{Twisted fracton models in three dimensions}},\
  }\bibfield  {journal} {\bibinfo  {journal} {Physical Review B}\ }\textbf
  {\bibinfo {volume} {99}},\ \href {https://doi.org/10.1103/PhysRevB.99.155118}
  {10.1103/PhysRevB.99.155118} (\bibinfo {year} {2019}),\ \Eprint
  {https://arxiv.org/abs/1805.06899} {arXiv:1805.06899} \BibitemShut {NoStop}%
\bibitem [{\citenamefont {Prem}\ and\ \citenamefont
  {Williamson}(2019)}]{Prem2019}%
  \BibitemOpen
  \bibfield  {author} {\bibinfo {author} {\bibfnamefont {A.}~\bibnamefont
  {Prem}}\ and\ \bibinfo {author} {\bibfnamefont {D.}~\bibnamefont
  {Williamson}},\ }\bibfield  {title} {\bibinfo {title} {{Gauging permutation
  symmetries as a route to non-Abelian fractons}},\ }\href
  {https://doi.org/10.21468/scipostphys.7.5.068} {\bibfield  {journal}
  {\bibinfo  {journal} {SciPost Physics}\ }\textbf {\bibinfo {volume} {7}},\
  \bibinfo {pages} {068} (\bibinfo {year} {2019})},\ \Eprint
  {https://arxiv.org/abs/1905.06309} {arXiv:1905.06309} \BibitemShut {NoStop}%
\bibitem [{\citenamefont {Prem}\ \emph {et~al.}(2019)\citenamefont {Prem},
  \citenamefont {Huang}, \citenamefont {Song},\ and\ \citenamefont
  {Hermele}}]{Prem2019Cage}%
  \BibitemOpen
  \bibfield  {author} {\bibinfo {author} {\bibfnamefont {A.}~\bibnamefont
  {Prem}}, \bibinfo {author} {\bibfnamefont {S.~J.}\ \bibnamefont {Huang}},
  \bibinfo {author} {\bibfnamefont {H.}~\bibnamefont {Song}},\ and\ \bibinfo
  {author} {\bibfnamefont {M.}~\bibnamefont {Hermele}},\ }\bibfield  {title}
  {\bibinfo {title} {{Cage-Net Fracton Models}},\ }\bibfield  {journal}
  {\bibinfo  {journal} {Physical Review X}\ }\textbf {\bibinfo {volume} {9}},\
  \href {https://doi.org/10.1103/PhysRevX.9.021010} {10.1103/PhysRevX.9.021010}
  (\bibinfo {year} {2019}),\ \Eprint {https://arxiv.org/abs/1806.04687}
  {arXiv:1806.04687} \BibitemShut {NoStop}%
\bibitem [{\citenamefont {Bulmash}\ and\ \citenamefont
  {Barkeshli}(2019)}]{Bulmash2019}%
  \BibitemOpen
  \bibfield  {author} {\bibinfo {author} {\bibfnamefont {D.}~\bibnamefont
  {Bulmash}}\ and\ \bibinfo {author} {\bibfnamefont {M.}~\bibnamefont
  {Barkeshli}},\ }\bibfield  {title} {\bibinfo {title} {{Gauging fractons:
  Immobile non-Abelian quasiparticles, fractals, and position-dependent
  degeneracies}},\ }\bibfield  {journal} {\bibinfo  {journal} {Physical Review
  B}\ }\textbf {\bibinfo {volume} {100}},\ \href
  {https://doi.org/10.1103/PhysRevB.100.155146} {10.1103/PhysRevB.100.155146}
  (\bibinfo {year} {2019}),\ \Eprint {https://arxiv.org/abs/1905.05771}
  {arXiv:1905.05771} \BibitemShut {NoStop}%
\bibitem [{\citenamefont {Dua}\ \emph {et~al.}(2019{\natexlab{c}})\citenamefont
  {Dua}, \citenamefont {Sarkar}, \citenamefont {Williamson},\ and\
  \citenamefont {Cheng}}]{Dua2019b}%
  \BibitemOpen
  \bibfield  {author} {\bibinfo {author} {\bibfnamefont {A.}~\bibnamefont
  {Dua}}, \bibinfo {author} {\bibfnamefont {P.}~\bibnamefont {Sarkar}},
  \bibinfo {author} {\bibfnamefont {D.~J.}\ \bibnamefont {Williamson}},\ and\
  \bibinfo {author} {\bibfnamefont {M.}~\bibnamefont {Cheng}},\ }\bibfield
  {title} {\bibinfo {title} {{Bifurcating entanglement-renormalization group
  flows of fracton stabilizer models}},\ }\href
  {https://doi.org/10.1103/PhysRevResearch.2.033021} {\bibfield  {journal}
  {\bibinfo  {journal} {PHYSICAL REVIEW RESEARCH}\ }\textbf {\bibinfo {volume}
  {2}},\ \bibinfo {pages} {33021} (\bibinfo {year} {2019}{\natexlab{c}})},\
  \Eprint {https://arxiv.org/abs/1909.12304} {arXiv:1909.12304} \BibitemShut
  {NoStop}%
\bibitem [{\citenamefont {Williamson}\ and\ \citenamefont
  {Cheng}(2023)}]{Williamson2020}%
  \BibitemOpen
  \bibfield  {author} {\bibinfo {author} {\bibfnamefont {D.~J.}\ \bibnamefont
  {Williamson}}\ and\ \bibinfo {author} {\bibfnamefont {M.}~\bibnamefont
  {Cheng}},\ }\bibfield  {title} {\bibinfo {title} {{Designer non-Abelian
  fractons from topological layers}},\ }\bibfield  {journal} {\bibinfo
  {journal} {Physical Review B}\ }\textbf {\bibinfo {volume} {107}},\ \href
  {https://doi.org/10.1103/PhysRevB.107.035103} {10.1103/PhysRevB.107.035103}
  (\bibinfo {year} {2023}),\ \Eprint {https://arxiv.org/abs/2004.07251}
  {arXiv:2004.07251} \BibitemShut {NoStop}%
\bibitem [{\citenamefont {Sullivan}\ \emph {et~al.}(2021)\citenamefont
  {Sullivan}, \citenamefont {Iadecola},\ and\ \citenamefont
  {Williamson}}]{Sullivan2021}%
  \BibitemOpen
  \bibfield  {author} {\bibinfo {author} {\bibfnamefont {J.}~\bibnamefont
  {Sullivan}}, \bibinfo {author} {\bibfnamefont {T.}~\bibnamefont {Iadecola}},\
  and\ \bibinfo {author} {\bibfnamefont {D.~J.}\ \bibnamefont {Williamson}},\
  }\bibfield  {title} {\bibinfo {title} {{Planar p-string condensation: Chiral
  fracton phases from fractional quantum Hall layers and beyond}},\ }\bibfield
  {journal} {\bibinfo  {journal} {Physical Review B}\ }\textbf {\bibinfo
  {volume} {103}},\ \href {https://doi.org/10.1103/PhysRevB.103.205301}
  {10.1103/PhysRevB.103.205301} (\bibinfo {year} {2021}),\ \Eprint
  {https://arxiv.org/abs/2010.15127} {arXiv:2010.15127} \BibitemShut {NoStop}%
\bibitem [{\citenamefont {Williamson}()}]{Williamson2023talk}%
  \BibitemOpen
  \bibfield  {author} {\bibinfo {author} {\bibfnamefont {D.~J.}\ \bibnamefont
  {Williamson}},\ }\href@noop {} {\bibinfo {title} {Saturating tradeoff bounds
  with topological defect networks}},\ \bibinfo {note} {talk at YITP workshop
  ``Quantum Error Correction''
  \url{https://www.yukawa.kyoto-u.ac.jp/seminar/s53088} (2023)}\BibitemShut
  {NoStop}%
\bibitem [{\citenamefont {Portnoy}(2023)}]{portnoy2023local}%
  \BibitemOpen
  \bibfield  {author} {\bibinfo {author} {\bibfnamefont {E.}~\bibnamefont
  {Portnoy}},\ }\href {https://doi.org/10.48550/ARXIV.2303.06755} {\bibinfo
  {title} {Local quantum codes from subdivided manifolds}} (\bibinfo {year}
  {2023})\BibitemShut {NoStop}%
\bibitem [{\citenamefont {Freedman}\ and\ \citenamefont
  {Hastings}(2021)}]{Freedman2021}%
  \BibitemOpen
  \bibfield  {author} {\bibinfo {author} {\bibfnamefont {M.}~\bibnamefont
  {Freedman}}\ and\ \bibinfo {author} {\bibfnamefont {M.}~\bibnamefont
  {Hastings}},\ }\bibfield  {title} {\bibinfo {title} {{Building manifolds from
  quantum codes}},\ }\bibfield  {journal} {\bibinfo  {journal} {Geometric and
  Functional Analysis}\ }\textbf {\bibinfo {volume} {31}},\ \href
  {https://doi.org/10.1007/s00039-021-00567-3} {10.1007/s00039-021-00567-3}
  (\bibinfo {year} {2021})\BibitemShut {NoStop}%
\bibitem [{\citenamefont {Gromov}\ and\ \citenamefont
  {Guth}(2012)}]{Gromov2012}%
  \BibitemOpen
  \bibfield  {author} {\bibinfo {author} {\bibfnamefont {M.}~\bibnamefont
  {Gromov}}\ and\ \bibinfo {author} {\bibfnamefont {L.}~\bibnamefont {Guth}},\
  }\bibfield  {title} {\bibinfo {title} {{Generalizations of the
  Kolmogorov-Barzdin embedding estimates}},\ }\bibfield  {journal} {\bibinfo
  {journal} {Duke Mathematical Journal}\ }\textbf {\bibinfo {volume} {161}},\
  \href {https://doi.org/10.1215/00127094-1812840} {10.1215/00127094-1812840}
  (\bibinfo {year} {2012})\BibitemShut {NoStop}%
\bibitem [{\citenamefont {Lin}\ \emph {et~al.}(2023)\citenamefont {Lin},
  \citenamefont {Wills},\ and\ \citenamefont {Hsieh}}]{Lin2023}%
  \BibitemOpen
  \bibfield  {author} {\bibinfo {author} {\bibfnamefont {T.-C.}\ \bibnamefont
  {Lin}}, \bibinfo {author} {\bibfnamefont {A.}~\bibnamefont {Wills}},\ and\
  \bibinfo {author} {\bibfnamefont {M.-H.}\ \bibnamefont {Hsieh}},\ }\bibfield
  {title} {\bibinfo {title} {Geometrically local quantum and classical codes
  from subdivision},\ }\href {https://arxiv.org/abs/2309.16104v1} {\bibfield
  {journal} {\bibinfo  {journal} {arXiv:2309.16104}\ } (\bibinfo {year}
  {2023})}\BibitemShut {NoStop}%
\bibitem [{\citenamefont {Alicki}\ \emph {et~al.}(2010)\citenamefont {Alicki},
  \citenamefont {Horodecki}, \citenamefont {Horodecki},\ and\ \citenamefont
  {Horodecki}}]{alicki2010thermal}%
  \BibitemOpen
  \bibfield  {author} {\bibinfo {author} {\bibfnamefont {R.}~\bibnamefont
  {Alicki}}, \bibinfo {author} {\bibfnamefont {M.}~\bibnamefont {Horodecki}},
  \bibinfo {author} {\bibfnamefont {P.}~\bibnamefont {Horodecki}},\ and\
  \bibinfo {author} {\bibfnamefont {R.}~\bibnamefont {Horodecki}},\ }\bibfield
  {title} {\bibinfo {title} {On thermal stability of topological qubit in
  kitaev's 4d model},\ }\href {https://doi.org/10.1142/S1230161210000023}
  {\bibfield  {journal} {\bibinfo  {journal} {Open Systems \& Information
  Dynamics}\ }\textbf {\bibinfo {volume} {17}},\ \bibinfo {pages} {1} (\bibinfo
  {year} {2010})},\ \Eprint
  {https://arxiv.org/abs/https://doi.org/10.1142/S1230161210000023}
  {https://doi.org/10.1142/S1230161210000023} \BibitemShut {NoStop}%
\bibitem [{\citenamefont {Bombin}\ \emph {et~al.}(2013)\citenamefont {Bombin},
  \citenamefont {Chhajlany}, \citenamefont {Horodecki},\ and\ \citenamefont
  {Martin-Delgado}}]{Bombin2013Self}%
  \BibitemOpen
  \bibfield  {author} {\bibinfo {author} {\bibfnamefont {H.}~\bibnamefont
  {Bombin}}, \bibinfo {author} {\bibfnamefont {R.~W.}\ \bibnamefont
  {Chhajlany}}, \bibinfo {author} {\bibfnamefont {M.}~\bibnamefont
  {Horodecki}},\ and\ \bibinfo {author} {\bibfnamefont {M.~A.}\ \bibnamefont
  {Martin-Delgado}},\ }\bibfield  {title} {\bibinfo {title} {Self-correcting
  quantum computers},\ }\href {https://doi.org/10.1088/1367-2630/15/5/055023}
  {\bibfield  {journal} {\bibinfo  {journal} {New Journal of Physics}\ }\textbf
  {\bibinfo {volume} {15}},\ \bibinfo {pages} {055023} (\bibinfo {year}
  {2013})}\BibitemShut {NoStop}%
\bibitem [{\citenamefont {Roberts}\ and\ \citenamefont
  {Bartlett}(2020)}]{roberts2020symmetry}%
  \BibitemOpen
  \bibfield  {author} {\bibinfo {author} {\bibfnamefont {S.}~\bibnamefont
  {Roberts}}\ and\ \bibinfo {author} {\bibfnamefont {S.~D.}\ \bibnamefont
  {Bartlett}},\ }\bibfield  {title} {\bibinfo {title} {Symmetry-protected
  self-correcting quantum memories},\ }\bibfield  {journal} {\bibinfo
  {journal} {Physical Review X}\ }\textbf {\bibinfo {volume} {10}},\ \href
  {https://doi.org/10.1103/physrevx.10.031041} {10.1103/physrevx.10.031041}
  (\bibinfo {year} {2020})\BibitemShut {NoStop}%
\bibitem [{\citenamefont {Hastings}(2021)}]{hastings2021quantum}%
  \BibitemOpen
  \bibfield  {author} {\bibinfo {author} {\bibfnamefont {M.~B.}\ \bibnamefont
  {Hastings}},\ }\href {https://doi.org/10.48550/ARXIV.2102.10030} {\bibinfo
  {title} {On quantum weight reduction}} (\bibinfo {year} {2021})\BibitemShut
  {NoStop}%
\bibitem [{\citenamefont {Sabo}\ \emph {et~al.}(2024)\citenamefont {Sabo},
  \citenamefont {Gunderman}, \citenamefont {Ide}, \citenamefont {Vasmer},\ and\
  \citenamefont {Dauphinais}}]{Sabo2024}%
  \BibitemOpen
  \bibfield  {author} {\bibinfo {author} {\bibfnamefont {E.}~\bibnamefont
  {Sabo}}, \bibinfo {author} {\bibfnamefont {L.~G.}\ \bibnamefont {Gunderman}},
  \bibinfo {author} {\bibfnamefont {B.}~\bibnamefont {Ide}}, \bibinfo {author}
  {\bibfnamefont {M.}~\bibnamefont {Vasmer}},\ and\ \bibinfo {author}
  {\bibfnamefont {G.}~\bibnamefont {Dauphinais}},\ }\bibfield  {title}
  {\bibinfo {title} {{Weight Reduced Stabilizer Codes with Lower Overhead}},\
  }\href {https://arxiv.org/abs/2402.05228v1} {\bibfield  {journal} {\bibinfo
  {journal} {arXiv:2402.05228}\ } (\bibinfo {year} {2024})},\ \Eprint
  {https://arxiv.org/abs/2402.05228} {arXiv:2402.05228} \BibitemShut {NoStop}%
\bibitem [{\citenamefont {Kubica}\ \emph {et~al.}(2015)\citenamefont {Kubica},
  \citenamefont {Yoshida},\ and\ \citenamefont
  {Pastawski}}]{kubica2015unfolding}%
  \BibitemOpen
  \bibfield  {author} {\bibinfo {author} {\bibfnamefont {A.}~\bibnamefont
  {Kubica}}, \bibinfo {author} {\bibfnamefont {B.}~\bibnamefont {Yoshida}},\
  and\ \bibinfo {author} {\bibfnamefont {F.}~\bibnamefont {Pastawski}},\
  }\bibfield  {title} {\bibinfo {title} {{Unfolding the color code}},\ }\href
  {https://doi.org/10.1088/1367-2630/17/8/083026} {\bibfield  {journal}
  {\bibinfo  {journal} {New Journal of Physics}\ }\textbf {\bibinfo {volume}
  {17}},\ \bibinfo {pages} {83026} (\bibinfo {year} {2015})},\ \Eprint
  {https://arxiv.org/abs/1503.02065} {arXiv:1503.02065} \BibitemShut {NoStop}%
\bibitem [{\citenamefont {Bombin}\ and\ \citenamefont
  {Martin-Delgado}(2007{\natexlab{b}})}]{bombin2007topological}%
  \BibitemOpen
  \bibfield  {author} {\bibinfo {author} {\bibfnamefont {H.}~\bibnamefont
  {Bombin}}\ and\ \bibinfo {author} {\bibfnamefont {M.~A.}\ \bibnamefont
  {Martin-Delgado}},\ }\bibfield  {title} {\bibinfo {title} {Topological
  computation without braiding},\ }\href
  {https://doi.org/10.1103/PhysRevLett.98.160502} {\bibfield  {journal}
  {\bibinfo  {journal} {Phys. Rev. Lett.}\ }\textbf {\bibinfo {volume} {98}},\
  \bibinfo {pages} {160502} (\bibinfo {year} {2007}{\natexlab{b}})}\BibitemShut
  {NoStop}%
\bibitem [{\citenamefont {Bravyi}\ and\ \citenamefont
  {K{\"o}nig}(2013)}]{bravyi2013classification}%
  \BibitemOpen
  \bibfield  {author} {\bibinfo {author} {\bibfnamefont {S.}~\bibnamefont
  {Bravyi}}\ and\ \bibinfo {author} {\bibfnamefont {R.}~\bibnamefont
  {K{\"o}nig}},\ }\bibfield  {title} {\bibinfo {title} {Classification of
  topologically protected gates for local stabilizer codes},\ }\href@noop {}
  {\bibfield  {journal} {\bibinfo  {journal} {Physical {R}eview {L}etters}\
  }\textbf {\bibinfo {volume} {110}},\ \bibinfo {pages} {170503} (\bibinfo
  {year} {2013})}\BibitemShut {NoStop}%
\bibitem [{\citenamefont {Haah}(2013)}]{haah2013commuting}%
  \BibitemOpen
  \bibfield  {author} {\bibinfo {author} {\bibfnamefont {J.}~\bibnamefont
  {Haah}},\ }\bibfield  {title} {\bibinfo {title} {Commuting pauli hamiltonians
  as maps between free modules},\ }\href
  {https://doi.org/10.1007/s00220-013-1810-2} {\bibfield  {journal} {\bibinfo
  {journal} {Communications in Mathematical Physics}\ }\textbf {\bibinfo
  {volume} {324}},\ \bibinfo {pages} {351–399} (\bibinfo {year}
  {2013})}\BibitemShut {NoStop}%
\bibitem [{\citenamefont {Temme}(2016)}]{temme2016thermalization}%
  \BibitemOpen
  \bibfield  {author} {\bibinfo {author} {\bibfnamefont {K.}~\bibnamefont
  {Temme}},\ }\bibfield  {title} {\bibinfo {title} {Thermalization time bounds
  for pauli stabilizer hamiltonians},\ }\href
  {https://doi.org/10.1007/s00220-016-2746-0} {\bibfield  {journal} {\bibinfo
  {journal} {Communications in Mathematical Physics}\ }\textbf {\bibinfo
  {volume} {350}},\ \bibinfo {pages} {603–637} (\bibinfo {year}
  {2016})}\BibitemShut {NoStop}%
\bibitem [{\citenamefont {Steane}(1996{\natexlab{c}})}]{steane1996multiple}%
  \BibitemOpen
  \bibfield  {author} {\bibinfo {author} {\bibfnamefont {A.}~\bibnamefont
  {Steane}},\ }\bibfield  {title} {\bibinfo {title} {Multiple-particle
  interference and quantum error correction},\ }\href@noop {} {\bibfield
  {journal} {\bibinfo  {journal} {Proceedings of the Royal Society A}\ }\textbf
  {\bibinfo {volume} {452}},\ \bibinfo {pages} {2551} (\bibinfo {year}
  {1996}{\natexlab{c}})}\BibitemShut {NoStop}%
\bibitem [{\citenamefont {Bravyi}\ \emph
  {et~al.}(2010{\natexlab{c}})\citenamefont {Bravyi}, \citenamefont
  {Hastings},\ and\ \citenamefont {Michalakis}}]{Bravyi2010Stability}%
  \BibitemOpen
  \bibfield  {author} {\bibinfo {author} {\bibfnamefont {S.}~\bibnamefont
  {Bravyi}}, \bibinfo {author} {\bibfnamefont {M.~B.}\ \bibnamefont
  {Hastings}},\ and\ \bibinfo {author} {\bibfnamefont {S.}~\bibnamefont
  {Michalakis}},\ }\bibfield  {title} {\bibinfo {title} {{Topological quantum
  order: Stability under local perturbations}},\ }\bibfield  {journal}
  {\bibinfo  {journal} {Journal of Mathematical Physics}\ }\textbf {\bibinfo
  {volume} {51}},\ \href {https://doi.org/10.1063/1.3490195}
  {10.1063/1.3490195} (\bibinfo {year} {2010}{\natexlab{c}})\BibitemShut
  {NoStop}%
\bibitem [{\citenamefont {Knill}\ and\ \citenamefont
  {Laflamme}(1997)}]{knill1997theory}%
  \BibitemOpen
  \bibfield  {author} {\bibinfo {author} {\bibfnamefont {E.}~\bibnamefont
  {Knill}}\ and\ \bibinfo {author} {\bibfnamefont {R.}~\bibnamefont
  {Laflamme}},\ }\bibfield  {title} {\bibinfo {title} {Theory of quantum
  error-correcting codes},\ }\href {https://doi.org/10.1103/PhysRevA.55.900}
  {\bibfield  {journal} {\bibinfo  {journal} {Phys. Rev. A}\ }\textbf {\bibinfo
  {volume} {55}},\ \bibinfo {pages} {900} (\bibinfo {year} {1997})}\BibitemShut
  {NoStop}%
\bibitem [{\citenamefont {Levin}\ and\ \citenamefont {Wen}(2003)}]{Levin2003}%
  \BibitemOpen
  \bibfield  {author} {\bibinfo {author} {\bibfnamefont {M.}~\bibnamefont
  {Levin}}\ and\ \bibinfo {author} {\bibfnamefont {X.~G.}\ \bibnamefont
  {Wen}},\ }\bibfield  {title} {\bibinfo {title} {{Fermions, strings, and gauge
  fields in lattice spin models}},\ }\bibfield  {journal} {\bibinfo  {journal}
  {Physical Review B - Condensed Matter and Materials Physics}\ }\textbf
  {\bibinfo {volume} {67}},\ \href {https://doi.org/10.1103/PhysRevB.67.245316}
  {10.1103/PhysRevB.67.245316} (\bibinfo {year} {2003})\BibitemShut {NoStop}%
\bibitem [{\citenamefont {Ellison}\ \emph
  {et~al.}(2022{\natexlab{a}})\citenamefont {Ellison}, \citenamefont {Chen},
  \citenamefont {Dua}, \citenamefont {Shirley}, \citenamefont
  {Tantivasadakarn},\ and\ \citenamefont {Williamson}}]{Ellison2022}%
  \BibitemOpen
  \bibfield  {author} {\bibinfo {author} {\bibfnamefont {T.~D.}\ \bibnamefont
  {Ellison}}, \bibinfo {author} {\bibfnamefont {Y.~A.}\ \bibnamefont {Chen}},
  \bibinfo {author} {\bibfnamefont {A.}~\bibnamefont {Dua}}, \bibinfo {author}
  {\bibfnamefont {W.}~\bibnamefont {Shirley}}, \bibinfo {author} {\bibfnamefont
  {N.}~\bibnamefont {Tantivasadakarn}},\ and\ \bibinfo {author} {\bibfnamefont
  {D.~J.}\ \bibnamefont {Williamson}},\ }\bibfield  {title} {\bibinfo {title}
  {{Pauli Stabilizer Models of Twisted Quantum Doubles}},\ }\bibfield
  {journal} {\bibinfo  {journal} {PRX Quantum}\ }\textbf {\bibinfo {volume}
  {3}},\ \href {https://doi.org/10.1103/PRXQuantum.3.010353}
  {10.1103/PRXQuantum.3.010353} (\bibinfo {year}
  {2022}{\natexlab{a}})\BibitemShut {NoStop}%
\bibitem [{\citenamefont {Moore}\ and\ \citenamefont
  {Seiberg}(1989)}]{Moore1989}%
  \BibitemOpen
  \bibfield  {author} {\bibinfo {author} {\bibfnamefont {G.}~\bibnamefont
  {Moore}}\ and\ \bibinfo {author} {\bibfnamefont {N.}~\bibnamefont
  {Seiberg}},\ }\bibfield  {title} {\bibinfo {title} {{Classical and quantum
  conformal field theory}},\ }\href {https://doi.org/10.1007/BF01238857}
  {\bibfield  {journal} {\bibinfo  {journal} {Communications in Mathematical
  Physics}\ }\textbf {\bibinfo {volume} {123}},\ \bibinfo {pages} {177}
  (\bibinfo {year} {1989})}\BibitemShut {NoStop}%
\bibitem [{\citenamefont {Bais}\ and\ \citenamefont
  {Slingerland}(2009)}]{Bais2009Condensate}%
  \BibitemOpen
  \bibfield  {author} {\bibinfo {author} {\bibfnamefont {F.~A.}\ \bibnamefont
  {Bais}}\ and\ \bibinfo {author} {\bibfnamefont {J.~K.}\ \bibnamefont
  {Slingerland}},\ }\bibfield  {title} {\bibinfo {title} {{Condensate-induced
  transitions between topologically ordered phases}},\ }\bibfield  {journal}
  {\bibinfo  {journal} {Physical Review B - Condensed Matter and Materials
  Physics}\ }\textbf {\bibinfo {volume} {79}},\ \href
  {https://doi.org/10.1103/PhysRevB.79.045316} {10.1103/PhysRevB.79.045316}
  (\bibinfo {year} {2009}),\ \Eprint {https://arxiv.org/abs/0808.0627}
  {arXiv:0808.0627} \BibitemShut {NoStop}%
\bibitem [{\citenamefont {Kong}(2014)}]{Kong2014Anyon}%
  \BibitemOpen
  \bibfield  {author} {\bibinfo {author} {\bibfnamefont {L.}~\bibnamefont
  {Kong}},\ }\bibfield  {title} {\bibinfo {title} {{Anyon condensation and
  tensor categories}},\ }\href
  {https://doi.org/10.1016/j.nuclphysb.2014.07.003} {\bibfield  {journal}
  {\bibinfo  {journal} {Nuclear Physics B}\ }\textbf {\bibinfo {volume}
  {886}},\ \bibinfo {pages} {436} (\bibinfo {year} {2014})},\ \Eprint
  {https://arxiv.org/abs/1307.8244} {arXiv:1307.8244} \BibitemShut {NoStop}%
\bibitem [{\citenamefont {Levin}(2013)}]{Levin2013Protected}%
  \BibitemOpen
  \bibfield  {author} {\bibinfo {author} {\bibfnamefont {M.}~\bibnamefont
  {Levin}},\ }\bibfield  {title} {\bibinfo {title} {{Protected Edge Modes
  without Symmetry}},\ }\bibfield  {journal} {\bibinfo  {journal} {Physical
  Review X}\ }\textbf {\bibinfo {volume} {3}},\ \href
  {https://doi.org/10.1103/PhysRevX.3.021009} {10.1103/PhysRevX.3.021009}
  (\bibinfo {year} {2013}),\ \Eprint {https://arxiv.org/abs/1301.7355}
  {arXiv:1301.7355} \BibitemShut {NoStop}%
\bibitem [{\citenamefont {Barkeshli}\ \emph {et~al.}(2013)\citenamefont
  {Barkeshli}, \citenamefont {Jian},\ and\ \citenamefont
  {Qi}}]{Barkeshli2013Classification}%
  \BibitemOpen
  \bibfield  {author} {\bibinfo {author} {\bibfnamefont {M.}~\bibnamefont
  {Barkeshli}}, \bibinfo {author} {\bibfnamefont {C.~M.}\ \bibnamefont
  {Jian}},\ and\ \bibinfo {author} {\bibfnamefont {X.~L.}\ \bibnamefont {Qi}},\
  }\bibfield  {title} {\bibinfo {title} {Classification of topological defects
  in abelian topological states},\ }\href
  {https://doi.org/10.1103/PHYSREVB.88.241103/FIGURES/3/MEDIUM} {\bibfield
  {journal} {\bibinfo  {journal} {Physical Review B - Condensed Matter and
  Materials Physics}\ }\textbf {\bibinfo {volume} {88}},\ \bibinfo {pages}
  {241103} (\bibinfo {year} {2013})}\BibitemShut {NoStop}%
\bibitem [{\citenamefont {Beigi}\ \emph {et~al.}(2011)\citenamefont {Beigi},
  \citenamefont {Shor},\ and\ \citenamefont {Whalen}}]{beigi2011quantum}%
  \BibitemOpen
  \bibfield  {author} {\bibinfo {author} {\bibfnamefont {S.}~\bibnamefont
  {Beigi}}, \bibinfo {author} {\bibfnamefont {P.~W.}\ \bibnamefont {Shor}},\
  and\ \bibinfo {author} {\bibfnamefont {D.}~\bibnamefont {Whalen}},\
  }\bibfield  {title} {\bibinfo {title} {{The Quantum Double Model with
  Boundary: Condensations and Symmetries}},\ }\href
  {https://doi.org/10.1007/s00220-011-1294-x} {\bibfield  {journal} {\bibinfo
  {journal} {Communications in Mathematical Physics}\ }\textbf {\bibinfo
  {volume} {306}},\ \bibinfo {pages} {663} (\bibinfo {year} {2011})},\ \Eprint
  {https://arxiv.org/abs/1006.5479} {arXiv:1006.5479} \BibitemShut {NoStop}%
\bibitem [{\citenamefont {Kitaev}\ and\ \citenamefont
  {Kong}(2012)}]{Kitaev2012Models}%
  \BibitemOpen
  \bibfield  {author} {\bibinfo {author} {\bibfnamefont {A.}~\bibnamefont
  {Kitaev}}\ and\ \bibinfo {author} {\bibfnamefont {L.}~\bibnamefont {Kong}},\
  }\bibfield  {title} {\bibinfo {title} {{Models for Gapped Boundaries and
  Domain Walls}},\ }\href {https://doi.org/10.1007/s00220-012-1500-5}
  {\bibfield  {journal} {\bibinfo  {journal} {Communications in Mathematical
  Physics}\ }\textbf {\bibinfo {volume} {313}},\ \bibinfo {pages} {351}
  (\bibinfo {year} {2012})},\ \Eprint {https://arxiv.org/abs/1104.5047}
  {arXiv:1104.5047} \BibitemShut {NoStop}%
\bibitem [{\citenamefont {Aasen}\ \emph {et~al.}(2019)\citenamefont {Aasen},
  \citenamefont {Lake},\ and\ \citenamefont {Walker}}]{aasen2017fermion}%
  \BibitemOpen
  \bibfield  {author} {\bibinfo {author} {\bibfnamefont {D.}~\bibnamefont
  {Aasen}}, \bibinfo {author} {\bibfnamefont {E.}~\bibnamefont {Lake}},\ and\
  \bibinfo {author} {\bibfnamefont {K.}~\bibnamefont {Walker}},\ }\bibfield
  {title} {\bibinfo {title} {{Fermion condensation and super pivotal
  categories}},\ }\bibfield  {journal} {\bibinfo  {journal} {Journal of
  Mathematical Physics}\ }\textbf {\bibinfo {volume} {60}},\ \href
  {https://doi.org/10.1063/1.5045669} {10.1063/1.5045669} (\bibinfo {year}
  {2019}),\ \Eprint {https://arxiv.org/abs/1709.01941} {arXiv:1709.01941}
  \BibitemShut {NoStop}%
\bibitem [{\citenamefont {Bombin}(2010)}]{bombin2010topological}%
  \BibitemOpen
  \bibfield  {author} {\bibinfo {author} {\bibfnamefont {H.}~\bibnamefont
  {Bombin}},\ }\bibfield  {title} {\bibinfo {title} {Topological order with a
  twist: Ising anyons from an abelian model},\ }\href@noop {} {\bibfield
  {journal} {\bibinfo  {journal} {Physical Review Letters}\ }\textbf {\bibinfo
  {volume} {105}},\ \bibinfo {pages} {030403} (\bibinfo {year}
  {2010})}\BibitemShut {NoStop}%
\bibitem [{\citenamefont {Haah}(2016)}]{haahInvariant}%
  \BibitemOpen
  \bibfield  {author} {\bibinfo {author} {\bibfnamefont {J.}~\bibnamefont
  {Haah}},\ }\bibfield  {title} {\bibinfo {title} {{An Invariant of
  Topologically Ordered States Under Local Unitary Transformations}},\ }\href
  {https://doi.org/10.1007/s00220-016-2594-y} {\bibfield  {journal} {\bibinfo
  {journal} {Communications in Mathematical Physics}\ }\textbf {\bibinfo
  {volume} {342}},\ \bibinfo {pages} {771} (\bibinfo {year} {2016})},\ \Eprint
  {https://arxiv.org/abs/1407.2926} {arXiv:1407.2926} \BibitemShut {NoStop}%
\bibitem [{\citenamefont {Siva}\ and\ \citenamefont
  {Yoshida}(2017)}]{siva2017topological}%
  \BibitemOpen
  \bibfield  {author} {\bibinfo {author} {\bibfnamefont {K.}~\bibnamefont
  {Siva}}\ and\ \bibinfo {author} {\bibfnamefont {B.}~\bibnamefont {Yoshida}},\
  }\bibfield  {title} {\bibinfo {title} {Topological order and memory time in
  marginally-self-correcting quantum memory},\ }\bibfield  {journal} {\bibinfo
  {journal} {Physical Review A}\ }\textbf {\bibinfo {volume} {95}},\ \href
  {https://doi.org/10.1103/physreva.95.032324} {10.1103/physreva.95.032324}
  (\bibinfo {year} {2017})\BibitemShut {NoStop}%
\bibitem [{\citenamefont {Horsman}\ \emph {et~al.}(2012)\citenamefont
  {Horsman}, \citenamefont {Fowler}, \citenamefont {Devitt},\ and\
  \citenamefont {Van~Meter}}]{horsman2012surface}%
  \BibitemOpen
  \bibfield  {author} {\bibinfo {author} {\bibfnamefont {C.}~\bibnamefont
  {Horsman}}, \bibinfo {author} {\bibfnamefont {A.~G.}\ \bibnamefont {Fowler}},
  \bibinfo {author} {\bibfnamefont {S.}~\bibnamefont {Devitt}},\ and\ \bibinfo
  {author} {\bibfnamefont {R.}~\bibnamefont {Van~Meter}},\ }\bibfield  {title}
  {\bibinfo {title} {Surface code quantum computing by lattice surgery},\
  }\href@noop {} {\bibfield  {journal} {\bibinfo  {journal} {New Journal of
  Physics}\ }\textbf {\bibinfo {volume} {14}},\ \bibinfo {pages} {123011}
  (\bibinfo {year} {2012})}\BibitemShut {NoStop}%
\bibitem [{\citenamefont {Witten}(1988)}]{Witten1988TQFT}%
  \BibitemOpen
  \bibfield  {author} {\bibinfo {author} {\bibfnamefont {E.}~\bibnamefont
  {Witten}},\ }\bibfield  {title} {\bibinfo {title} {{Topological quantum field
  theory}},\ }\href {https://doi.org/10.1007/BF01223371} {\bibfield  {journal}
  {\bibinfo  {journal} {Communications in Mathematical Physics 1988 117:3}\
  }\textbf {\bibinfo {volume} {117}},\ \bibinfo {pages} {353} (\bibinfo {year}
  {1988})}\BibitemShut {NoStop}%
\bibitem [{\citenamefont {Atiyah}(1988)}]{atiyah1988topological}%
  \BibitemOpen
  \bibfield  {author} {\bibinfo {author} {\bibfnamefont {M.}~\bibnamefont
  {Atiyah}},\ }\bibfield  {title} {\bibinfo {title} {{Topological quantum field
  theories}},\ }\href {https://doi.org/10.1007/BF02698547} {\bibfield
  {journal} {\bibinfo  {journal} {Publications Math{\'{e}}matiques de
  l'Institut des Hautes Scientifiques}\ }\textbf {\bibinfo {volume} {68}},\
  \bibinfo {pages} {175} (\bibinfo {year} {1988})},\ \Eprint
  {https://arxiv.org/abs/0011260} {arXiv:0011260 [hep-th]} \BibitemShut
  {NoStop}%
\bibitem [{\citenamefont {Ellison}\ \emph
  {et~al.}(2022{\natexlab{b}})\citenamefont {Ellison}, \citenamefont {Chen},
  \citenamefont {Dua}, \citenamefont {Shirley}, \citenamefont
  {Tantivasadakarn},\ and\ \citenamefont {Williamson}}]{ellison2022pauli}%
  \BibitemOpen
  \bibfield  {author} {\bibinfo {author} {\bibfnamefont {T.~D.}\ \bibnamefont
  {Ellison}}, \bibinfo {author} {\bibfnamefont {Y.-A.}\ \bibnamefont {Chen}},
  \bibinfo {author} {\bibfnamefont {A.}~\bibnamefont {Dua}}, \bibinfo {author}
  {\bibfnamefont {W.}~\bibnamefont {Shirley}}, \bibinfo {author} {\bibfnamefont
  {N.}~\bibnamefont {Tantivasadakarn}},\ and\ \bibinfo {author} {\bibfnamefont
  {D.~J.}\ \bibnamefont {Williamson}},\ }\bibfield  {title} {\bibinfo {title}
  {Pauli topological subsystem codes from abelian anyon theories},\ }\href@noop
  {} {\bibfield  {journal} {\bibinfo  {journal} {arXiv preprint
  arXiv:2211.03798}\ } (\bibinfo {year} {2022}{\natexlab{b}})}\BibitemShut
  {NoStop}%
\bibitem [{\citenamefont {Kesselring}\ \emph {et~al.}(2022)\citenamefont
  {Kesselring}, \citenamefont {de~la Fuente}, \citenamefont {Thomsen},
  \citenamefont {Eisert}, \citenamefont {Bartlett},\ and\ \citenamefont
  {Brown}}]{kesselring2022anyon}%
  \BibitemOpen
  \bibfield  {author} {\bibinfo {author} {\bibfnamefont {M.~S.}\ \bibnamefont
  {Kesselring}}, \bibinfo {author} {\bibfnamefont {J.~C.~M.}\ \bibnamefont
  {de~la Fuente}}, \bibinfo {author} {\bibfnamefont {F.}~\bibnamefont
  {Thomsen}}, \bibinfo {author} {\bibfnamefont {J.}~\bibnamefont {Eisert}},
  \bibinfo {author} {\bibfnamefont {S.~D.}\ \bibnamefont {Bartlett}},\ and\
  \bibinfo {author} {\bibfnamefont {B.~J.}\ \bibnamefont {Brown}},\ }\bibfield
  {title} {\bibinfo {title} {Anyon condensation and the color code},\
  }\href@noop {} {\bibfield  {journal} {\bibinfo  {journal} {arXiv preprint
  arXiv:2212.00042}\ } (\bibinfo {year} {2022})}\BibitemShut {NoStop}%
\bibitem [{\citenamefont {Vuillot}\ \emph {et~al.}(2019)\citenamefont
  {Vuillot}, \citenamefont {Lao}, \citenamefont {Criger}, \citenamefont
  {García~Almudéver}, \citenamefont {Bertels},\ and\ \citenamefont
  {Terhal}}]{vuillot2019code}%
  \BibitemOpen
  \bibfield  {author} {\bibinfo {author} {\bibfnamefont {C.}~\bibnamefont
  {Vuillot}}, \bibinfo {author} {\bibfnamefont {L.}~\bibnamefont {Lao}},
  \bibinfo {author} {\bibfnamefont {B.}~\bibnamefont {Criger}}, \bibinfo
  {author} {\bibfnamefont {C.}~\bibnamefont {García~Almudéver}}, \bibinfo
  {author} {\bibfnamefont {K.}~\bibnamefont {Bertels}},\ and\ \bibinfo {author}
  {\bibfnamefont {B.~M.}\ \bibnamefont {Terhal}},\ }\bibfield  {title}
  {\bibinfo {title} {Code deformation and lattice surgery are gauge fixing},\
  }\href {https://doi.org/10.1088/1367-2630/ab0199} {\bibfield  {journal}
  {\bibinfo  {journal} {New Journal of Physics}\ }\textbf {\bibinfo {volume}
  {21}},\ \bibinfo {pages} {033028} (\bibinfo {year} {2019})}\BibitemShut
  {NoStop}%
\bibitem [{\citenamefont {Brown}(2023)}]{Brown2023Conservation}%
  \BibitemOpen
  \bibfield  {author} {\bibinfo {author} {\bibfnamefont {B.~J.}\ \bibnamefont
  {Brown}},\ }\bibfield  {title} {\bibinfo {title} {{Conservation Laws and
  Quantum Error Correction: Towards a Generalised Matching Decoder}},\
  }\bibfield  {journal} {\bibinfo  {journal} {IEEE BITS the Information Theory
  Magazine}\ }\href {https://doi.org/10.1109/mbits.2023.3246025}
  {10.1109/mbits.2023.3246025} (\bibinfo {year} {2023})\BibitemShut {NoStop}%
\bibitem [{\citenamefont {Anshu}\ \emph {et~al.}(2023)\citenamefont {Anshu},
  \citenamefont {Breuckmann},\ and\ \citenamefont {Nirkhe}}]{anshu2023nlts}%
  \BibitemOpen
  \bibfield  {author} {\bibinfo {author} {\bibfnamefont {A.}~\bibnamefont
  {Anshu}}, \bibinfo {author} {\bibfnamefont {N.~P.}\ \bibnamefont
  {Breuckmann}},\ and\ \bibinfo {author} {\bibfnamefont {C.}~\bibnamefont
  {Nirkhe}},\ }\bibfield  {title} {\bibinfo {title} {Nlts hamiltonians from
  good quantum codes},\ }in\ \href {https://doi.org/10.1145/3564246.3585114}
  {\emph {\bibinfo {booktitle} {Proceedings of the 55th Annual ACM Symposium on
  Theory of Computing}}},\ \bibinfo {series and number} {STOC ’23}\ (\bibinfo
   {publisher} {ACM},\ \bibinfo {year} {2023})\BibitemShut {NoStop}%
\bibitem [{\citenamefont {Flammia}\ \emph {et~al.}(2017)\citenamefont
  {Flammia}, \citenamefont {Haah}, \citenamefont {Kastoryano},\ and\
  \citenamefont {Kim}}]{flammia2017limits}%
  \BibitemOpen
  \bibfield  {author} {\bibinfo {author} {\bibfnamefont {S.~T.}\ \bibnamefont
  {Flammia}}, \bibinfo {author} {\bibfnamefont {J.}~\bibnamefont {Haah}},
  \bibinfo {author} {\bibfnamefont {M.~J.}\ \bibnamefont {Kastoryano}},\ and\
  \bibinfo {author} {\bibfnamefont {I.~H.}\ \bibnamefont {Kim}},\ }\bibfield
  {title} {\bibinfo {title} {Limits on the storage of quantum information in a
  volume of space},\ }\href@noop {} {\bibfield  {journal} {\bibinfo  {journal}
  {Quantum}\ }\textbf {\bibinfo {volume} {1}},\ \bibinfo {pages} {4} (\bibinfo
  {year} {2017})}\BibitemShut {NoStop}%
\bibitem [{\citenamefont {Delfosse}\ and\ \citenamefont
  {Z{\'e}mor}(2013)}]{delfosse2013upper}%
  \BibitemOpen
  \bibfield  {author} {\bibinfo {author} {\bibfnamefont {N.}~\bibnamefont
  {Delfosse}}\ and\ \bibinfo {author} {\bibfnamefont {G.}~\bibnamefont
  {Z{\'e}mor}},\ }\bibfield  {title} {\bibinfo {title} {Upper bounds on the
  rate of low density stabilizer codes for the quantum erasure channel},\
  }\href@noop {} {\bibfield  {journal} {\bibinfo  {journal} {Quantum
  Information \& Computation}\ }\textbf {\bibinfo {volume} {13}},\ \bibinfo
  {pages} {793} (\bibinfo {year} {2013})}\BibitemShut {NoStop}%
\end{thebibliography}%

\end{document}